\documentclass[11pt,reqno]{amsart}

\usepackage[utf8]{inputenc}
\usepackage[english]{babel}
\usepackage[foot]{amsaddr}
\usepackage[margin=1in]{geometry}
\usepackage{amsmath, amsthm, amsfonts, amssymb, tikz-cd, mathtools, mathrsfs,theoremref, txfonts}
\usepackage[T1]{fontenc}
\usepackage{graphicx}
\usepackage{hyperref}

\theoremstyle{definition}
\newtheorem{theorem}{Theorem}[section]

\newtheorem{lemma}[theorem]{Lemma}

\newtheorem{claim}[theorem]{Claim}
\theoremstyle{definition}
\newtheorem{definition}[theorem]{definition}

\usepackage{tcolorbox}
\newtcolorbox{mybox}{colback=red!5!white,colframe=red!75!black}
\newtcolorbox{alebox}{colback=blue!5!white,colframe=blue!75!black}

\newcommand{\C}{\mathbb{C}}
\newcommand{\E}{\mathbb{E}}
\newcommand{\F}{\mathbb{F}}

\newcommand{\Pb}{\mathbb{P}}
\newcommand{\R}{\mathbb{R}}
\newcommand{\Z}{\mathbb{Z}}
\newcommand{\ep}{\epsilon}
\newcommand{\fB}{\mathfrak{B}}

\newcommand{\fs}{f_{\text{struc}}}
\newcommand{\twopartial}{\abs{\widehat{\partial_{a, b}f}(\phi(a,b))}}

\DeclareMathOperator{\codim}{codim}
\DeclareMathOperator{\poly}{poly}
\DeclareMathOperator{\qpoly}{quasi-poly}
\DeclareMathOperator{\proj}{proj}
\DeclareMathOperator{\Spec}{Spec}
\DeclareMathOperator{\Mat}{Mat}
\DeclareMathOperator{\dist}{dist}

\DeclareMathOperator{\spn}{span}
\DeclareMathOperator{\arank}{arank}
\DeclareMathOperator{\rank}{rank}
\DeclareMathOperator{\Poly}{Poly}
\DeclareMathOperator{\CSM}{CSM}

\DeclareMathOperator{\Span}{Span}

\newcommand{\norm}[1]{\left\lVert#1\right\rVert}
\newcommand{\abs}[1]{\left\lvert#1\right\rvert}

\newcommand{\algtopskip}{5pt}
\newcommand{\algbotskip}{5pt}
\newcommand*{\Scale}[2][4]{\scalebox{#1}{$#2$}}

\usepackage{stmaryrd, txfonts, comment}
\renewcommand{\lozenge}{\talloblong}

\usepackage{abstract}

\title{Cubic Goldreich-Levin}
\author{Dain Kim, Anqi Li \and Jonathan Tidor}
\address{Massachusetts Institute of Technology, Cambridge, MA 02139, USA}
\thanks{Tidor was supported by NSF Graduate Research Fellowship Program DGE-1745302.}
\email{\{dain0327,anqili,jtidor\}@mit.edu}

\date{}

\begin{document}

\maketitle

\renewcommand{\abstractname}{}    
\renewcommand{\absnamepos}{empty}

\begin{abstract}  In this paper, we give a cubic Goldreich-Levin algorithm which makes polynomially-many queries to a function $f \colon \mathbb F_p^n \to \mathbb C$ and produces a decomposition of $f$ as a sum of cubic phases and a small error term.
This is a natural higher-order generalization of the classical Goldreich-Levin algorithm. The classical (linear) Goldreich-Levin algorithm has wide-ranging applications in learning theory, coding theory and the construction of pseudorandom generators in cryptography, as well as being closely related to Fourier analysis.
Higher-order Goldreich-Levin algorithms on the other hand involve central problems in higher-order Fourier analysis, namely the inverse theory of the Gowers $U^k$ norms, which are well-studied in additive combinatorics.
The only known result in this direction prior to this work is the quadratic Goldreich-Levin theorem, proved by Tulsiani and Wolf in 2011. The main step of their result involves an algorithmic version of the $U^3$ inverse theorem.
More complications appear in the inverse theory of the $U^4$ and higher norms. Our cubic Goldreich-Levin algorithm is based on algorithmizing recent work by Gowers and Mili\'cevi\'c who proved new quantitative bounds for the $U^4$ inverse theorem.

Our cubic Goldreich-Levin algorithm is constructed from two main tools: an algorithmic $U^4$ inverse theorem and an arithmetic decomposition result in the style of the Frieze-Kannan graph regularity lemma. As one application of our main theorem we solve the problem of self-correction for cubic Reed-Muller codes beyond the list decoding radius. Additionally we give a purely combinatorial result: an improvement of the quantitative bounds on the $U^4$ inverse theorem.
\end{abstract}

\section{Introduction}

Classical Fourier analysis has played an important role in computer science over the past decades. One foundational application is in property testing -- the field that studies the design and analysis of extremely efficient algorithms which determine whether an input is ``close'' to or ``far'' from a certain property. In particular, one of the first uses of classical Fourier analysis in property testing was to prove the correctness of the Blum-Luby-Rubinfeld (BLR) algorithm \cite{BLR93}, which, with a constant number of queries, detects if a function $f\colon \F_p^n \to \F_p$ is close to linear. A further problem in the same vein is as follows. Given a function $f\colon \F_p^n\to \F_p$ that is close to linear, efficiently identify one (or all) of the linear functions that $f$ is close to. The solution to this problem is well-known; one application of the celebrated Goldreich-Levin algorithm \cite{GL89} is to solve this problem. Beyond this context, the Goldreich-Levin algorithm also has wide-ranging applications in many areas of theoretical computer science, including in learning theory \cite{KM93}, coding theory \cite{AGS03}, and the construction of pseudorandom generators in cryptography \cite{HILL99}, the latter being the context that first motivated its study.

We can also describe the Goldreich-Levin algorithm through the lens of coding theory. Suppose we are given a function $f$ that is close to a linear function. A linear function can be interpreted as a Walsh-Hadamard codeword, so query access to $f$ corresponds to query access to a corrupted codeword. Thus the Goldreich-Levin algorithm also solves the list decoding problem for the Walsh-Hadamard code. In the context of coding theory, Reed-Muller codes are a generalization of Walsh-Hadamard codes from the linear setting to the setting of higher degree polynomials. Many natural problems that arise in this setting require higher-order generalizations of techniques from Fourier analysis.

Higher-order Fourier analysis is an extension of classical Fourier analysis to higher-order characters which are polynomial phase functions instead of linear phase functions. This theory was first developed by Gowers to give a new proof of Szemer\'edi's theorem in additive combinatorics \cite{G01} but recently has found many applications in theoretical computer science, especially in the field of property testing. We refer the reader to the papers \cite{BFL12, BFHHL13, TZ20} as well as the book \cite{HHL19} for further discussion on the recent applications of higher order Fourier analysis in property testing.

We describe one central application.

Extending the BLR linearity test, a natural problem to study is the property testing of polynomiality. Given a function $f\colon \F_p^n\to \F_p$, we wish to detect whether $f$ is close to a polynomial of degree at most $d$. There are two natural regimes for studying this problem -- the ``99\% regime'' where the goal is to detect if $f$ agrees with a polynomial on a $(1-\epsilon)$-fraction of the domain and the ``1\% regime'' where the goal is to detect if $f$ agrees with a polynomial on a $(1/p+\epsilon)$-fraction of the domain. The AKKLR test of Alon, Kaufman, Krivelevich, Litsyn, and Ron solves this problem in the 99\% regime, by sampling $f$ at $(d+1)$-dimensional parallelepipeds \cite{AKKLR05}.

The same test has potential to solve the problem in the 1\% regime, but the analysis is much more difficult. Indeed, proving the correctness of this test in the 1\% regime is essentially equivalent to understanding the inverse theory of the Gowers $U^{d+1}$-norm, the central problem of higher-order Fourier analysis. Work of Bergelson, Tao, and Ziegler \cite{BTZ10, TZ10, TZ12} resolves this problem for finite field vector spaces, proving the correctness of this 1\% test in the high-characteristic regime $p\geq d$. In the low-characteristic regime $p<d$, this test is known to fail (see \cite{GT09} and independently \cite{LMS11} for the $p=2$, $d=3$ case, and \cite{BSST21} for the full range $p<d$). This problem is still open in the low-characteristic regime; see, e.g., \cite[Conjecture 18.2]{HHL19} for some discussion of this problem.

While the testing problem is now fairly well understood, the problem of finding a polynomial that the input correlates with is still wide open. This is the problem that we focus on in this paper, which can be viewed as a higher-order generalization of the Goldreich-Levin algorithm. Another perspective on this problem is as follows: given a function $f\colon\F_p^n\to\F_p$ that agrees with a polynomial of degree at most $d$ on a $(1/p+\epsilon)$-fraction of the domain, we wish to give an efficient algorithm for finding one possible such polynomial. A third perspective on our problem is as an algorithmic $U^{d+1}$-inverse theorem. We first formally state this version of the problem, and then discuss its relation to these other versions.

\subsection*{Algorithmic $U^k$ inverse theorems}

The Gowers $U^k$ norm for a function $f\colon \F_p^n \to \C$ is defined by \[\norm{f}_{U^k}^{2^k} = \E_{x, h_1, \ldots, h_k \in \F_p^n} \partial_{h_1} \partial_{h_2} \cdots \partial_{h_k} f(x)\] where $\partial_{h} f(x) = f(x+h) \overline{f(x)}$ is the discrete multiplicative derivative. Since a degree $k-1$ polynomial vanishes upon taking $k$ successive discrete additive derivatives, it follows that $f(x) = \omega^{p(x)}$, where $p(x)$ is a degree $k-1$ polynomial and $\omega=e^{2\pi i/p}$ satisfies $\norm{f(x)}_{U^k} = 1$. The inverse problem for the Gowers $U^k$ norm asks for a partial converse to this statement. Namely, if $f\colon \F_p^n \to \C$ where $\norm{f}_\infty \le 1$ is a function for which $\norm{f}_{U^k} \geq \delta$ then does there exist a degree $k-1$ polynomial $p(x)$ such that $f$ has non-negligible correlation with $\omega^p$? This theorem is known to be true \cite{BTZ10, TZ10, TZ12}, though in the low-characteristic regime $p\leq k-2$ we must replace polynomials with a generalization known as \emph{non-classical polynomials} which we will formally define later in the paper.

Though the work of Bergelson, Tao, and Ziegler resolves the $U^k$-inverse theorem over $\F_p^n$ for all $p,k$, the techniques they use come from ergodic theory and thus give no quantitative bounds. The problem of proving a quantitative inverse theorem has been heavily-studied in the recent years and is important in many computer science applications such as those on communication complexity \cite{VW07} and pseudorandom generators which fool low-degree polynomials \cite{BV10}, since the existence of efficient algorithms often relies on good quantitative bounds from these inverse theorems. 

Good bounds for the $U^3$-inverse theorem have been known for some time; Green and Tao resolve the problem for $p>2$ \cite{GT08E} and Samorodnitsky for $p=2$ \cite{Sam07}. In contrast, quantitative bounds for the $U^4$ and higher inverse theorems is quite a difficult problem.

In 2017, Gowers and Mili\'cevi\'c \cite{GM17} gave the first quantitative bounds for the $U^4$ inverse theorem for $p\geq 5$. These bounds are approximately double exponential. Further work by Gowers and Mili\'cevi\'c \cite{GM20} gave quantitative bounds for the $U^k$ inverse theorem in the high-characteristic regime $p\geq k$. Finally, a recent work by the third author gives quantitative bounds for the $U^4$ inverse theorem in the low-characteristic regime $p=2,3$ \cite{T21}. 

The $U^2$ inverse theorem follows immediately from classical Fourier analysis while the Goldreich-Levin algorithm gives an algorithmic $U^2$ inverse theorem. The algorithmic $U^3$ inverse theorem was proved by Tulsiani and Wolf \cite{TW11} in 2011 as the main ingredient in the quadratic Goldreich-Levin theorem. In this paper we study the algorithmic $U^4$ inverse theorem. We first state the result and then we discuss its applications and the techniques we use to prove it.

\begin{theorem}[algorithmic $U^4$ inverse theorem]
\label{thm:alg-u4-inverse}
\label{thm:main}
Given a prime $p$ and $\delta, \epsilon>0$, set $\eta^{-1}=\exp\qpoly(\epsilon^{-1})$. For a 1-bounded function $f\colon\mathbb F_p^n\to\mathbb C$ that satisfies $\|f\|_{U^4}\geq\epsilon$, there is an algorithm that makes \\
$O(\poly(n, \eta^{-1}, \log(\delta^{-1})))$ queries to $f$ and, with probability at least $1-\delta$, outputs a cubic polynomial \\ $P\colon\mathbb F_p^n\to\mathbb F_p$ such that 
\begin{equation*}
    |\E_x f(x)\omega^{-P(x)}|>\eta.
\end{equation*}
For $p=2$, the cubic polynomial is non-classical.
\end{theorem}

\subsection*{Self-correction for Reed-Muller codes}

As a corollary of our algorithmic $U^4$ inverse theorem we give a self-correction algorithm for cubic Reed-Muller codes which works well beyond the list decoding radius. Given a function $f\colon \F_p^n\to\F_p$ that is at distance at most $1-\tfrac1p-\epsilon$ from a codeword, our algorithm produces a codeword which is at distance at most $1-\tfrac1p-\eta$. Note that $1-\tfrac1p$ is significantly beyond the list decoding radius (for $p=2$ this is $1/2$ while the list decoding radius is $1/8$). In this regime list decoding is impossible, yet we can still solve the corresponding self-correction problem. Furthermore $1-\tfrac1p-\epsilon$ is best possible, since every string is at distance $1-\tfrac1p$ from a codeword.

\begin{theorem}[self-correction of cubic Reed-Muller codes]
\label{thm:decoding}
For $f: \F_p^n \to \F_p$ such that there exists a cubic polynomial $P$ satisfying $\dist(f,P) \leq 1 - 1/p - \ep$, there is an algorithm that makes $O(\poly(n, \eta^{-1},\log(\delta^{-1})))$ queries to $f$ and, with probability at least $1-\delta$, outputs a cubic polynomial $Q$ such that $\dist(f,Q) \leq 1 - 1/p - \eta$. Here $\dist(\cdot,\cdot)$ is the normalized Hamming distance and $\eta^{-1} = \exp \qpoly(\ep)^{-1}$.
\end{theorem}

The self-correction problem for quadratic Reed-Muller codes was solved by Tulsiani and Wolf in their aforementioned paper on the quadratic Goldreich-Levin theorem \cite{TW11}.

For cubic Reed-Muller codes, the self-correction problem was studied by Hatami and Tulsiani in the case $p=2$. They gave a self-correction algorithm that works out to the radius $1/2-\sqrt{1/8}-\epsilon$ \cite{HT18}. In contrast our Theorem~\ref{thm:decoding} works for all $p$ and out to the optimal radius of $1-1/p-\epsilon$.

\subsection*{Algorithmic decomposition results}

The classical Goldreich-Levin algorithm produces a decomposition of a function $f\colon\F_p^n\to\C$ as the sum of a structured function and a pseudorandom function. Specifically, the structured function is the sum of a bounded number of linear phase functions and the pseudorandom function is Fourier uniform.

Tulsiani and Wolf's quadratic Goldreich-Levin algorithm produces a similar decomposition except that the structured function is the sum of a bounded number of \emph{quadratic} phase functions and the pseudorandom function is uniform in the $U^3$ norm. Also, for technical reasons Tulsiani and Wolf's decomposition includes a third error term which is small in the $\ell^1$ norm.

Our cubic Goldreich-Levin algorithm is an analogous result with cubic phase functions and the $U^4$ norm. In addition, we give a different general decomposition result as compared to that in Tulsiani and Wolf's paper. In our result we remove the $\ell^1$ error term at the cost of having slightly more cubic phase functions in our decomposition.

\begin{theorem}[cubic Goldreich-Levin]\label{thm:cub-GL}
Let $\ep, \delta > 0$. Then there exists $\eta^{-1} = \exp(\qpoly(\epsilon^{-1}))$ and a randomized algorithm, given any 1-bounded function $f \colon \F_p^n \to \C$, outputs with probability at least $1 - 2\delta/\eta^2$ a decomposition 
\[ f= c_1 {q_1} + \cdots + c_r {q_r} + g \]
where the $c_i$ are constants, the $q_i$  are cubic phase functions, such that $r \leq \exp(\eta^{-2})$, and $\norm{g}_{U^4} \leq \ep$. The algorithm makes at most $r$ calls to the algorithm in Theorem~\ref{thm:alg-u4-inverse}. For $p=2$, the cubic phase functions are non-classical.
\end{theorem}

\subsection*{Quantitative bounds}

The final result we prove in this paper is purely additive combinatorial. We give refine an aspect of Gowers and Mili\'cevi\'c's proof of the $U^4$ inverse theorem to get a quantitatively better dependence of $\eta(\epsilon)$ in terms of $\epsilon$ that is approximately exponential instead of double exponential. This is the reason that the quantitative dependence in the previous three theorems is single exponential instead of double exponential. We recently learned that through personal communicatio that this result was independently proved by Shachar Lovett.

\begin{theorem}\label{thm:exp-drop}
Given a prime $p$ and $\epsilon>0$, there is a constant $\eta^{-1} = O( \exp (\qpoly(\epsilon^{-1},p))$ with the following property: for every 1-bounded function $f\colon \F_p^n \to \C$ with $\norm{f}_{U^4} \ge \epsilon$, there is a cubic polynomial $P\colon \F_p^n \to \F_p$  such that $\abs{\E_x f(x)\omega^{-P(x)}} > \eta$. For $p=2$, the cubic polynomial is non-classical.
\end{theorem}

\subsection{Comparison to previous work}

The only other higher-order Goldreich-Levin algorithm known is a quadratic Goldreich-Levin algorithm proved by Tulsiani and Wolf \cite{TW11} in 2011. This algorithm comes in two parts, first an algorithmic $U^3$ inverse theorem, and second a general decomposition result.

By combining our algorithmic $U^4$ inverse theorem and their decomposition result we could produce a cubic Goldreich-Levin algorithm slightly different than the form given in Theorem~\ref{thm:cub-GL}. In particular, our decomposition result differs from theirs in the following way. Tulsiani and Wolf give an algorithm for decomposing a function $f$ as a sum $f=f_{str}+f_{psr}+f_{sml}$ where $f_{str}$ is a structured function, $f_{psr}$ is a pseudorandom function, and $f_{sml}$ is small in $L^1$ norm. Our general decomposition result shows that the $f_{sml}$ term is not necessary but as a trade-off the bound on the length of the structured part is worse.\footnote{For those familiar with the terminology, a decomposition of the form $f=f_{str}+f_{psr}+f_{sml}$ is necessary for ``strong regularity lemmas.'' Both of these general decomposition results as well as the Frieze-Kannan graph decomposition result are ``weak regularity lemmas'' which are the only type of decomposition result which can be computed efficiently. Weak regularity lemmas generally do not require the $f_{sml}$ term in the decomposition.}

While in the quadratic Goldreich-Levin case, the main tool that was used was an algorithmic Balog-Szemer\'edi-Gowers Theorem, in order to establish the cubic Goldreich-Levin we will need a wider assortment of tools. We will give a more detailed discussion for why this is the case in the following section, but the main gist is that the quadratic Goldreich-Levin problem reduces to finding an affine function that overlaps greatly with a function defined via a suitable large Fourier spectrum. In the cubic case, however, we will need to find a bi-affine rather than affine function which agrees on a significant fraction of inputs. All of the standard additive combinatorics tools are in the univariate setting, and to handle the bivariate case we will need to piece together various theorems from the standard toolbox in intricate ways. 

We believe that our techniques will help to prove higher-order Goldreich-Levin algorithms of all orders. However, we have not done so here for two reasons. The first is simply that the technical details greatly increase between the cubic and higher-order cases. The second reason is more serious, which is that quantitative bounds are not known for the $U^k$-inverse theorem in low characteristic for $k\geq 5$. We believe it would be very difficult to prove the correctness of an efficient algorithm without also giving a proof of a quantitative inverse theorem. Thus while the additive combinatorics technology may currently be sufficient to prove a higher-order Goldreich-Levin theorem in high characteristic, the case most applicable in theoretic computer science, $p=2$, seems out of reach of the current techniques.

\medskip
\noindent\textbf{Outline.} In Section 2, we give an outline of the algorithm and discuss some of the difficulties one faces in generalizing from the quadratic setting to the cubic one. In Section 3, we prove an algorithmic decomposition result which, combined with our algorithmic $U^4$ inverse theorem, gives the cubic Goldreich-Levin algorithm. In Section 4, we collect some algorithmic primitives that we will be frequently using in our algorithmic $U^4$ inverse theorem, including a slightly generalized form of Goldreich-Levin as well as algorithmic versions of some additive combinatorics theorems. We first modularize the algorithmic $U^4$ inverse theorem and give self-contained proofs for each individual piece in Section 5, which we then combine together in subsection 5.1 to give a complete proof of our main theorem. In Section 6, we discuss how to improve the quantitative bounds of Gowers and Mili\'cevi\'c to remove an exponential in the bound for $\eta$. For these sections of the paper, we work with $\F_p^n$ for $p \geq 5$. The case of low characteristics $p = 2,3$ is more complex due to technical reasons, and in Section 7 we show how to modify our algorithm to handle these subtleties and extend our results to these low characteristic cases. Finally, in Section 8 we give an application of our algorithmic $U^4$ inverse theorem to completely resolve the question of list decoding of cubic Reed-Muller codes beyond the list decoding radius.

\medskip

\noindent \textbf{Acknowledgements.} Part of this research was conducted while Kim and Li were participants and Tidor was a mentor in the 2021 Summer Program in Undergraduate Research+ (SPUR+) of the MIT Mathematics Department. We would like to express our gratitude towards David Jerison and Ankur Moitra for organizing this program and their helpful advice. Furthermore, we would like to thank Yufei Zhao for many insightful discussions and also for his constant support and encouragement. 

\section{Ideas of proof}

In this section, we give an outline of the tools which we use in the proof and explain some of the reasons why the cubic Goldreich-Levin theorem is significantly more difficult than the quadratic Goldreich-Levin. We give a high-level overview of the flow of the algorithm and describe some of the algorithmic sampling strategies used in several key steps.

\subsection{Overview of quadratic Goldreich-Levin}
We begin by providing an overview of the proof of quadratic Goldreich-Levin by Tulsiani and Wolf. Here we recall the statement of their result.
\begin{theorem}[{\cite[Theorem 1.2]{TW11}}]
Given $\ep, \delta > 0$, there exists $\eta = \exp(-1/\ep^C)$ and a randomized algorithm \texttt{find-quadratic} running in time $O(n^4\log n \cdot \poly(1/\ep, 1/\eta, \log(1/\delta))$ which, given query access to $f\colon  \mathbb{F}_2^n \rightarrow \{-1, 1\}$, either outputs a quadratic form $q$ or $\perp$. The algorithm has the following guarantee:
\begin{itemize}
    \item If $\norm{f}_{U^3} \geq \ep$ then with probability at least $1 - \delta$ it finds a quadratic form $q$ such that $\langle f, \omega^{q} \rangle \geq \eta$. 
    \item The probability that the algorithm outputs a quadratic form $q$ with $\langle f, \omega^q \rangle \leq \eta/2$ is at most $\delta$. 
\end{itemize}
\end{theorem}

In order to provide some intuition for the first step of their proof, consider the simplest case when $\|f\|_{U^3}=1$. Recalling the definition \[\norm{f}_{U^3}^{8} = \E_{x, h_1, h_2,h_3} \partial_{h_1} \partial_{h_2} \partial_{h_3} f(x)\] where $\partial_{h} f(x) = f(x+h) \overline{f(x)}$, this implies that $f$ is a quadratic phase function, i.e., $f(x) = (-1)^{q(x)}$ where $q(x) = x^{T}Mx$ is a quadratic form for some $M \in \Mat_n(\F_2)$. Now the discrete multiplicative derivative satisfies
\begin{align*}
    \partial_h (-1)^{q(x)} &= (-1)^{(x+h)^TM(x+h)} (-1)^{-x^TMx} \\
    &= \underbrace{(-1)^{h^TMh}}_{\text{constant}} (-1)^{\langle x, (M+M^T)h \rangle}.
\end{align*}
In other words $\partial_h (-1)^{q(x)}$, viewed as a function of $x$, has precisely one large Fourier coefficient which occurs at $(M+M^{\intercal})h$.

Working instead with the weaker assumption that $\|f\|_{U^3}\geq\epsilon$, a similar phenomenon occurs --  $\partial_h f$ typically has few large Fourier coefficients whose location encode the arithmetic structure of $f$. 

In more detail we define the \emph{$\gamma$-large Fourier spectrum of $f\colon \F_p^n \to \C$} by
\[ \Spec_{\gamma}(f) := \{ r \in \F_p^n: |\widehat{f}(r)| \geq \gamma\}.\] The classical Goldreich-Levin algorithm lets us compute $\Spec_{\gamma}(\partial_h f)$. Let $\phi(h)$ be a randomly-chosen element of this large spectrum. One can prove that the assumption of large $U^3$ norm implies that $\phi$ has some ``weak arithmetic structure.'' If we can find ``strong arithmetic structure,'' specifically an affine map $T(x) = Lx + b$ which agrees with $\phi$ on a large portion of the domain, then we can reverse the argument in the above paragraph by ``anti-differentiating'' $L$ to recover the desired quadratic form.

The main difficulty in Tulsiani and Wolf's quadratic Goldreich-Levin theorem is to prove algorithmic versions of several results in additive combinatorics, namely the Balog-Szemer\'edi-Gowers theorem and Freiman's theorem. These tools allow one to efficiently find the ``strong arithmetic structure'' present in $\phi$.

\subsection{Setup of cubic Goldreich-Levin}

Our main result is the algorithmic $U^4$ inverse theorem which is modelled on the Gowers and Mili\'cevi\'c's proof of quantitative bounds for the $U^4$ inverse theorem. We start with the formula
\[\norm{f}_{U^4}^{16} = \E_{a,b,c,d,x}\partial_{a,b,c,d}f(x)=\E_{a,b} \norm{\partial_{a,b} f}_{U^2}^4\]
where $\partial_{a, b, c, d}$ is shorthand for $\partial_a \partial_b \partial_c \partial_d$.

Fourier analysis tells us that $\|g\|_{U^2}^4=\|\hat g\|_{4}^4$ so whenever $\norm{\partial_{a,b} f}_{U^2}$ is large, there is some large Fourier coefficient of $\partial_{a,b}f$. Let us call $A$ the set of pairs $(a,b)$ where $\norm{\partial_{a,b} f}_{U^2}$ is large and define $\phi(a,b)$ to be some large Fourier coefficient of $\partial_{a,b}f$. More formally, suppose that $\norm{f}_{U^4} \geq c$. Then there is a set $A \subset \F_p^n \times \F_p^n$ of density at least $c^{16}/2$ such that $\norm{\partial_{a,b}f}^4_{U^2} \geq c^{16}/2$ for all $(a,b)\in A$ and a function $\phi\colon A \to \F_p^n$ such that $\twopartial \geq c^8/2$ for all $(a,b)\in A$. Furthermore, since we can approximate the $U^2$ norm by sampling, we can produce an approximate membership tester for $A$ and using the Goldreich-Levin algorithm we can compute $\phi$ algorithmically.

In the analogous part of the argument in the $U^3$ setting one produces a set $A\subset\F_p^n$ and a function $\phi\colon A\to\F_p^n$ with ``weak linear structure.'' In this setting we have $A\subset\F_p^n\times\F_p^n$ and a function $\phi\colon A\to\F_p^n$. In their proof of the $U^4$ inverse theorem, Gowers and Mili\'cevi\'c show that the assumption that $\|f\|_{U^4}\geq c$ implies that $A$ and $\phi$ have ``weak bilinear structure.'' They spend the rest of the proof finding stronger and stronger bilinear structure in $A$ and $\phi$, eventually proving the existence of a \textit{bi-affine} map $T$ such that $T(a,b) = \phi(a,b)$ for many $(a,b) \in A\subset\F_p^n\times\F_p^n$. Our goal is to algorithmically produce such a bi-affine map $T$.

\subsection{Finding bi-affine structure}

We now give a rough scheme of the proof, but because we are giving a high level overview it is more convenient for us to use language such as ``1\% structure'' and ``99\% structure.'' Recall we describe something as being 1\% if the density of the object in the appropriate ambient space is something like $\ep > 0$, while we call it 99\% if its density is more like on the scale of $1 - \ep$. 

Recall that the setting we are working with is as follows: we have membership tester for a large set $A \subset (\F_p^n)^2$ as well as query access to a function $\phi \colon A \to \F_p^n$. The goal is to find a bi-affine function $T$ such that $T(a,b) = \phi(a,b)$ holds for a large proportion of $(a,b)$ in the domain $A$.

There will be two concepts that arise in this section: one is the idea of additive structure satisfied by $\phi$ on a set and another is the idea of additive structure on the domain itself. In the dream case $\phi$ is close to bi-affine, so we would expect $\phi$ to possess some form of additive structure. In additive combinatorics there is also the notion of a set possessing additive properties, usually in relation to its successive sumsets or difference sets containing linear structure. It turns out that having both types of additive structure will be crucial in the argument. 

Actually instead of working with $\phi$ it will turn out to be more convenient, for technical reasons, to work with a certain convolution $\psi$ of $\phi$; morally we can think of $\psi$ as a suitable weighted average of $\phi$ across parallelograms. Intuitively, this form of averaging used to define $\psi$ will allow us to do some form of majority vote over parallelograms to select a bi-affine map possessing large overlap with $\phi$. Nevertheless, technicalities aside, $\psi$ should possess similar additive properties as $\phi$. 

\begin{enumerate}
    \item [(1)] (1\% $\Rightarrow$ 99\% structure for $\phi$) From the first step of the argument described above, we see that $\phi$ has 1\% additive structure on $A$. In Gowers and Mili\'cevi\'c's proof, they pass to a subset $A' \subset A$ to boost this 1\% structure of $\phi$, so that $\phi\bigr|_{A'}$ has 99\% structure. They do this via a ``dependent random selection'' probabilistic argument, where the rough idea is that we probabilistically select elements of $A$ to include in $A'$ via a certain distribution that biases our choices towards the inclusion of elements on which $\phi$ respects additive structure. Because of the probabilistic nature of this proof of existence of $A'$, it is not surprising that one can turn it into a probabilistic algorithm for testing membership in $A'$; we can give a sampling randomized algorithm for testing membership in $A'$ as long as we have a certifier which checks that the output set has the desired property of $\phi$ possessing 99\% structure on it. By some algebra, we can show that $\psi$ also possesses a suitable version of 99\% additive structure on $A'$.
    
    \item [(2)] (Obtaining additive structure for the underlying set) As we have alluded to earlier, we would also like to pass from $A'$ to a related set $A''$ which possesses some additive structure, while maintaining the property that $\psi\bigr|_{A''}$ still possesses 99\% additive structure. The kind of set structure that is useful for us in this context turns out to be that of a high rank bilinear Bohr set, namely the level set of a bi-affine map $\beta$. Roughly speaking, high rank bilinear Bohr sets are quasi-random in the sense that the number of solutions to linear equations on this Bohr set is approximately what we would expect for a random subset of $(\F_p^n)^2$. 
    
    This is helpful in our context because suppose $A''$ was completely unstructured, then despite knowing that $\psi\bigr|_{A''}$ is additive we do not have enough control over whether we can suitably interpolate the values of $\psi$ on $A''$ to obtain a bi-affine map $T' \colon A'' \to \F_p^n$. Therefore, having some structure on the underlying set $A''$ helps us to extract more information about $\psi$. 
    
    To that end we will first need to identify the bi-affine map $\beta$, and then find an appropriate high rank level set. The latter is comparatively easier. The former can be done via a bilinear extension of the classical Bogolyubov theorem. The subtlety is that while the classical Bogolyubov theorem is established by examining the large Fourier spectrum of an appropriate convolution and can therefore be algorithmized easily by an application of Goldreich-Levin, the bilinear variant is much more involved. The bilinear variant requires careful successive applications of versions of Balog-Szemer\'edi-Gowers and Freiman's theorems to find affine maps which cover a large Fourier spectrum, before stitching them together in an appropriate way. Since the version of Balog-Szemer\'edi-Gowers theorem that we require differ from that used in Tulsiani and Wolf, we develop this in detail in the section of Algorithmic Tools. 
    
    \item [(3)] (99\% structure $\Rightarrow$ 100\% structure) At this stage we have restricted our attention to a set $A''$ that itself has a lot of structure and $\psi\bigr|_{A''}$ has 99\% structure. By an intricate analysis using the quasi-random properties, namely that $A''$ possesses roughly an expected number of linear patterns with $\psi$ ``respecting'' these linear patterns, we can recover some bi-affine $T'\colon A'' \to \F_p^n$ that agrees with $\psi$ via some form of majority vote over the linear patterns. With some manipulations, we can also show that this $T'$ agrees with $\phi$ on a significant fraction of $(\F_p^n)^2$ as well.
    
    Next, we extend the domain $A''$ of $T'$ to $(\F_p^n)^2$. Gowers and Mili\'cevi\'c construct $T\colon (\F_p^n)^2 \to \F_p^n$ by showing that we can specify the values of $T$ on $(\F_p^n)^2 \backslash A''$ in a way that extends $T'$ consistently, by invoking the quasi-random properties of $A''$. These same quasi-random properties of $A''$ also enables us to sample many linear structures with the property that in each, all but one of its elements lie in $A''$. For each of these structures, we may then extend the domain of $T'$ to include this additional point by linearity. It turns out that doing so gives us query access to $T$ on 99\% of $(\F_p^n)^2$, from which it is not difficult to extend the function further to construct $T\colon (\F_p^n)^2 \to \F_p^n$. 
    
    \item [(4)] (``Anti-differentiating'' and symmetrization) At this point we have achieved the stated goal of recovering a bi-affine function $T$ such that $T(a,b) = \phi(a,b)$. Recall that $\phi(a,b)$ picked out a large Fourier coefficient of $\widehat{\partial_{a,b}f}$. We would therefore need to ``anti-differentiate'' $\phi(a,b)$ in order to recover information of $f$. For technical reasons, we also need $T$ to have some symmetry properties in order for this ``anti-differentiating'' step to work out. This symmetrization step involves dividing by 6, so in $\F_3$ and $\F_2$ some more care needs to be taken and there are a couple more algorithmic linear algebraic steps. After implementing this ``anti-differentiating'' step we will have recovered the degree 3 term $\kappa(x)$ in our cubic phase that correlates with $f$. 
    
    To recover the lower degree terms, it can be shown that $\norm{f \omega^{-\kappa(x)}}_{U^3}$ is large; by implementing the $U^3$ inverse theorem and quadratic Goldreich-Levin algorithm we can recover $q(x)$ such that $\omega^{q(x)}$ has large correlation with $f \omega^{-\kappa(x)}$. Putting this together, we get that $r(x) = \kappa(x) + q(x)$ is the desired cubic with large correlation with $f$. 
\end{enumerate}

\section{Arithmetic decomposition theorem}

Theorem~\ref{thm:alg-u4-inverse}, the algorithmic $U^4$ inverse theorem, is effectively a result of the form ``if 1-bounded $f$ has non-negligible $U^4$ norm then we can retrieve one of its large 'cubic Fourier coefficients'.'' Oftentimes in additive combinatorics and also computer science, however, it is fruitful to study the set of all large Fourier coefficients rather than just one of the large Fourier coefficients. In the classical setting, we have the Goldreich-Levin algorithm which achieves this goal. We will develop an analogue of this in the higher-order Fourier analysis setting. This was also a problem studied by Tulsiani and Wolf in \cite{TW11}. However, as we have mentioned, their decomposition introduces an extra $L^1$ error term. By using the idea of averaging projections, we are able to remove this error term at the expense of having more terms in our decomposition. 

The property of $f$ having a large cubic Fourier coefficient is equivalent to saying that $f$ correlates highly with a cubic phase function. It turns out that the correlation of functions is more convenient than Fourier coefficients, e.g. it enables us to apply Gram-Schmidt process, so we next introduce some formal definitions. For any two functions $f, g \colon \F_p^n \to \mathbb C$, the correlation of $f$ and $g$, denoted as $\langle f, g \rangle$, is defined as 
\[
    \E_{x \in \F_p^n} f(x)\overline{g(x)}.
\]

Tulsiani and Wolf proved the following general decomposition result \cite[Theorem 3.1]{TW11}. 

\begin{theorem}\label{thm:TW-decomp}
Let $X$ be a finite domain and let $\norm{\cdot}_S$ be a semi-norm defined for functions $f \colon X \to \R$ and $\mathcal{Q}$ be an arbitrary class of functions $\overline{q}\colon X \to [-1,1]$ that is also closed under negation. Let $\ep, \delta > 0$ and $B>1$. Let $A$ be an algorithm which, given oracle access to a function $f\colon X \to [-B, B]$ satisfying $\norm{f}_S \geq \ep$, outputs with probability at least $1- \delta$ a function $\overline{q} \in \mathcal{Q}$ such that $\langle f, \overline{q} \rangle \geq \eta$ for some $\eta = \eta(\ep, B)$. Then there exists an algorithm which, given any function $g \colon X \to [-1,1]$, outputs with probability at least $1 - \delta/\eta^2$ a decomposition 
\[ g = c_1 \overline{q_1} + \cdots + c_k \overline{q_k} + e + f\]
satisfying $k \leq 1/\eta^2$, $\norm{f}_S \leq \ep$ and $\norm{e}_1 \leq 1/2B$. Also, the algorithm makes at most $k$ calls to $A$.
\end{theorem}

A high level summary of their proof is as follows. At each step we greedily identify the closest approximation to $f$ from $\mathcal{Q}$. Initialize $f_0 = f$, and at step $t$, we find some $q_t$ which has good correlation with $f_{t-1}$ via \texttt{find-quadratic}. Then we update $f_{t} = f_{t-1} - \langle f_{t-1}, q_t \rangle q_t$. The issue with this that is pointed out in \cite{TW11} is that $\norm{f_t}_{\infty}$ cannot be controlled, and it can be checked that $\langle f_t, q_t \rangle$ degrades as $\norm{f_t}_{\infty}$ increases. To that end we will need to truncate $f_t$ as we iterate so as to have a uniform $\ell^{\infty}$ bound. This truncation introduces an error term $e$. 

By taking $\mathcal{Q}$ to be the set of cubic polynomial phases, $\norm{\cdot}_S$ as $\norm{\cdot}_{U^4}$, we can combine our algorithmic $U^4$ inverse theorem with Tulsiani-Wolf's Theorem \ref{thm:TW-decomp} to obtain the following decomposition result. 

\begin{theorem}\label{thm:TW-U4}
Let $\epsilon, \delta > 0$ and $B > 1$. Then there exists $\eta = \eta(\ep, B)$ and a randomized algorithm which given any 1-bounded function $f \colon \F_p^n \to \C$ as an oracle, outputs with probability $1 - \delta\eta^{-2}$ a decomposition 
\[ f = c_1q_1 + \cdots +c_rq_r + e + g \]
where the $c_i$ are constants, the $q_i$ are cubic phase functions satisfying $r \leq \eta^{-2}$, $\norm{g}_{U^4} \leq \epsilon$, and $\norm{e}_1 \leq (2B)^{-1}$. Also, the algorithm makes at most $r$ calls to the algorithmic $U^4$ inverse theorem \texttt{find-cubic}.
\end{theorem}

In this section we will prove a new decomposition result which removes the $\ell^1$ error term to get an analogue of a kind of Frieze-Kannan weak regularity theorem \cite{FK99} for functions. We will work in the context of proving such a decomposition into polynomial phase functions, assuming that we have as a primitive an appropriate algorithmic $U^k$ inverse theorem. The trade-off we have to make is that instead of having $\poly(\eta^{-1})$ many polynomial phase functions in the decomposition, we will end up having $\exp(\poly(\eta^{-1}))$ such phase functions instead. 

The high level idea is that every bounded function is the sum of a ``structured'' function that is constant on the atoms of a $\sigma$-algebra $\mathcal{B}$ formed by some degree $k-1$ polynomial phase functions and another ``pseudorandom'' function with small $U^k$ norm. Suppose $\mathcal{B}$ is formed by the degree $k-1$ polynomial phases $\{ \omega^{f_1(x)}, \omega^{f_2(x)}, \ldots, \omega^{f_k(x)} \}$. The projection onto $\mathcal{B}$ effectively can be rewritten as a weighted sum of polynomial phases $\omega^{\beta_1 f_1(x) + \cdots + \beta_k f_k(x)}$ for some $\beta_1, \ldots, \beta_k \in \F_p$. We iteratively build up $\mathcal{B}$: each time we identify a new polynomial phases via the algorithmic $U^k$ inverse theorem primitive. In this setup, we do not fix the coefficients $c_i$ in our decomposition and instead re-compute it each time we enlarge $\mathcal{B}$. In \cite{TW11}, this is not accounted for; the coefficients are instead fixed and they study $f - \sum_{i=1}^{k} c_i \omega^{q_i}$. As such, while they can control $\norm{\cdot}_2$, they end up losing control of $\norm{\cdot}_{\infty}$. In our case, we can instead control for both of these norms at the same time, removing the need to do any form of truncation and thereby circumventing the need to introduce $e$.  

We terminate the algorithm when we obtain a $\sigma$-algebra $\mathcal{B}$ such that the residual $g := f - \E(f | \mathcal{B})$ satisfies $\norm{g}_{U^4} \leq \epsilon$. To that end, we will need to be able to compute the $U^4$ norm of a function. Since the $U^4$ norm is an expected value, we can give a probabilistic sampling procedure to make such a calculation.

\vspace{\algtopskip}
\noindent \fbox{
\parbox{\textwidth}{
\texttt{Uk(f,$\ep$,k)}:\\
\textbf{Input} a query access to $f \colon \F_p^n \to \C$, $\epsilon >0$, and $k \ge 2$\\
\textbf{Output} 1 if $\norm{f}_{U^k} \leq \ep$ and 0 otherwise 
\begin{itemize}
    \item Sample $\poly(\log(\delta^{-1}))$ $(k+1)$-tuples $(x, a_1, \ldots, a_{k}) \in \F_p^{k+1}$ and for each such tuple compute $y_{a_1,a_2,\ldots,a_k} = \partial_{a_1,a_2,\ldots,a_k}f(x)$. Compute the average of all these $y_{a_1,a_2,\ldots,a_k}$, and let the value of this average be $\alpha$.
    \item If $\alpha \leq \ep$, return 1. Otherwise return 0.
\end{itemize}
}
}

\begin{theorem}
Let $\mathcal{Q}$ be the class of degree $k-1$ polynomial phases $\overline{q} = \omega^{q(x)}$ where $\omega = e^{2\pi i/p}$ and $q(x)$ is a degree $k-1$ polynomial in $x$. Let $\ep, \delta > 0$. Let \texttt{find-poly} be an algorithm which given oracle access to a 1-bounded function $f\colon \F_p^n \to \C$ satisfying $\norm{f}_{U^k} \geq \ep$ outputs with probability at least $1- \delta$ a function $\overline{q} \in \mathcal{Q}$ such that $\abs{\langle f, \overline{q} \rangle} \geq \eta$ for some $\eta = \eta(\ep)$. Then there is an algorithm \texttt{Uk-weak-regularity}, given any 1-bounded function $g \colon \F_p^n \to \C$, which outputs with probability at least $1 - 20/9 \cdot \delta \eta^{-2}$ a decomposition into degree $k-1$ polynomial phase functions
\[ f= c_1 \overline{q_1} + \cdots + c_r \overline{q_r} + g \]
satisfying $r \leq p^{10/9 \cdot \eta^{-2}}$ and $\norm{g}_{U^k} \leq \ep$. The algorithm makes at most $r$ calls to \texttt{find-poly}.
\end{theorem}

\vspace{\algtopskip}
\noindent \fbox{
\parbox{\textwidth}{
\texttt{Weak-regularity(f)}:
\begin{itemize}
    \item Initialize $\widetilde{g} = f$, $\widetilde{\fs} = 0$ and $\mathcal{L} = \{ 0 \}$. We use $\mathcal{L}$ to store polynomials; the corresponding polynomial phase functions will appear in the decomposition of $f$. 
    \item Run \texttt{find-poly} on (a suitably normalized version of) $\widetilde{g}$. If the output of \texttt{find-poly} is $\perp$, then return $f = \fs+\widetilde{g}$.
    \item Otherwise, suppose the output of \texttt{find-poly} is $\overline{q} = \omega^{q(x)}$. For each $r \in \mathcal{L}$ and all $i$, add $iq +r$ to $\mathcal{L}$.
    \item Using Gaussian elimination, retrieve the maximally independent subset of $\mathcal{L}$ and discard all the elements of $\mathcal{L}$ which do not lie in this maximally independent subset. 
    \item Run Gram-Schmidt on $\mathcal{L}$ and let the output be $\mathcal{L'}$. Note that Gram-Schmidt also outputs approximations for the coefficients $\alpha_{qb} = \langle q, b \rangle$ for $q \in \mathcal{L}, b \in \mathcal{L}'$. Now, since $\beta_b = \langle f, b \rangle$ for $b \in \mathcal{L}'$ is an inner product we can estimate it as $\widetilde{\beta_b}$ by standard sampling. 
    \item Update $\widetilde{\fs} = \sum_{q \in \mathcal{L}} \left( \sum_{b \in \mathcal{L'}} \alpha_{qb} \widetilde{\beta_b} \right) q$ and $\widetilde{g} = f - \widetilde{\fs}$. If \texttt{U4(f,$\ep$,k)} returns 1, terminate. Otherwise, repeat from \texttt{find-poly}.
\end{itemize}
}
}\vspace{\algbotskip}

We will use the notation $\overline{q} \in \mathcal{Q}$ to mean that $\overline{q}$ is the corresponding polynomial phase function to the polynomial $q$. 

\begin{proof}

Before proceeding further, we begin by setting up some notation that will help in the arguments to come. Let the elements of $\mathcal{L}$ at the $r$th step of the algorithm be $q_1, \ldots, q_r$. Let the $\sigma$-algebra formed by $q_1, \ldots, q_r$ be $\mathcal{B}_r$. As mentioned earlier, note that the Gram-Schmidt operation recovers the coefficients $\alpha_{qb}$. For $\beta_b$, we can only approximate them. To that end, note by an application of Lemma~\ref{lem:CH}, we may assume we have a primitive \texttt{approx-iprod($\epsilon, \delta$)} that runs in time $O(\poly(\epsilon^{-1}, \log(\delta^{-1})))$ that produces some $\lambda \in \C$ such that with probability at least $1 - \delta$ we have $\abs{\langle f, \widetilde{g} \rangle - \lambda} \leq \epsilon$. 

Define $(\fs)_i = \sum_{q,b} \alpha_{qb}\beta_b q$ and $g_i = f - (\fs)_i$ to be the corresponding precise values at the $i$th step of the algorithm in the dream case when no approximation is necessary. For simplicity we will sometimes drop the index $i$ when the context is clear that we are considering a particular step of the algorithm. 

First, we will justify that $(\fs)_r = \E(f \mid \mathcal{B}_r)$. To do this, we start by showing that $\E(f \mid \mathcal{B}_r) = \sum_{\mathbf{j}}c_{\mathbf{j}} \omega^{p_{\mathbf{j}}}$ where each $p_{\mathbf{j}}$ is a $\F_p$-linear combination of $q_1, \ldots, q_r$. Observe that we can decompose the level sets as follows 
\[ \mathbf{1}(q_1(x) = a_1, \ldots, q_r(x) = a_r) = \prod_{j=1}^{r} \left( \frac{1}{p} \sum_{i=0}^{p-1} \omega^{i(q_j(x) - a_j)}\right) = \frac{1}{p^r} \sum_{0 \leq i_1, \ldots, i_r \leq p-1} \omega^{-(i_1a_1 + \cdots + i_r a_r)} \omega^{i_1q_1(x) + \cdots + i_rq_r(x)}.\]
In particular, by collecting terms, this implies that we are able to write $\E(f \mid \mathcal{B}_r)(x)$ in the following form
\[ \E(f \mid \mathcal{B}_r)(x) = \E_{y \in \F_p^n} f(y)\mathbf{1}(q_1(y) = q_1(x), \ldots, q_r(y) = q_r(x)) = \sum_{0 \leq i_1, \ldots, i_r \leq p-1} c_{i_1, \ldots, i_r} \omega^{i_1q_1(x) + \cdots + i_rq_r(x)}  \]
for some $c_{i_1, \ldots, i_r} \in \C$. Note that we may assume WLOG that the terms $\omega^{i_1q_1(x) + \cdots + i_rq_r(x)}$ in the sum above with non-zero coefficients $c_{i_1, \ldots, i_r}$ are all linearly independent. The next step is to retrieve the coefficients $c_{i_1, \ldots i_r}$. If the terms $\omega^{i_1q_1(x) + \cdots i_rq_r(x)} =: \overline{p_{i_1, \ldots, i_r}}$ with non-zero coefficients were all orthogonal, then we could just retrieve the coefficients as $\langle \E(f \mid \mathcal{B}_r), \overline{p_{i_1, \ldots, i_r}} \rangle = \langle f,  \overline{p_{i_1, \ldots, i_r}} \rangle $. To that end, we will first run Gram-Schmidt to orthogonalize these polynomial to get the polynomials phases in $\mathcal{L}'$. The $c_{i_1, \ldots, i_r}$ are suitable linear combinations of $\langle f, b \rangle$. Precisely, write $\overline{p_{i_1, \ldots, i_r}} = \sum_{i=1}^{r} \langle \overline{p_{i_1, \ldots, i_r}}, b_i \rangle b_i$ for $b_i \in \mathcal{L}$ then $c_{i_1, \ldots, i_r} = \sum_{i=1}^{r} \langle \overline{p_{i_1, \ldots, i_r}}, b_i \rangle \beta_{b_i}$. As noted, $\langle \overline{p_{i_1, \ldots, i_r}}, b_i \rangle$ are the coefficents obtained by Gram-Schmidt. Note that although there are many projections $\langle \overline{p_{i_1, \ldots, i_r}} , b_i \rangle $ to compute, these can be computed without querying the oracle \texttt{find-poly} and therefore do not affect query complexity. 

Next, we need to show $\norm{\widetilde{g}}_{\infty}$ is bounded by a constant so that after a suitable scaling, we can run \texttt{find-poly} in the second step; this is because \texttt{find-poly} takes as input a 1-bounded function. We will prove that $\norm{\widetilde{g}}_{\infty} \leq 3$, which means we can pass the normalized form $\widetilde{g}/3$ into \texttt{find-poly}. By what we have established earlier, we have that $g_r = f - \E(f \mid \mathcal{B}_r)$ and so $\norm{g_r}_{\infty} \leq 2$. Recall that in computing $\widetilde{g}$ we had to estimate the projections $\langle f, b \rangle$ using \texttt{approx-iprod}. We can pick the parameters in \texttt{approx-iprod} to ensure that we can approximate $g$ arbitrarily well with $\widetilde{g}$ since the run-time of \texttt{approx-iprod} is independent of $n$. Before we proceed further, it will turn out that it is more convenient to work with the following modified form of \texttt{Weak-regularity}.

\vspace{\algtopskip}
\noindent \fbox{
\parbox{\textwidth}{
\texttt{mod-Weak-regularity(f)}:
\begin{itemize}
    \item Initialize $\widetilde{g} = f$, $\widetilde{\fs} = 0$ and $\mathcal{L} = \{ 0 \}$. We use $\mathcal{L}$ to store polynomials; the corresponding polynomial phase functions will appear in the decomposition of $f$. 
    \item If $\abs{\mathcal{L}} \geq p^{2 \cdot \eta^{-2}}$, return $f = \widetilde{\fs} + \widetilde{g}$.  Otherwise, run \texttt{find-poly} on (a suitably normalized form of) $\widetilde{g}$. 
    \item If the output of \texttt{find-poly} is $\perp$, then return $f = \widetilde{\fs}+\widetilde{g}$.
    \item Otherwise, suppose the output of \texttt{find-poly} is $\overline{q} = \omega^{q(x)}$. For each $r \in \mathcal{L}$, add $iq +r$ to $\mathcal{L}$ where $1 \leq i \leq p-1$. 
    \item Using Gaussian elimination, retrieve the maximally independent subset of $\mathcal{L}$ and discard all the elements of $\mathcal{L}$ which do not lie in this maximally independent subset. 
    \item Run Gram-Schmidt on $\mathcal{L}$ and let the output be $\mathcal{L'}$. Note that Gram-Schmidt also outputs approximations for the coefficients $\alpha_{qb} = \langle q, b \rangle$ for $q \in \mathcal{L}, b \in \mathcal{L}'$. Now, since $\beta_b = \langle f, b \rangle$ for $b \in \mathcal{L}'$ is an inner product we can estimate it as $\widetilde{\beta_b}$. 
    \item Update $\widetilde{\fs} = \sum_{q \in \mathcal{L}} \left( \sum_{b \in \mathcal{L'}} \alpha_{qb} \widetilde{\beta_b} \right) q$ and $\widetilde{g} = f - \fs$. If \texttt{U4(f, $\ep$, k)} returns 1, terminate. Otherwise, repeat from \texttt{find-poly}.
\end{itemize}
}
}\vspace{\algbotskip}

The only additional clause in \texttt{mod-Weak-regularity} we have added as compared to \texttt{Weak-regularity} is the final step. Eventually, we will justify that this additional clause of checking $\abs{\mathcal{L}} \leq p^{2 \cdot \eta^{-2}}$ is inconsequential and so \texttt{mod-Weak-regularity} has essentially the same output as \texttt{Weak-regularity}. Its introduction is merely for convenience. By applying \texttt{approx-iprod}($\eta^2p^{4\eta^{-2}}/50, \delta p^{-2\eta^{-2}}$), since $\abs{\mathcal{L}} \leq p^{2 \cdot \eta^{-2}}$ it follows that we have with probability at least $1 - \delta r^{-2}$ that $\abs{\langle f,b \rangle - \widetilde{\beta_b}} \leq \eta^2/(50r)$. In particular, since $\widetilde{\fs} = \sum_{q \in \mathcal{L}} \left( \sum_{b \in \mathcal{L'}} \alpha_{qb} \widetilde{\beta_b} \right) \omega^{q}$ and there are at most $r$ terms $\widetilde{\beta_b}$, it follows that we have $\norm{\widetilde{\fs} - \fs}_{\infty} \leq \eta^2/50$. Consequently, with probability at least $1 - \delta$, we have $\norm{\widetilde{g}}_{\infty} = \norm{f - \widetilde{\fs}}_{\infty} \leq \norm{g}_{\infty} + \eta^2/50 \leq 3$.

By construction, we have with probability at least $1- \delta$ that $\norm{\widetilde{g}}_{U^k} < \epsilon$ which is the guarantee from \texttt{find-Q}; if $\norm{\widetilde{g}}_{U^k} \geq \epsilon$ then with probability at least $1- \delta$, \texttt{find-Q} would have enumerated another $\overline{q} \in \mathcal{Q}$ in the algorithm and we would not have terminated. Next, we bound the length of the decomposition. We will utilize an energy increment argument. This is encapsulated in the following, which can be thought of as a ``noisy'' version of \cite[Lemma 3.8]{Montreal}.

\begin{lemma}\label{lem:ener-inc}
    Let $\mathcal{B}$ be the $\sigma$-algebra corresponding to the elements of $\mathcal{L}$ at a certain stage of \texttt{mod-Weak-regularity} and suppose that $\norm{\widetilde{g}}_{U^k} \geq \ep$. Then in the next stage \texttt{mod-Weak-regularity} extends $\mathcal{L}$ by an element to $\mathcal{L}_1$ with corresponding $\sigma$-algebra $\mathcal{B}_1$ such that
    \[ \norm{\E(f \mid \mathcal{B}_1)}_2^2 \geq \norm{\E(f \mid \mathcal{B})}_2^2 + 9\eta^2/10. \]
\end{lemma}

\begin{proof}
    Note that an application of algorithm \texttt{find-Q} outputs some $\overline{q} \in \mathcal{Q}$ such that $\eta \leq \abs{\langle \widetilde{g}, \overline{q} \rangle}$. In particular, \texttt{mod-Weak-regularity} forms $\mathcal{L}_1$ by adding $q$ to $\mathcal{L}$. Let the $\sigma$-algebra generated by $q$ be $\mathcal{B}_q$. Observe that $ \langle \widetilde{g}, \overline{q} \rangle = \E_x \widetilde{g}(x) \overline{q}(x) = \E_x \E(\widetilde{g} \mid \mathcal{B}_q)(x) \overline{q}(x)$ since $\overline{q}(x)$ is evidently $\mathcal{B}_q$ measurable. In particular, since $\norm{\overline{q}}_{\infty} \leq 1$, it follows that $\norm{\E(\widetilde{g} \mid \mathcal{B}_q)}_1 \geq \eta$.
    
    By the Cauchy-Schwarz inequality and the triangle inequality, we have that 
    \[ \left| \norm{\E(g \mid \mathcal{B}_1)}_2^2 - \norm{\E(\widetilde{g} \mid \mathcal{B}_1)}_2^2 \right| \leq \norm{\E(g- \widetilde{g} \mid \mathcal{B}_1)}_2 \left( \norm{\E(g \mid \mathcal{B}_1)}_2 + \norm{\E(\widetilde{g} \mid \mathcal{B}_1)}_2 \right). \]
    Since $\norm{\E(g \mid \mathcal{B}_1)}_2 \leq \norm{g}_{\infty} \leq 2$ and $\norm{\E(\widetilde{g} \mid \mathcal{B}_1)}_2 \leq \norm{\widetilde{g}}_{\infty} \leq 3$, it follows that 
    \[ \left| \norm{\E(g \mid \mathcal{B}_1)}_2^2 - \norm{\E(\widetilde{g} \mid \mathcal{B}_1)}_2^2 \right| \leq 5 \norm{\E(g- \widetilde{g} \mid \mathcal{B}_1)}_2 \leq 5 \norm{g - \widetilde{g}}_2 \leq 5 \norm{g - \widetilde{g}}_{\infty} \leq \eta^2/10, \]
    where we use $\norm{g-\widetilde{g}}_\infty = \norm{\fs - \widetilde{\fs}}_\infty \le \eta^2/50$ at the last inequality.
    
    Now we are in a position to establish the energy increment, via Pythagoras' Theorem. Note that Pythagoras' tells us that 
    \[ \norm{\E(f \mid \mathcal{B}_1)}_2^2 = \norm{\E(f \mid \mathcal{B})}_2^2 + \norm{\E(f \mid \mathcal{B}_1) - \E(f \mid \mathcal{B})}_2^2. \]
    This rearranges as 
    \begin{align*}
        \norm{\E(f \mid \mathcal{B}_1)}_2^2 - \norm{\E(f \mid \mathcal{B})}_2^2 &= \norm{\E(f \mid \mathcal{B}_1) - \E(f \mid \mathcal{B})}_2^2 \\
        &= \norm{\E(g \mid \mathcal{B}_1)}_2^2 \\
        &\geq \norm{\E(\widetilde{g} \mid \mathcal{B}_1)}_2^2 - \left| \norm{\E(g \mid \mathcal{B}_1)}_2^2 - \norm{\E(\widetilde{g} \mid \mathcal{B}_1)}_2^2 \right|\\
        &\geq 9\eta^2/10,
    \end{align*}
    where in the last line we recall the earlier bound of $\norm{\E(\widetilde{g} \mid \mathcal{B}_q)}_1 \geq \eta$, which upon applying the Cauchy-Schwarz inequality shows that $\norm{\E(\widetilde{g} \mid \mathcal{B}_1)}_2^2 \geq \norm{\E(\widetilde{g} \mid \mathcal{B}_q)}_1^2 \geq \eta^2$.
\end{proof}

If $\norm{f - f_{\text{struc}}}_{U^k} \leq \ep$ then \texttt{mod-Weak-regularity} would have terminated. Otherwise Lemma~\ref{lem:ener-inc} allows us to extend $\mathcal{L}$ with a corresponding increment in energy by $9 \eta^2/10$. Since $f$ is 1-bounded, the energy $\norm{\E(f \mid \mathcal{B})}_2^2$ lies in the interval $[0,1]$. This means that the algorithm has to terminate in at most $10/9 \cdot \eta^{-2}$ steps, as desired. 
\end{proof}

In our setting, we can take $\mathcal{Q}$ to be the class of cubic phase functions. Recalling our algorithmic $U^4$ inverse theorem (Theorem~\ref{thm:main}), we obtain as a corollary the cubic Goldreich-Levin algorithm of Theorem~\ref{thm:cub-GL}. 

\begin{theorem}

Let $\ep, \delta > 0$. Let \texttt{find-cubic} be an algorithm which given query access to a 1-bounded function $f\colon \F_p^n \to \C$ satisfying $\norm{f}_{U^4} \geq \ep$ outputs with probability at least $1- \delta$ a cubic polynomial phase function $\overline{q}$ such that $\abs{\langle f, \overline{q}} \rangle \geq \eta$ for some $\eta = \eta(\ep)$. Then there is an algorithm \texttt{U4-weak-regularity} such that given any 1-bounded function $g \colon \F_p^n \to \C$ outputs with probability at least $1 - (20/9)\cdot \delta \eta^{-2}$ a decomposition 
\[ f= c_1 \overline{q_1} + \cdots + c_r \overline{q_r} + g \]
where $\overline{q_i}$ are cubic phase functions, such that $r \leq p^{10/9 \cdot \eta^{-2}}$ and $\norm{g}_{U^4} \leq \ep$. The algorithm makes at most $r$ calls to \texttt{find-cubic}.
\end{theorem}

\section{Algorithmic tools}

In this section, we enumerate some algorithmic primitives that we will be utilizing in later sections. We will be using the standard Chernoff bounds throughout the paper.

\begin{lemma}\label{lem:CH}
    If $X$ is a random variable with $|X| \leq 1$ and $\mu_t = \frac{X_1 + \cdots + X_t}{t}$ where $X_i$ are samples, then 
    \[ \Pb[|\E[X] - \mu_t| \geq \eta] \leq 2 \exp(-2 \eta^2 t).\]
\end{lemma}

We use several versions of the classical Goldreich-Levin algorithm which we state below.

\begin{theorem}[Classical Goldreich-Levin algorithm]
Given query access to $f \colon \F_2^n \to \F_2$ and input $0 < \tau \leq 1$ there exists a $\poly(n, 1/\tau)$-time algorithm \texttt{Goldreich-Levin(f,$\tau$)}, which with high probability, outputs a list $L = \{ r_1, \ldots, r_k \}$ with the following guarantee:
\begin{itemize}
    \item If $|\widehat{f}(r)| \geq \tau$ then $r \in L$.
    \item For $r_i \in L$, we have $|\widehat{f}(r_i)| \geq \tau/2$.
\end{itemize}
\end{theorem}

\begin{theorem}[Noisy Goldreich-Levin] \label{thm:noisy-GL}
Let $\eta, \omega, \delta > 0$ and $0 < \tau \le 1$.
Let $f\colon \mathbb{F}_p^n \to \C$ be a 1-bounded function. Given query access to a random function $f'\colon \mathbb{F}_p^n \to \C$ such that with probability at least $1- \eta$ we have $\abs{f'(x) - f(x)} \leq \omega $, there is a randomized algorithm \texttt{noisy-GL} that makes $O(n \log n\poly(1/\tau, 1/\eta, 1/\omega,\log(1/\delta)))$ queries to $f'$ and with probability at least $1- \delta$ outputs a list $L= \{ r_1, \cdots, r_k \}$ with the following guarantee: 
   \begin{itemize}
    \item If $|\widehat{f}(r)| \geq \tau$ then $r \in L$.
    \item For $r_i \in L$, we have $|\widehat{f}(r_i)| \geq \frac{\tau}{2} - \frac{3}{2} \left( \eta + (1-\eta)  \omega \right)$.
\end{itemize}
\end{theorem}

We have not found this noisy version stated explicitly in the literature, though it can be proved using the same techniques as the original Goldreich-Levin theorem. For completeness we give the proof in Appendix~\ref{sec:app}.

We will need to use the algorithmic $U^3$ inverse theorem of Tulsiani and Wolf. Though their algorithm is only stated for $\F_2^n$, a small modification of their algorithm works over $\F_p^n$.

\begin{theorem}[Algorithmic $U^3$ inverse theorem \cite{TW11}]\label{thm:u3inv}
Given $\ep, \delta > 0$, there exists $\eta = \exp(-1/\ep^C)$ and a randomized algorithm \texttt{find-quadratic} running in time $O(n^4\log n \cdot \poly(1/\ep, 1/\eta, \log(1/\delta))$ which, given query access to $f\colon  \mathbb{F}_p^n \rightarrow \C$ that is 1-bounded, either outputs a quadratic form $q$ or $\perp$. The algorithm has the following guarantee:
\begin{itemize}
    \item If $\norm{f}_{U^3} \geq \ep$ then with probability at least $1 - \delta$ it finds a quadratic form $q$ such that $\langle f, \omega^{q} \rangle \geq \eta$. 
    \item The probability that the algorithm outputs a quadratic form $q$ with $\langle f, \omega ^q \rangle \leq \eta/2$ is at most $\delta$. 
\end{itemize}
\end{theorem}

Another algorithmic tool from \cite{TW11} that we will use is the algorithmic Balog-Szemer\'edi-Gowers theorem. Technically \cite{TW11} gives a modified version which only applies over $\F_2^n$. We state and prove a more general version that applies in all finite abelian groups.

\begin{theorem}\label{thm:BSG-test}
Let $\rho, \delta > 0$. Let $A$ be a subset of a finite abelian group for which we have query access as well as the ability to sample a random element. Suppose $E_+(A)\geq \rho|A|^3$ where $E_+(A) = |\{ (a_1, a_2, a_3, a_4) \in A^4 : a_1+a_2=a_3+a_4 \}|$. Then for each $u \in A$, there exist sets $A^{(1)}(u) \subset A^{(2)}(u) \subset A$ and an algorithm \texttt{BSG-Test} such that the output of \texttt{BSG-Test} satisfies the following with probability at least $1 - \delta$. For each $u, v \in A$,
\begin{itemize}
    \item \texttt{BSG-Test}($u,v,\rho,\delta$) = 1 then $v \in A^{(2)}(u)$. 
    \item \texttt{BSG-Test}($u,v, \rho,\delta$) = 0 then $v \not \in A^{(1)}(u)$.  
\end{itemize}
Moreover, if $u$ is chosen uniformly random from $A$, then with probability at least $\poly(\rho)$ we have that:
\begin{itemize}
    \item $|A^{(1)}(u)| \geq \poly(\rho) \cdot |A|$,
    \item $|A^{(2)}(u) + A^{(2)}(u)| \leq \poly(\rho^{-1}) \cdot |A|$.
\end{itemize}
\end{theorem}

The proof of this result is quite similar to the corresponding result in \cite{TW11}, so we defer the proof to Appendix~\ref{sec:app}.

One very useful fact in additive combinatorics is that for a set $A\subseteq \F_p^n$ and a function $\phi\colon A\to\F_p^n$, if $\phi$ preserves many additive quadruples in the sense that there are $\rho |A|^3$ quadruples $x - y = z - w$ such that $\phi(x) - \phi(y) = \phi(z) - \phi(w)$, then $\phi$ must agree with an affine map on a large (quasi-polynomial) fraction of $A$. This fact is proved by combining Balog-Szemer\'edi-Gowers with the Freiman's theorem.

An algorithm version of this result was proved in \cite{BRTW12} building upon a quantitatively weaker version \cite{TW11} for the case $p=2$. This same argument works for all $p$ except for the algorithmic Balog-Szemer\'edi-Gowers step which was tailored to $p=2$. Combining the above \ref{thm:BSG-test} with their arguments one proves the following.

\begin{theorem}\label{thm:find-affine}
Let $\rho, \delta > 0$. Let $A$ be a subset of $\F_p^n$ for which we have query access as well as the ability to sample a random element via \texttt{sampler-A}. Let $\phi\colon A \to \F_p^n$ be a function such that there exist $\rho |A|^3$ quadruples $(x, y, z, w)$ satisfying $x - y = z - w$ and $\phi(x) - \phi(y) = \phi(z) - \phi(w)$. Then there exists an algorithm \texttt{find-affine-map} which makes $O(n^3\poly(\rho^{-1}, \log(\delta^{-1})))$ queries to $A$, \texttt{sampler-A}, and $\phi$ such that with probability at least $1- \delta$ outputs an affine map $T$ that agrees with $\phi$ on at least a $\qpoly(\rho)$ fraction of $A$.
\end{theorem}

\section{Finding correlated cubic phases}
In this section, we prove our main technical result, the algorithmic $U^4$ inverse theorem under two slight weakenings: the quantitative bounds are slightly worse, and we work in $\F_p^n$ only when $p\geq 5$. We do so for ease of exposition. In the two following sections we will explain how to modify the algorithm to overcome these limitations, improving the quantitative bounds by a single exponential and then extending to all $p$.

We first give several subroutines which make up the algorithm and prove their correctness. Then we show how combining these subroutines proves the main theorem.

\textbf{Notation and conventions:} Throughout this section, we fix a prime $p$ and use $G$ to denote $\F_p^n$. We say that a function $f\colon G \to \mathbb C$ is \emph{bounded} if $\norm{f}_\infty \le 1$.

In this section, we will often assume that we have certain oracles that give us either query access to a function or to a probability distribution. Specifically, we say that we have a membership tester for a set $A$ if there is an oracle which tells us whether an input $x$ is in $A$ or not. Also, we say that we have query access to a function $f\colon X \to Y$ if there is an oracle which for an input $x \in X$ returns $f(x) \in Y$. The last type of oracle that we use is oracle access to a probability distribution; for a function $f\colon X \to Y$ where $Y$ is the space of probability distribution on $Z$, then for each input $x \in X$, the oracle returns $z \in Z$ according to the probability distribution $f(x)$.

Throughout the algorithms, we consider each variable to be global, meaning that even if some sub-algorithms are terminated we can still access variables that were computed already.

Before we begin in earnest, we formalize how to sample from a set $A \subset G^2$ or $A \subset G$. Here we assume that we have query access to $A$ and the ambient group is $X$, which will be either $G$ or $G^2$. 

\vspace{\algtopskip}
\noindent \fbox{
\parbox{\textwidth}{
     \texttt{sampler(A,t,X)}: \# With high probability samples $t$ elements of $A$; suppose $A$ has density $\alpha$. 
    \begin{itemize}
        \item Take $r = O(\alpha^{-1}t)$ samples $x_1, \cdots, x_r \in G$ and output only those for which $x_i \in A$.
    \end{itemize}
}
}\vspace{\algbotskip}

The first step of the algorithm is to restrict to the large subset $A\subset G\times G$ defined by $(a,b)\in A$ if there exists $\xi\in G$ such that $|\widehat{\partial_{a,b}f}(\xi)|\geq\epsilon$ and define a function $\phi\colon A\to G$ such that $\phi(a,b)$ is one of the $\xi$ satisfying the previous inequality. However, since we can only approximate these Fourier coefficients, all we can actually do is sandwich $A$ between two sets $A_1\subset A\subset A_2$.

\begin{theorem}\label{thm:step1}
Given a bounded $f \colon G \to \mathbb C$ and $\ep > 0$, define $A_1\subset A_2 \subset G\times G$ so that $(a,b)\in A_1$ if $\norm{\widehat{\partial_{a,b}f}}_{\infty} \geq 2\epsilon$ and $(a,b)\in A_2$ if $\norm{\widehat{\partial_{a,b}f}}_{\infty} \geq \epsilon$. 

There is an algorithm \texttt{member-A} that makes $O(\poly(n, 1/\epsilon, \log(1/\delta)))$ queries to $f$ and with probability at least $1-\delta$ outputs 1 if $(a,b)\in A_1$ and with probability at least $1-\delta$ outputs 0 if $(a,b)\not\in A_2$. There is an algorithm \texttt{query-phi} that makes $O(\poly(n, 1/\epsilon, \log(1/\delta)))$ queries to $f$ and with probability at least $1- \delta$ outputs $\phi(a,b)$ such that $\abs{\widehat{\partial_{a,b}f}(\phi(a,b))} \geq \epsilon$ if $(a,b) \in A_2$ and has no guarantees otherwise. 
\end{theorem}

\vspace{\algtopskip}
\noindent \fbox{
\parbox{\textwidth}{
\texttt{membership-A(f,a,b)}:\\
\textbf{Input} query access to $f\colon G\to \mathbb C$, $(a,b) \in G \times G$\\
\textbf{Output} 1 if $(a,b) \in A_1$ and 0 if $(a,b) \notin A_2$ with high probability
\begin{itemize}
    \item Using query access to $f$ we can obtain query access to $\partial_{a,b}f = f(x) \overline{f(x+a)f(x+b)} f(x+a+b)$. 
    \item Run \texttt{Goldreich-Levin}($\partial_{a,b}f$, $2\epsilon$). Return 1 if the output is non-empty and 0 otherwise.
\end{itemize}
}
}\vspace{\algbotskip}

\vspace{\algtopskip}
\noindent \fbox{
\parbox{\textwidth}{
\texttt{phi(f,a,b)}:\\
\textbf{Input} query access to $f\colon  G \to \mathbb C$, $(a,b) \in G \times G$\\
\textbf{Output} $\phi(x,y)$ such that $|\widehat{\partial_{a,b}f}(\phi(a,b))| \ge \epsilon$
\begin{itemize}
    \item If the output of \texttt{membership-A(f,a,b)} is 0, return $\perp$. Else, return an arbitrary element from \texttt{Goldreich-Levin}($\partial_{a,b}f$, $2\epsilon$). 
\end{itemize}
}
}\vspace{\algbotskip}

\begin{proof}
We proceed $\texttt{Goldreich-Levin}(\partial_{a,b}f, 2\epsilon)$, which with $O(\poly(n, 1/\epsilon, \log(1/\delta) ))$ many queries to $\partial_{a,b}f$ outputs a list $L_{a,b} = \{ r_1, \ldots, r_k \}$ which satisfies the following with probability at least $1-\delta$:
If $\abs{ \widehat{\partial_{a,b}f}(r) } \ge 2\epsilon$ then $r \in L_{a,b}$, and for $r_i \in L_{a,b}$, $\abs{\widehat{\partial_{a,b}f}(r_i)} \ge \epsilon$. We output 1 if $L$ is non-empty and 0 otherwise.

Suppose that $(a,b) \in A_1$. Then assuming $\texttt{Goldreich-Levin}$ does not fail, there is $r \in L_{a,b}$ such that $\abs{\widehat{\partial_{a,b}}(r)} \ge 2\epsilon$. Hence the output is 1. The algorithm fails only if $\texttt{Goldreich-Levin}$ fails, so the overall algorithm succeeds with probability at least $1-\delta$.

On the other hand, if $(a,b) \notin A_2$, then $\norm{\widehat{\partial_{a,b}f}}_\infty < \epsilon$. Assuming $\texttt{Goldreich-Levin}$ does not fail, if there is $r \in L_{a,b}$, then $\abs{\widehat{\partial_{a,b}f}} \ge \epsilon$, a contradiction. Hence $L$ is empty, and the output is 0. The overall algorithm fails only if $\texttt{Goldreich-Levin}$ fails, so the overall algorithm succeeds with probability at least $1-\delta$.
\end{proof}

In the next theorem, we utilize the notions of a 4-arrangement and a second-order 4-arrangement following \cite{GM17}. These structures play important roles in finding affine structures in $\phi\colon G \times G \to G$. We start with the notion of a vertical parallelogram.

\begin{definition}[(second-order) vertical parallelogram]
A \emph{vertical parallelogram} is a set of 4 points $(x,y), (x,y+h), (x+w,y'), (x+w,y'+h) \in G \times G$ for some $x, y, y', h, w \in G$. We call $w$ and $h$ the respective \emph{width} and \emph{height} of the vertical parallelogram. A \emph{second-order vertical parallelogram} is a quadruple $Q = (P_1, P_2, P_3, P_4)$ such that $((w(P_1),h(P_1)), (w(P_2), h(P_2)), (w(P_3), h(P_3)), (w(P_4), h(P_4))$ form a vertical parallelogram where $w(P_i)$ and $h(P_i)$ denote the width and height of $P_i$, respectively.
\end{definition}

\begin{definition}[4-arrangement, second-order 4-arrangement]
A \emph{4-arrangement} is the set of 8 vertices corresponding to a pair $(P_1, P_2)$ of vertical parallelograms of the same width and height. A \emph{second-order 4-arrangement} is the set of 16 vertices corresponding to a pair $(Q_1, Q_2)$ of second-order vertical parallelograms of the same width and height.
\end{definition}

Essentially, we define 4-arrangement and second-order 4-arrangement to figure out whether a map $\phi$ behaves like an affine map in each vertical parallelogram of fixed width and height. Therefore we need another concept that measures how well $\phi$ behaves with the vertical parallelograms. 
\begin{definition}[$\phi$ respects the (second-order) 4-arrangement] Given a map $\phi \colon G \times G \to G$ and a vertical parallelogram $P = ((x,y), (x,y+h), (x+w,y'), (x+w,y'+h))$, define $\phi(P)$ as $\phi(x,y)-\phi(x,y+h)-\phi(x+w,y')+\phi(x+w,y'+h)$. Then for a 4-arrangement $(P_1, P_2)$, $\phi$ \emph{respects} $(P_1, P_2)$ if $\phi(P_1) = \phi(P_2)$. Similarly, for a second-order vertical parallelogram $Q = (P_1, P_2, P_3, P_4)$, define $\phi(Q)$ as $\phi(P_1)-\phi(P_2)-\phi(P_3)+\phi(P_4)$. Then for a second-order 4-arrangement $(Q_1, Q_2)$, $\phi$ respects $(Q_1, Q_2)$ if $\phi(Q_1) = \phi(Q_2)$.
\end{definition}

By \cite[Lemma 3.11]{GM17}, any set $A\subset G\times G$ that respects ``1\%'' of the second-order 4-arrangements has a fairly large subset $A'$ that respects ``99\%'' of the second-order 4-arrangements. We find the set $A'$ through a randomized ``dependent random selection'' process introduced by Gowers and Mili\'cevi\'c, though in our algorithm we have to be careful to make a good choice of randomness at the start so all membership queries to $A'$ that we make give consistent answers.

\begin{theorem} \label{thm:step2.1}
Let $\eta, \delta, \epsilon>0$.
Given $A_1 \subset A_2\subset G\times G$ and $\phi\colon G \times G \to G$, $\phi$ respects at least $\epsilon|G|^{32}$ second-order 4-arrangements in $A_1$. Let $A'$ be a subset of $A_2$ that contains at least $\poly(\eta, \epsilon)|G|^{32}$ second-order 4-arrangements such that the proportion of its arrangements that are respected by $\phi$ is at least $1-\eta$.

Let \texttt{member-A}($u,\delta$) be an algorithm that, with probability at least $1-\delta$, accepts if $u\in A_1$ and rejects if $u\not\in A_2$. Suppose we also have query access to $\phi$.

Then the algorithm \texttt{member-A-prime} makes $O(\poly(\log(\delta^{-1}), \ep^{-1}, \eta^{-1}))$ queries to \texttt{member-A} and $\phi$ and with probability at least $1-\delta$ outputs 1 if $(a,b)\in A'\cap A_1$ and 0 if $(a,b)\not\in A'$.
\end{theorem}

To specify $A'$ we select random elements $\{s_i \}_{i=1}^{k}$ with $s_i \in \F_p$, random $n \times n$ matrices $\{M_i \}_{i=1}^{k}$ with $M_i \in \Mat_n(\F_p)$, and also $\{ r_{(x,y)} \}_{(x,y) \in G \times G}$ with $r_{(x,y)} \in [0,1]$. Specifically, we have that $(x,y) \in A'$ if $r_{(x,y)} \leq 2^{-k} \prod_{i=1}^k \left( 1 + \cos \left( \frac{2 \pi}{p} (\langle s_i, \phi(x,y) \rangle + \langle x, M_iy \rangle ) \right) \right)$. 

\vspace{\algtopskip}
\noindent \fbox{
    \parbox{\textwidth}{
\texttt{weighted-member-A-prime($\phi$, x, y)}:\\
\textbf{Input} membership test for $A$, query access to $\phi$, $(x,y) \in A$\\
\textbf{Output} the probability to choose $(x,y)$ as an element of $A'$
\begin{itemize}
    \item Sample $k$ random elements $s_1, \ldots, s_k$ as well as independent random $n \times n$ matrices $M_1, \ldots, M_k$ over $\mathbb{F}_p$. Return $2^{-k} \prod_{i=1}^k \left( 1 + \cos \left( \frac{2 \pi}{p} (\langle s_i, \phi(x,y) \rangle + \langle x, M_iy \rangle ) \right) \right)$.
\end{itemize}
}
}\vspace{\algbotskip}

Observe that in \texttt{weighted-member-A-prime}, we effectively have an output of a weighted set. To remove this source of randomness, we introduce a certifier for weighted sets $A'$; we can then repeat the selection procedure until we pass the certifier. For a second-order 4-arrangement $\mathcal{Q}$, write $r_{\mathcal{Q}} = \{ r_{(x,y)}: (x,y) \in \mathcal{Q}\}$. 

\vspace{\algtopskip}
\noindent \fbox{
    \parbox{\textwidth}{
    \texttt{certifier-A-prime($\phi, \rho$)}:\\
    \textbf{Input} query access to $\phi \colon  G \times G \to G$, $\rho > 0$\\
    \textbf{Output} verification whether we have suitable guarantees, $\mathcal{R}$
    \begin{itemize}
        \item Sample $\alpha r$ random 32-tuples from $G^{32}$.
        \item Sample $32 \alpha r$ random reals from $[0,1]$ for the $r_{\mathcal{Q}}$ for each of the second-order 4-arrangements $\mathcal{Q}$ corresponding to each 32-tuple. 
        \item Only retain those tuples for which we have $r_{(x,y)} \leq \texttt{weighted-member-A-prime}(x,y)$ for all $(x,y) \in \mathcal{Q}$. If less that $r$ tuples remain, return 0.
        \item Otherwise, let the corresponding second-order 4-arrangements be $\mathcal{Q}_1, \cdots, \mathcal{Q}_r$ and write $\mathcal{R} = \bigcup_i r_{\mathcal{Q}_i} $. Note that each tuple represents a second-order 4-arrangement and therefore can be thought of as two second-order vertical parallelograms, so we can write $Q_i = (P_1^{(i)}, P_2^{(i)})$. 
        \item For each $i \in [r]$, compute $\ell_i = \phi(P_1^{(i)}) - \phi(P_2^{(i)})$ and let $R$ be the number of $i$ such that $\ell_i = 0$. If $R/r \geq 1- \rho$, return 1 and store $\mathcal{R}$. Otherwise, return 0. 
    \end{itemize}
}
}\vspace{\algbotskip}

Putting everything together, we get the desired membership tester for $A'$. 

\vspace{\algtopskip}
\noindent \fbox{
    \parbox{\textwidth}{
    \texttt{member-A-prime(A,$\phi$,x,y,$\eta$)}:\\
    \textbf{Input} membership test for $A$, query access to $\phi$, $(x,y) \in G \times G$, $\eta > 0$\\
    \textbf{Output} 1 if $(x,y) \in A'$ and 0 otherwise with high probability
    \begin{itemize}
        \item Run \texttt{membership-A(f,a,b)} and if the output is 0, return $\perp$.
        \item Run \texttt{certifier-A-prime($\phi$, $3\eta/4$)} $s$ times and if it never returns 1, then return $\perp$. Otherwise, suppose that the choice of random elements when \texttt{certifier-A-prime} first returns $1$ are $s_1, \ldots, s_k, M_1,\ldots, M_k, \mathcal{R}$. 
        \item If $r_{(x,y)} \in \mathcal{R}$ then return 1 if $r_{(x,y)} \leq \texttt{weighted-member-A-prime}(x,y)$ and 0 otherwise.
        \item Otherwise, if $r_{(x,y)} \not \in \mathcal{R}$, then sample a random real $r_{(x,y)} \in [0,1]$ in the process adding $r_{(x,y)}$ to $\mathcal{R}$. As before, return $1$ if $r_{(x,y)} \leq \texttt{weighted-member-A-prime}(x,y)$ and 0 otherwise.
    \end{itemize}
}
}\vspace{\algbotskip}

Note that this membership tester is dynamic, since we update $\mathcal{R}$ as we call \texttt{member-A-prime} on the fly. 

\begin{proof}

We will first prove that by picking the right parameters, we can ensure that with probability at least $1- \delta$ the output of \texttt{certifier-A-prime} has the following property: if \texttt{certifier-A-prime} outputs 1 and stores the corresponding $\mathcal{R}$, then for any possible extension of $\mathcal{R}$ to $\{ r_{(x,y)}\}_{(x,y) \in G \times G}$ obtained by drawing additional random reals from $[0,1]$ when necessary, the set $A'$ corresponding to these choices of $\{ r_{(x,y)}\}_{(x,y) \in G \times G}$, $\{s_i\}_{i=1}^{k}$ and $\{M_i \}_{i=1}^{k}$ has the property that:
\begin{enumerate}
    \item [(a)] $A'$ contains at least $\poly(\eta, \ep) \abs{G}^{32}$ second-order 4-arrangements, and
    \item [(b)] $\phi$ respects at least a $(1 - \eta)$-fraction of these second-order 4-arrangements.
\end{enumerate}

Take $\alpha = \poly(\eta^{-1}, \ep^{-1}, \log(1/\delta)$ and $ r = \poly(\log(1/\delta), \eta^{-1})$, then by Lemma~\ref{lem:CH} if of the $\alpha r$ random tuples we sample, we have retained $r$ of them then with probability at least $1- \delta/2$ we have that $A'$ satisfies property (a). 

Note that at this stage in the algorithm, we may assume that any 32-tuple we work with corresponds to a second-order 4-arrangement with all its constituent elements lying in $A'$. By adjusting the constants, we can guarantee the existence of a set $A''$ that contains $\poly(\eta, \ep) \abs{G}^{32}$ second-order 4-arrangements with $\phi$ respecting at least a $(1 - 3\eta/4)$-fraction of them. Set $\rho = 3\eta/4$. This implies that with probability at least $1- \delta$ if \texttt{certifier-A-prime} returns 1 then the proportion of 4-arrangements that $\phi$ respects in $A'$ is at least $ 1 - \eta$, by a standard Chernoff bound.

The upshot is that with probability at least $1- \delta$ if \texttt{certifier-A-prime} returns 1 then we have the guarantees of (a) and (b); here we know that both can be satisfied simultaneously because of the proof of existence in  \cite[Lemma 3.11]{GM17}. This also ensures that if \texttt{member-A-prime} does not return $\perp$ then it has the guarantees we desire. 

Lastly, we need to check that with high probability \texttt{member-A-prime} does not return $\perp$. To that end we need to calculate the probability that \texttt{certifier-A-prime} returns 1. For a choice of random elements $\{s_i \}_{i=1}^{k}$ with $s_i \in \F_p$, random $n \times n$ matrices $\{M_i \}_{i=1}^{k}$ with $M_i \in \Mat_n(\F_p)$, and also $\{ r_{(x,y)} \}_{(x,y) \in G \times G}$ with $r_{(x,y)} \in [0,1]$, let $X$ be the random variable denoting the number of second-order 4-arrangements that are respected by $\phi$ and let $Y$ be the number of second-order 4-arrangements that are not. We claim that by taking $s = \poly(\log(1/\delta), \ep^{-1}, \eta^{-1})$ we will be able to ensure that with probability at least $1-\delta$ \texttt{member-A-prime} will not return $\perp$. Equivalently, we will prove that $\Pb[X - \eta^{-1} Y \geq 0] \geq \poly(\eta, \ep)$. To that end recall that $\E[X- \eta^{-1} Y] \geq \poly(\eta, \ep) \abs{G}^{32}$. We also know that $X - \eta^{-1} Y$ is bounded above by the number of second-order arrangements respected by $\phi$ in $A \supset A'$ which is in turn at most $\abs{G}^{32}$. In other words, 
\[ \Pb[X - \eta^{-1} Y \geq 0] \cdot \abs{G}^{32} \geq \poly(\eta, \ep) \abs{G}^{32} \]
which is equivalent to the desired claim. 
\end{proof}

In addition to restricting to a set that respects ``99\%'' of second-order 4-arrangements, for technical reasons it is also convenient to restrict our domain to a set where $\phi$ is a Freiman homomorphism on the columns.

The following notation will be convenient. For a set $A \subset G^2$ we will write $A_{\bullet b} = \{ a \in G : (a,b) \in A\}$ and for a function $\phi$ on the domain $G^2$ we will write $\phi_{\bullet b}$ to denote the function $\phi_{\bullet b}(a) = \phi(a,b)$.

\begin{theorem} \label{thm:step2.2}
Let $\alpha, \delta > 0$ and $f\colon G \to \C$ be a 1-bounded function. Given subsets $A_1 \subset A_2 \subset G \times G$ where $A_1$ has density at least $\alpha$ and $\phi \colon G^2 \to G$ such that $\abs{\widehat{\partial_{a,b}f}(\phi(a,b))} \geq 2\sqrt{\alpha}$ for $(a,b) \in A_1$ and $\abs{\widehat{\partial_{a,b}f}(\phi(a,b))} \geq \sqrt{\alpha}$ for $(a,b) \in A_2$, then $A_2$ has a subset $\widetilde{A}$ of density $\Omega(\qpoly(\alpha))$ such that for each $b$ we have that $\phi \bigr|_{\widetilde{A} \cap (G \times \{ b\})}$ is a Freiman homomorphism. 

Let \texttt{member-A}($u,\delta$) be an algorithm that, with probability at least $1-\delta$, accepts if $u\in A_1$ and rejects if $u\not\in A_2$. Suppose we also have query access to $\phi$.

Then there is an algorithm \texttt{member-A-tilde} that makes $O(\poly(\alpha^{-1}, \log(\delta^{-1})))$ queries to \texttt{member-A} and $\phi$ and with probability at least $1 - \delta$ outputs 1 if $(a,b) \in \widetilde{A} \cap A_2$ and 0 if $(a,b) \not \in \widetilde{A}$.
\end{theorem}

\vspace{\algtopskip}
\noindent \fbox{
\parbox{\textwidth}{
\texttt{member-A-tilde(A,$\phi$,a,b)}:\\
\textbf{Input} membership tests for $A_1$ and $A_2$, query access to $\phi \colon G \times G \to G$, $(a,b) \in G \times G$\\
\textbf{Output} 1 if $(a,b) \in \widetilde{A}$ and 0 otherwise with high probability
\begin{itemize}
    \item Using query access to $\phi_{a \bullet}$, execute \texttt{find-affine-map($A, \phi_{a \bullet}$)} with output $T_a$. 
    \item If $\phi(a,b) \neq T_ab$, output 0. Else, output 1.
\end{itemize}
}
}\vspace{\algbotskip}

\begin{proof}
Note that for each $b$
\[
\E_a \mathbf{1}_{A_{\bullet b}}(a) \abs{\widehat{\partial_a(\partial_b f)}(\phi(a,b))}^2 \ge \alpha d(b),
\]
where $d(b)$ is the density of $A_{\bullet b} = \{a \in G : (a,b) \in A_2 \} \subset G$. Then by \cite[Lemma 3.1]{GM17}, there are at least $\alpha^4 d(b)^4 p^{3n}$ quadruples $(x, y, z, w) \in A_2^4$  such that $x+y=z+w$ and $\phi(x,b)+\phi(y,b)=\phi(z,b)+\phi(w,b)$. Therefore by \texttt{find-affine-map}, there is an affine map $T_b$ which agrees with $\phi(\cdot, b)$ on at least $\qpoly(\alpha d(b))$ fraction of $A_{\bullet b}$. Such subset of $A_{\bullet b}$ that agrees with $\phi(\cdot, b)$ is the set $\widetilde{A}$. From $(1/p^n)\sum_{b \in G} d(b) = \alpha$, at least $\alpha/(1+\alpha)$ fraction of $b \in G$ satisfies $d(b) \ge \alpha^2$, which implies that the density of $\widetilde{A}$ is at least $\qpoly(\alpha)$. Also, since for each $b$ the map $\phi$ agrees with an affine map on $\widetilde{A}$, they are Freiman homomorphisms as well. 

For $(a,b) \in G \times G$, if $(a,b) \notin A_2$, output 0. Otherwise, using \texttt{find-affine-map}, output an affine map $T_b$. If $\phi(a,b) = T_b(a)$, output 1, otherwise 0.

The overall algorithm fails when \texttt{find-affine-map} fails, so the algorithm succeeds with probability at least $1-\delta$.
\end{proof}

In the next three steps we take the set $A'$ and refine it to a smaller set that has a lot of additive structure. The first of these steps is called the bilinear Bogolyubov argument which finds affine maps $T_1,\ldots, T_m\colon G\to G$ such that $T_1h,\ldots,T_mh$ captures the large spectrum of a function $f_{\bullet h}$. We will end up applying this the bilinear Bogolyubov argument to a convolution of the indicator function of $A'$ to find additive structure in $A'$.

\begin{theorem}\label{thm:step3.1}
Let $\xi, \gamma, \delta >0$ and $f \colon G \times G \to \C$ be a bounded function. Suppose \texttt{approx-f}($\ep, \delta, x$) is an oracle such that for every $x \in G \times G$ we have with probability at least $1 - \delta$ that $\abs{f(x) - \texttt{approx-f}(\ep, \delta, x)} \leq \ep$. 

Given oracle access to \texttt{approx-f}, there exists an algorithm \texttt{bogo-aff-map} that makes $O(\qpoly(\xi^{-1}, \gamma^{-1}) \cdot \poly(n, \log(\delta^{-1})))$ queries to \texttt{approx-f} and with probability at least $1- \delta$ outputs affine maps $T_1, \ldots, T_m$ such that for all but at most $\xi \abs{G}^2$ points $(h,u) \in \{ (h,u): \abs{\widehat{f_{\bullet h}}(u)}^2 \geq \gamma\}$ we have $T_ih = u$, where $m = O(\qpoly(\gamma^{-1}, \xi^{-1}))$. 
\end{theorem}

The existence of the bi-affine maps $T_1, \cdots, T_m$ with such a property follows from \cite[Lemma 4.10]{GM17}. We next describe an algorithm to identify them. 

\vspace{\algtopskip}
\noindent \fbox{
\parbox{\textwidth}{
    \texttt{bogo-aff-map(f)}: \\
    \textbf{Input} query access to $f\colon G \times G \to \C$\\
    \textbf{Output} explicit descriptions of $T_1, \ldots, T_m$
    \begin{itemize}
    \item Initialize $\mathcal{L} = \emptyset$. We will use $\mathcal{L}$ to store the linear maps $T_i$ that we identify.
    \item Sample $r$ elements from $G$ and call this set $K_1$. For each $h \in K_1$, we can get query access to an estimate of $f_{\bullet h}$ via \texttt{approx-f}($\nu, \omega, x$).
    \item Using the oracle access to this noisy version of $f_{\bullet h}$, run $\texttt{noisy-GL}(f_{\bullet h}, \gamma)$ to retrieve the large Fourier spectrum $L_{\gamma}^h$.
    \item Iterating through $h \in K_1$, if $L_{\gamma}^h = \emptyset$ or $L_{\gamma}^h \subset \mathcal{L}(h) := \{ T_i(h) : T_i \in \mathcal{L} \}$, prune $K_1$ by removing this value of $h$.
    \item If after pruning the number of remaining element in $K_1$ is less than $\xi$, terminate and return $\mathcal{L}$.
    \item Otherwise, let $Q$ be the set of $h \in G$ such that there is an element of $L_{\gamma}^h$ not covered by the existing affine maps in $\mathcal{L}$. For each $h \in Q$, let $\sigma(h)$ be some $u$ such that $u$ is not in the list $L_{\gamma}^h$. Apply \texttt{find-affine-map} to $Q$ and $\sigma$ and add the affine map obtained to $\mathcal{L}$. 
\end{itemize}
}
}\vspace{\algbotskip}

\begin{proof}
We can think of \texttt{bogo-aff-map} as operating in two stages. The first stage which samples $r$ elements and then runs \texttt{noisy-GL($f_{\bullet h}, \gamma$)} is effectively a certifier stage; we verify if the linear maps in $\mathcal{L}$ already has the covering property we desire and terminate the algorithm if it does. Otherwise, $\mathcal{L}$ does not cover the large Fourier spectrum and we can invoke the discussion in \cite[Section 4.6]{GM17} to proceed to the second stage where we generate an additional linear map to add to $\mathcal{L}$. 

For simplicity of notation write $\Sigma_{\gamma} = \{ (h,u): \abs{\widehat{f_{\bullet h}}(u)} \geq \gamma\}$. Let $\tilde{\delta} = O(\delta/\qpoly(\xi^{-1}, \delta^{-1}))$ and observe that $\log \tilde{\delta} = O(\poly(\log(\delta^{-1}), \xi^{-1}))$. Take $r = O(\poly(\xi^{-1}, \log(\delta^{-1})))$, $\nu = \omega = \gamma/10$. Note that if we consider the corresponding $(h,u)$ from the output of \texttt{noisy-GL}($f_{\bullet h}, \gamma$) we obtain a set $\Sigma$ slightly larger than $\Sigma_{\gamma}$. In particular, we have $\Sigma_{\gamma} \subset \Sigma \subset \Sigma_{\gamma/5}$. We will show that the parameters we pick ensure that with probability at least $1 - \tilde{\delta}$ we have $\bigcup_{h}\{ (h,u): u \in \mathcal{L}(h)\}$ covers an at least $1 - \xi$ fraction of $\Sigma$, which by our earlier observation will imply that the same is true for $\Sigma_{\gamma}$. This would then show that when the algorithm terminates it would have the desired guarantees. Indeed, by Lemma~\ref{lem:CH} with probability at least $1 - \tilde{\delta}$ if we have less than $\xi$ elements remaining in $K_1$ after the pruning in the first stage, then for at least $\xi \abs{G}$ values of $h$ such that there exists some $r \in G$ with $\widehat{f_{\bullet h}}(r) \geq \gamma^{1/2}$ and $r \not \in \mathcal{L}(h)$. 

Before we move on to the rest of the proof, note by Parseval's theorem (as in the proof of \cite[Lemma 4.10]{GM17}) since $\Sigma \subset \Sigma_{\gamma/5}$ this algorithm should terminate after at most $\qpoly(\xi^{-1}, \gamma^{-1})$ iterations. 

Next, we study the second stage. We can ensure that \texttt{find-affine-map} succeeds with probability at least $1 - \tilde{\delta}$. However, we also need to ensure that we are able to obtain the oracle inputs in order to apply \texttt{find-affine-map}. To that end we need to check that we can give a polynomial time algorithm for each of the following tasks:
\begin{itemize}
    \item Check for membership in $Q$.
    \item Sample a random element from $Q$.
    \item Query access to $\sigma$. 
\end{itemize}

\vspace{\algtopskip}
\noindent \fbox{
\parbox{\textwidth}{
\texttt{member-Q}(h):
\begin{itemize}
    \item Run \texttt{noisy-GL}($f_{\bullet h}, \gamma$). If the resulting list $L_{\gamma}^h$ is empty, return $0$. 
    \item Otherwise, for each $\ell \in \mathcal{L}$ if $\ell(h) \in L_{\gamma}^h$ remove the corresponding value. At the end of this process, if $L_{\gamma}^h$ is empty, return $0$. Otherwise, return $1$.  
\end{itemize}
}
}\vspace{\algbotskip}

\vspace{\algtopskip}
\noindent \fbox{
\parbox{\textwidth}{
\texttt{query-sigma}(h):
\begin{itemize}
    \item Run \texttt{member-Q} on $h$. If the output is $0$, return $\perp$. 
    \item Otherwise, run \texttt{noisy-GL}($f_{\bullet h}, \gamma$). For each $\ell \in \mathcal{L}$ if $\ell(h) \in L_{\gamma}^h$ remove the corresponding value. Return an arbitrary element from the resulting list. 
\end{itemize}
}
}\vspace{\algbotskip}

\vspace{\algtopskip}
\noindent \fbox{
\parbox{\textwidth}{
\texttt{sampler-Q}:
\begin{itemize}
    \item Sample $s$ elements from $G$ and run \texttt{member-Q} on each of them, returning the first element on which \texttt{member-Q} outputs 1. 
\end{itemize}
}
}\vspace{\algbotskip}

In each of these possibilities we can pick the parameters such that with probability at least $1- \tilde{\delta}$ they achieve the goal of the algorithm. Given the guarantees for \texttt{noisy-GL} this is obviously true for \texttt{member-Q} and \texttt{query-sigma}. For \texttt{sampler-Q}, take $s = \poly(\log(\delta^{-1}), \xi)$. Since $\abs{Q}$ has density at least $\xi$, each of the $r$ elements does not lie in $Q$ with probability at least $1- \xi$. With our chosen parameters, it follows that the probability at least one of the elements that we sample lies in $Q$ is indeed $1 - \tilde{\delta}$. 

Note by our choice of $\tilde{\delta}$ and the fact that each subroutine is run on the order of $O(\qpoly(\xi^{-1}, \delta^{-1}))$ times, we have that overall the algorithm succeeds with probability at least $1- \delta$. We also observe that since $\abs{\mathcal{L}} = \qpoly(\xi^{-1}, \gamma^{-1})$, each of the sub-routines above runs in polynomial time.

Lastly, we can observe that since in each iteration, we have that each sub-routine runs with at most $O(\poly(\gamma^{-1}, n, \xi^{-1}, \log(\delta^{-1})))$ queries to \texttt{approx-f} and as we observed at the beginning, there are at most $O(\qpoly(\xi^{-1}, \gamma^{-1}))$ iterations, it follows that the total number of queries that the algorithm makes to \texttt{approx-f} is bounded by $O(\qpoly(\xi^{-1}, \gamma^{-1}) \cdot \poly(n, \log(\delta^{-1})))$. 
\end{proof}

Next we recall the notion of mixed convolution as defined in \cite{GM17}. For $f_1, f_2, f_3, f_4 \colon G \times G \to \C$, the \emph{mixed convolution} $\lozenge (f_1, f_2, f_3, f_4)$ is defined as
\[
\lozenge (f_1, f_2, f_3, f_4)(w,h)
= \E_{x, y, y'} f_4(x,y) \overline{ f_3(x,y+h)f_2(x+w,y') }f_1(x+w,y'+h).
\]
As a shorthand we will often write $\lozenge (f,f,f,f)$ as $\lozenge f$. We will also denote $[f] = (\langle \lozenge f, \lozenge f \rangle)^{1/8}$. The mixed convolution is a fairly natural notion since for indicator functions $f_1, f_2, f_3, f_4$, note that $\lozenge(f_1, f_2, f_3, f_4)(w,h)$ becomes the proportion of vertical parallelograms whose vertices are in each set corresponding to $f_i$ in a certain order.

\begin{theorem}\label{thm:step3.2}
Let $f\colon  G \times G \to \C$ be a bounded function. Let \texttt{approx-f}($\ep, \delta, x$) be an oracle such that for every $x \in G^2$ we have with probability at least $1 - \delta$ that $\abs{f(x) - \texttt{approx-f}(\ep, \delta, x)} \leq \ep$. Let $T_1, \ldots, T_m$ be affine maps such that for all but at most $\xi \abs{G}^2$ points $(h,u) \in \{ (h,u): \abs{\widehat{f_{\bullet h}}(u)}^2 \geq \gamma\}$ we have $T_ih = u$. Write $F = \lozenge f$.

Given query access to \texttt{approx-f} and also an explicit description of the maps $T_1, \ldots, T_m$, there exists an algorithm \texttt{bohr-aff-map} that makes $O(\exp(\qpoly(\xi^{-1})) \cdot \poly(n, \log(\delta^{-1})) )$ queries to \texttt{approx-f} and with probability at least $1 - \delta$ returns a bi-affine map $\beta\colon  G \times G \to \F_p^k$ such that $\norm{F - \proj_{\beta} F}_2 \leq \xi$ where $k = O(\exp \qpoly(\xi^{-1}))$.
\end{theorem}

We briefly recall the argument in \cite[Theorem 4.15]{GM17}, which establishes the existence of such $\beta$, to motivate our algorithm. The goal here is to find an $L^2$ approximation $\proj_{\beta}F$ of $F$. Given the maps $\{ T_i \}$, a natural choice of such an approximation would be $F'(x,y) = \sum_{i=1}^{m} \widehat{F_{\bullet y}}(T_i y)\omega^{x. T_i y}$. Because of the presence of redundant maps, namely $T_i h = T_j h$ for some $i \neq j$, we lose $L^2$ control easily and need to do one further truncation. We pick out distinct Fourier coefficients via
\[ u_i(y) = \begin{cases} 0 &T_jy = T_i y \text{ for some } j<i, \\ \widehat{F_{\bullet y}}(T_i y) & \text{otherwise} \end{cases}\]
defined for each $i$, and then do one more round of approximation by picking out the large Fourier coefficients of $u_i$. In particular, suppose the list of large Fourier coefficients for $u_i$ is given by $K_i = \{ v_{i1}, \cdots, v_{ik_i} \}$ then if we consider $\beta_{ij}(x,y) = x.T_iy + v_{ij}y$, the bi-affine map we desire is given by 

\[ \beta(x,y) =(\beta_{11}(x,y), \cdots, \beta_{1k_1}(x,y), \cdots, \beta_{m1}(x,y), \cdots, \beta_{mk_m}(x,y)).\]

In the following sub-routines, whenever we need to query $f$ we will use the oracle access to \texttt{approx-f} to estimate $f$. 

\vspace{\algtopskip}
\noindent \fbox{
\parbox{\textwidth}{
    \texttt{box($\phi$,w,h)}:\\
    \textbf{Input} query access to $\phi \colon  G \times G \to G$, $w, h \in G$\\
    \textbf{Output} estimate of $\lozenge \phi(w,h)$
    \begin{itemize}
    \item Sample $3s$ values $\{x_i\}_{i=1}^s, \{y_i\}_{i=1}^s, \{y_i'\}_{i=1}^s$ and output
    \[
    \dfrac{1}{s}\sum_{i=1}^s \phi(x_i,y_i)\overline{ \phi(x_i, y_i+h) \phi(x_i+w, y_i')}\phi(x_i+w, y_i'+h).
    \]
\end{itemize}
}
}\vspace{\algbotskip}

\vspace{\algtopskip}
\noindent \fbox{
\parbox{\textwidth}{
    \texttt{bogo-u($\mathcal{L}$,f,i,y)}:\\
    \textbf{Input} query access to affine maps $T_1, \ldots, T_m$ in $\mathcal{L}$, query access to $f$, integer $1 \le i \le m$, $y \in G$\\
    \textbf{Output} estimate of $u_i(y)$
    \begin{itemize}
        \item Iterate through $j=1, \ldots, i-1$ and if $T_iy = T_jy$ then return 0.
        \item Otherwise, using \texttt{box(f,w,h)} to get a query access to $\lozenge f$, sample $r$ values $\{x_i\}_{i=1}^r$ from $G$ and return
        \[
        \dfrac{1}{r} \sum_{i=1}^r \lozenge f(x_i,y) \omega^{-x_i\cdot T_iy}.
        \]
\end{itemize}
}
}\vspace{\algbotskip}

\vspace{\algtopskip}
\noindent \fbox{
\parbox{\textwidth}{
    \texttt{bohr-aff-map(f)}:\\
    \textbf{Input} query access to $f\colon G\times G \to G$\\
    \textbf{Output} explicit expression of the bi-affine map $\beta$
    \begin{itemize}

        \item Using \texttt{bogo-u} to get a query access to each $u_i$, run \texttt{noisy-GL($u_i, \zeta$)} and let the output be $L_i=\{ v_{i1}, \ldots, v_{ik_i} \}$.
        \item For each $i=1, \ldots, m$ and $j=1, \ldots, k_i$, let $\beta_{ij}(x,y) = x.T_iy + v_{ij}y$ and return
        \[
        \beta(x,y) = ( \beta_{11}(x,y), \ldots, \beta_{1k_1}(x,y), \ldots, \beta_{m1}(x,y), \ldots, \beta_{mk_m}(x,y))
        \]
\end{itemize}
}
}\vspace{\algbotskip}

\begin{proof}
Let $\tilde{\delta} = O(\delta/m)$ where $m = \qpoly(\xi^{-1})$ and also let $\zeta = \xi^2/(m^22^m)$. Suppose we are able to obtain query access to some $u_i'$ such that with probability at least $1 - \zeta/10$ we have $\norm{u_i - u_i'}_{\infty} \leq \zeta/10$, then for each $i$ we can ensure with probability at least $1 - \tilde{\delta}$ that the output of \texttt{noisy-GL}($u_i', \zeta$) is a list $L_i$ with the property that $\Spec_{\zeta}(u_i) \subset L_i \subset \Spec_{\zeta/5}(u_i)$. In particular, if we write $w_i(y) = \sum_{v \in L_i} \widehat{u_i}(v) \omega^{v.y}$ then by H\"older's inequality and \cite[Corollary 4.13]{GM17}, we have that $\norm{u_i - w_i}_2 = O(\zeta)$ with probability at least $1 - \tilde{\delta}$. In particular, if we set $H(x,y) = \sum_{i=1}^{m}w_i(y) \omega^{x \cdot T_iy}$ and $H'(x,y) =  \sum_{v \in \{ T_1y, \ldots, T_m y\}}\widehat{(\lozenge f)_{\bullet y}} \omega^{x.v}$ then $\norm{H - H'}_2 \leq \xi/2$ with probability at least $1 - \delta$. Combining with \cite[Lemma 4.14]{GM17} which states that $\norm{H' - F}_2 \leq \xi/2$, by the triangle inequality we have that $\norm{F - H}_2 \leq \xi$ with probability at least $1 - \delta$. 

Furthermore, by Parseval's identity and the fact that $L_i \subset \Spec_{\zeta/5}(f)$ (analogous to the proof of \cite[Theorem 4.15]{GM17}), it follows that $\abs{L_i} = O(m^22^{2m}/\xi^2)$. In particular, if we can obtain query access to some $u_i'$ as described via \texttt{bogo-u} then with probability at least $1- \delta$ we get a bi-affine map with the desired bound on its codimension and also $\norm{F - \proj_{\beta}F}_2 \leq \xi$. It suffices to check that we can indeed obtain such query access to $u_i'$. 

To that end, we consider the approximations given by \texttt{bogo-u}. They come in three stages. First, we can approximate $f$ by \texttt{approx-f}($ \zeta/30, \zeta/30 , \cdot$) which has the property that $\norm{f - \texttt{approx-f}(\zeta/30, \zeta/30, \cdot)}_{\infty} \leq \zeta/30$ with probability at least $1 - \zeta/30$. Second, by taking $s = O(\poly(\zeta^{-1}))$ in \texttt{box} we can ensure that $\norm{\lozenge f - \texttt{box}(f, \cdot )}_{\infty} \leq \zeta/30$ with probability at least $1 - \zeta/30$. Third, by taking $r = O(\poly(\zeta^{-1}))$ in \texttt{bogo-u} and assuming that we have (perfect) query access to $\lozenge f$, we can ensure that with probability at least $1 - \zeta/30$ \[\norm{ \E_x \talloblong f (x,y) \omega^{x. T_iy} \cdot \mathbf{1}(y: i = \min \{j : T_j y = T_iy \}) - \texttt{bogo-u}(\cdot)}_{\infty} \leq \zeta/30.\] It follows that we can approximate each $u_i$ via \texttt{bogo-u} up to an additive error of at most $\zeta/10$ with probability at least $1 - \zeta/10$, as desired.

Lastly, we check the runtime guarantees of the algorithm. By Theorem~\ref{thm:step3.1}, we have that \texttt{bogo-aff-map} makes $O(\qpoly(\zeta^{-1}) \cdot \poly(n, \log(\delta^{-1})))$ queries to \texttt{approx-f}. Note that the overall number of queries that the sub-routine \texttt{bogo-u} makes to \texttt{approx-f} is $O(\poly(\zeta^{-1})) \cdot m$. Next, the application of \texttt{noisy-GL} in \texttt{bohr-aff-map} makes $O(\poly(n, \zeta^{-1}, \log(\delta^{-1})))$ queries to \texttt{bogo-u}. In summary, it follows that we make $O(\qpoly(\zeta^{-1}) \cdot \poly(n, \log(\delta^{-1}))) = O(\exp(\qpoly(\xi^{-1})) \cdot \poly(n, \log(\delta^{-1})) )$ queries to \texttt{approx-f}. 
\end{proof}

Now we introduce a crucial concept, the rank of bi-affine map. Later, we will get a quasirandomness property from a high rank bi-affine map. The dimension of a bi-affine map is the dimension of its range.

\begin{definition}
For a one-dimensional bi-affine map $\beta \colon  G \times G \to \F_p$, if we write it as $\beta(x,y) = x.Ty+x.A+B.y+C$ for $T \in \Mat_n(\F_p)$ and column vectors $A, B, C$, then the \emph{rank} of $\beta$ is defined to be the rank of $T$. For a bi-affine map $\beta \colon  G \times G \to \F_p^k$, the \emph{rank} of $\beta$ is the least rank of any one-dimensional bi-affine map $(x,y) \mapsto u.\beta(x,y)$ for nonzero $u \in \F_p^k$.
\end{definition}

\begin{theorem}\label{thm:step4.1}
Given an explicit representation of a bi-affine map $\beta\colon  G \times G \to \F_p^k$ and $t \in \Z^{+}$, then there exists an algorithm that runs in time $O(kp^k + k^4)$ that outputs a basis for $X_0$ and $Y_0$ such that $\dim X_0, \dim Y_0 \leq tk$ with the property that the corresponding Bohr decomposition has rank at least $t$.
\end{theorem}

We first introduce a certifier of sorts for whether our decomposition has achieved the desired high rank condition. The algorithm \texttt{linear-translate} outputs a value of $u$, if it exists, such that $(x,y) \mapsto u.\beta(x,y)$ has rank at most $t$.

\vspace{\algtopskip}
\noindent \fbox{
\parbox{\textwidth}{
    \texttt{linear-translates($\beta$)}:\\
    \textbf{Input} explicit expression of a bi-affine map $\beta \colon G \times G \to \F_p^k$\\
    \textbf{Output} $\perp$ if the rank of $\beta$ is greater than $t$, $u \in \F_p^k$ if $u.\beta(x,y)$ is of rank at most $t$
    \begin{itemize}
    \item For each of the $p^k$ possibilities of $u \in \mathbb{F}_p^k$ compute the rank of $u.\beta(x,y)$ and output any choice of $u$ for which this value is at most $t$. If no such choice of $u$ exists, output $\perp$.
\end{itemize}
}
}\vspace{\algbotskip}

Using this certifier, we can then iteratively prune our space to identify the desired $X_0$ and $Y_0$. 

\vspace{\algtopskip}
\noindent \fbox{
\parbox{\textwidth}{
 \texttt{bohr-decomp($\beta$,t)}:\\
 \textbf{Input} explicit expression of a bi-affine map $\beta \colon G \times G \to \F_p^k$, a positive integer $t$\\
 \textbf{Output} basis for $X_0$ and $Y_0$
    \begin{itemize}
        \item If the output of \texttt{linear-translates($\beta$)} is $\perp$ then output $X$ and $Y$. Otherwise, if the output of \texttt{linear-translates($\beta$)} is $u$, then we can compute a basis $b_1, \cdots, b_{\ell}$ for $\langle u \rangle^{\perp}$. 
        \item Writing $P = [b_1 \cdots b_{\ell}]$ we have that the projection $\pi_u$ to  $\langle u \rangle^{\perp}$ is given by $P(P^{T}P)^{-1}P^T$. Replace $\beta$ by $\pi \circ \beta$ by composing the appropriate matrix and repeat from the first step. 
    \end{itemize}
}
}\vspace{\algbotskip}

\begin{proof}
Since \texttt{linear-translate} brute forces through all possibilities of $u \in \F_p^k$, it runs in time $p^k$. Computing a basis for $\langle u \rangle^{\perp} \subset \F_p^k$ takes time at most $O(k^3)$. Therefore each iteration of the loop in \texttt{bohr-decomp} takes time $O(p^k + k^3)$. By \cite[Lemma 5.1]{GM17}, \texttt{bohr-decomp} terminates after at most $k$ iterations for a total runtime of $O(kp^k + k^4)$ and has the guarantees we desire. 
\end{proof}

From Theorem \ref{thm:step4.1}, we now find a Bohr decomposition of $\beta \colon G \times G \to G$ as follows: for each $(v,w,z) \in X_0 \times Y_0 \times G$, define $B_{v, w, z}$ be a level set $\{(x,y) \in G \times G : x|_{X_0} = v, y|_{Y_0} = w, \beta(x,y)=z\}$. In particular, we call such a Bohr decomposition a \emph{bilinear Bohr decomposition}. We define the rank of a bilinear Bohr decomposition as the smallest rank of $B_{v, w, z}$ for each $(v,w,z) \in X_0 \times Y_0 \times G$.

In the next theorem, $\mathcal{A}$ is the group algebra of $G$ and $\Sigma(\mathcal{A})$ is the subset of $\mathcal{A}$ consisting of elements the sum of whose coefficients is 1. We can think of $\Sigma(\mathcal{A})$ as a technical object that allows us to describe the spread of values taken by $\phi(P)$ as we vary $P$ in the family of vertical parallelograms of width $w$ and height $h$. Specifically, $\Sigma(\mathcal{A})$ in some sense corresponds to a probability distribution over the values of $\phi(P)$. Since our goal is to extract the bilinear part of $\phi$, we would ideally want $\phi(P)$ to be constant across vertical parallelograms $P$ of same width and height, which in turn corresponds to the probability distribution of $\phi(P)$ being close to a delta distribution. This naturally leads us to the following notion of a $(1-\eta)$-bihomomorphism, as given in \cite{GM17}.

\begin{definition}
Given a non-negative function $\mu\colon G \times G \to \mathbb{R}$, $\phi \colon G \times G \to \mathcal{A}$, and a constant $0\le \eta \le 1$, $\phi$ is a \emph{$(1-\eta)$-bihomomorphism with respect to $\mu$} if
\[
\E_{w,h} \norm{ \E_{P \in \mathcal{P}(w,h)} \mu(P)\phi(P) }_2^2
\ge (1-\eta)\E_{w,h}\abs{\E_{P \in \mathcal{P}(w,h)}\mu(P)}^2,
\]
where $\mathcal{P}(w,h)$ is the set of vertical parallelograms whose width and height are $w$ and $h$, respectively. Furthermore, if $P$ is a vertical parallelogram whose vertices are $(x,y), (x, y+h), (x+w,y'), (x+w,y'+h)$, then $\mu(P)$ and $\phi(p)$ are given as follows:
\begin{align*}
    \mu(P) &= \mu(x,y)\mu(x,y+h)\mu(x+w,y')\mu(x+w,y'+h)\\
    \phi(P) &= \phi(x,y)\phi(x,y+h)^\ast\phi(x+w,y')^\ast\phi(x+w,y'+h),
\end{align*}
where for $a = \sum_{g \in G} c_gg \in \mathcal{A}$, $a^\ast = \sum_{g} \overline{c_g}(-g)$.
\end{definition}

Recall the shorthand of $[f] = (\langle \lozenge f, \lozenge f \rangle)^{1/8}$. One way to interpret $[f]$ is that it quantifies whether $f(P)$ depends highly on width and height of a random vertical parallelogram $P$. 

\begin{theorem}\label{thm:step4.2}
Suppose there is a Bohr decomposition of a bi-affine map $\beta\colon G \times G \to \F_p^k$ of rank $t$ and codimension $k$ with corresponding Bohr sets $\{ B_{v,w,z} \}$. Let $\mu$ and $\xi$ be functions taking values on $[0,1]$ that are constant on each $B_{v,w,z}$. Let $\phi\colon G \times G \to \Sigma(\mathcal{A})$ be a $(1 - \eta)$-bihomomorphism with respect to $\mu$. Suppose also that $\E \xi = \zeta$. Suppose $0 < \gamma \leq \eta[\mu]^8/8$ and $p^{-t} \leq \eta p^{-9k}/8$. Then there exists $(v,w,z)$ such that $\phi$ is a $(1- 4 \eta)$-bihomomorphism with respect to $\mathbf{1}_{B_{v,w,z}}$, the value of $\mu$ on $B_{v,w,z}$ is at least $\eta [\mu]^8$ and the value of $\xi$ on $B_{v,w,z}$ is at most $\gamma^{-1} \zeta$. 

Suppose we have query access to the probability distribution $\phi\colon G \times G \to \Sigma(\mathcal{A})$. For any $\delta > 0$, suppose we have query access to $\mu_{\delta}'\colon G \times G \to [-\ep_1(\delta),1+\ep_1(\delta)]$ which for each $(x,y) \in G \times G$ satisfies $\abs{\mu(x,y) - \mu'(x,y)} \leq \ep_1(\delta)$ with probability at least $1- \delta$ and also query access to $\xi_{\delta}'\colon G \times G \to [-\ep_2(\delta), 1+\ep_2(\delta)]$ which for each $(x,y) \in G \times G$ satisfies $\abs{\xi'(x,y) - \xi(x,y)} \leq \ep_2(\delta)$ with probability at least $1- \delta$. Suppose we have an explicit representation of $\beta$ as well as basis for the corresponding $X_0, Y_0$ in the Bohr decomposition. Then there is an algorithm \texttt{high-rk-bohr-set} running in time $O(\poly(\log(\delta^{-1}), p^{r+s}, p^{k}, \ep_1^{-1}))$ that with probability at least $1- \delta$ outputs $v,w,z$ corresponding to a Bohr set $B_{v,w,z}$ which satisfies the following properties:

\begin{itemize}
        \item $\phi$ is a $(1- 5\eta)$-bihomomorphism with respect to $\mathbf{1}_{B_{v,w,z}}$,
        \item $\mu(x) \geq [\mu]^8/2 - \rho_1(\delta)$ for any $x \in B_{v,w,z}$ where $\rho_1(\delta) = 3 \ep_1(\delta/3)$, and
        \item $\xi(x) \leq \gamma^{-1} \zeta + \rho_2(\delta)$ for any $x \in B_{v,w,z}$, where $\rho_2(\delta) = 3 \ep_2(\delta/3)$.
    \end{itemize}
\end{theorem} 

The existence of such a Bohr set $B_{v,w,z}$ follows from \cite[Theorem 5.8]{GM17}. Algorithmically, we will go over all possible values of $v,w,z$ and run a certifier on each possibility. 

Note that in order to compute $[\phi] = \langle \lozenge \phi, \lozenge \phi \rangle$, we can use the fact that $\lozenge \phi$ is the distribution obtained by evaluating $\phi$ at a randomly chosen vertical parallelogram. 

\vspace{\algtopskip}
\noindent \fbox{
\parbox{\textwidth}{
\texttt{box-dist(f)}:
\begin{itemize}
    \item Sample 3 elements $x,y,y'$ uniformly at random from $G$ which corresponds to a random vertical parallelogram $P = \{ (x,y), (x,y+h), (x+w, y'), (x+w, y'+h) \}$ and output $f(P)$ where
    \[
    f(P) = f(x,y)\overline{f(x,y+h)f(x+w,y')}f(x+w,y'+h)
    \]
\end{itemize}
}
}\vspace{\algbotskip}

Moreover, $\langle f, g \rangle = \Pb_{a \sim f, v \sim g}[a = b]$ which we can therefore approximate by sampling some $\{ f_i \}$ and $\{g_i \}$ according to the distributions $f$ and $g$ respectively, and returning ${(\#\{ i: f_i = g_i \})}/{(\# \text{ of samples})}$. 

\vspace{\algtopskip}
\noindent \fbox{
\parbox{\textwidth}{
\texttt{inner-product-dist(f,g)}: \# assuming that we have sample access to the probability distributions $f,g$.\\
\textbf{Input} oracle accesses to probability distributions $f, g$\\
\textbf{Output} estimate of $\langle f, g \rangle$
\begin{itemize}
    \item Sample $r$ elements according to the distribution $f$, call them $a_1, \cdots, a_r$. Similarly, sample $r$ elements according to the distribution $g$ and call them $b_1, \cdots, b_r$.
    \item Output the fraction of $i$ such that $a_i = b_i$. 
\end{itemize}
}
}\vspace{\algbotskip}

\vspace{\algtopskip}
\noindent \fbox{
\parbox{\textwidth}{
\texttt{sq-brac-dist(f)}:
\begin{itemize}
    \item Execute \texttt{inner-product-dist(\texttt{box-dist}(f), \texttt{box-dist}(f))}. 
\end{itemize}
}
}\vspace{\algbotskip}

Next, we will introduce a primitive for estimating $\langle \lozenge f, \lozenge f \rangle$. 

\vspace{\algtopskip}
\noindent \fbox{
\parbox{\textwidth}{
\texttt{sq-brac(f,g)}:
\begin{itemize}
    \item We approximate $\lozenge f$ by $\widetilde{f}$ and $\lozenge g$ by $\widetilde{g}$. Sample $3s_1$ values $\{x_i \}_{i=1}^{s_1}, \{y_i\}_{i=1}^{s_1}, \{y_i' \}_{i=1}^{s_1}$ and let \[ \widetilde{f}(w,h) = \frac{1}{r} \sum_{i=1}^{r} f(x_i,y_i) \overline{f(x_i, y_i+h)f(x_i+w, y_i')} f(x_i+w, y_i'+h)\]
    and similarly for $\widetilde{g}(w,h)$.
    \item Sample $s_2$ pairs $(w_i,h_i) \in G \times G$ and return $\frac{1}{s_2} \sum_{i=1}^{s_2} \widetilde{f}(w_i,h_i)\widetilde{g}(w_i,h_i)$.
\end{itemize}
}
}\vspace{\algbotskip}

Because the number of Bohr sets is within a tolerable bound, we can enumerate all possibilities of the Bohr set and it suffices to output a Bohr set with each of the three properties we desire. We build such a certifier in the following algorithm. 

\vspace{\algtopskip}
\noindent \fbox{
\parbox{\textwidth}{
\texttt{high-rk-bohr-set($X_0, Y_0, \phi, A', \beta$)}:\\
\textbf{Input} basis for $X_0, Y_0$, query access to $\phi \colon G \times G \to G$, membership test for $A'$, explicit expression of a bi-affine map $\beta \colon G \times G \to \F_p^k$\\
\textbf{Output} $(v, w, z) \in X_0 \times Y_0 \times G$ such that $B_{v, w, z}$ has suitable properties.
\\ Suppose $\dim X_0 = r$ and $\dim Y_0 = s$.
\begin{itemize}
    \item For each of the $p^{tk} = O(\eta^{-1})$ possible choices for each of $v \in X_0$ and $ w \in Y_0$ as well as the $p^k$ possible choices for $z$ (for a total of $p^{2tk+k} = O(\eta^{-1})$ choices for the triple $(v,w,z)$), run each of the following tests.
    \item \textbf{Test A}: 
    \begin{itemize}
        \item Execute \texttt{sq-brac-dist($\phi$)} and let its output be $\ell$. Return 1 if $\ell > (3/2 - 5 \eta)p^{-3r-5s}(p^{-7k} - 4p^{2k-t})$. 
    \end{itemize}
    \item \textbf{Test B}: 
    \begin{itemize}
        \item We execute \texttt{sampler($B_{v,w,z}, 1, G^2$)} to select an element $x$ from $B_{v,w,z}$. 
        \item Estimate $[\mu]$ via \texttt{sq-brac($\mu_{\delta}'$, $\mu_{\delta}'$)} and let the output be $R$. Return 1 if $\mu_{\delta}'(x) \geq R - \rho_1/2$.
    \end{itemize}
    \item \textbf{Test C}:
    \begin{itemize}
        \item For the value of $x$ in Test B, return 1 if $\xi'(x) \leq \gamma^{-1}\zeta + \rho_2/2$.
    \end{itemize}
    \item If the output for all three tests above is 1, return the corresponding value of $v,w,z$. 
\end{itemize}
}
}\vspace{\algbotskip}

\begin{proof}
By taking $r = O(\poly(\log(\delta^{-1}), p^{r+s}, p^{k}))$ in \texttt{inner-product-dist}, we are able to estimate $\ell$ to within an additive error of at most $ \frac{1}{2}p^{-3r-5s}(p^{-7k} - 4p^{2k-t})$ with confidence $1 - \delta$ by invoking Lemma~\ref{lem:CH}. In particular this means that with probability at least $1 - \delta$, we have $[\phi]^8 \geq (1-5 \eta) p^{-3r-5s}(p^{-7k} - 4p^{2k-t})$. By \cite[Lemma 5.6]{GM17}, we have that $[\mathbf{1}_{B_{v,w,z}}]^8 \leq p^{-3r-5s}(p^{-7k}+4p^{2k-t})$. As a consequence of our choice of parameter $t$ it follows that $[\phi\mathbf{1}_{B_{v,w,z}}]^8 \geq (1-5\eta)[\mathbf{1}_{B_{v,w,z}}]^8$. That is, if \textbf{Test A} returns 1 then with probability at least $1- \delta$ we have that $\phi$ is a $(1-5\eta)$-bihomorphism with respect to $\mathbf{1}_{B_{v,w,z}}$. 

Let $\rho_1 = 3 \ep_1(\delta)$. For simplicity of notation we make the dependence on $\delta$ implicit and write $\mu_{\delta}' = \mu'$ and $\ep_1(\delta) = \ep_1$. By taking $s_1 = O(\poly(\log(\delta^{-1}), \ep_1^{-1}))$ in \texttt{sq-brac} we can estimate $\lozenge \mu'$ by $\widetilde{\mu'}$ to within an additive error of at most $\ep_1(\delta)/4$ with confidence $1 - \delta/4$. Further, by taking $s_2 = O(\poly(\log(\delta^{-1}), \ep_1^{-1}))$ in \texttt{inner-product} we can estimate $[\widetilde{\mu'}] $ to within an additive error of at most $\ep_1(\delta)/4$ with confidence $1 - \delta/4$, so overall we will able to estimate $[\mu']$ to within an additive error of $\ep_1(\delta)/2$ with confidence $1 - \delta/2$. By our assumptions on $\rho'$ approximating $\rho$, this means that if \textbf{Test B} returns 1, then with probability at least $1- \delta$ we have that $\mu(x) \geq [\mu]^8/2 - \rho_1$ for any $x \in B_{v,w,z}$ (recall that $\mu$ is constant on $B_{v,w,z}$). 

Lastly, let $\rho_2 = 2 \ep_2(\delta)$. For simplicity of notation we write $\xi_{\delta}' = \xi'$. Since $\xi'$ approximates $\xi$ to an additive error of at most $\ep_2(\delta)$ with confidence $1 - \delta$, it follows that if \textbf{Test C} returns 1 then with probability at least $1- \delta$ we have that $\xi(x)\leq \gamma^{-1}\zeta + \rho_2$.
\end{proof}

When we introduced the notion of a $(1-\eta)$-bihomomorphism, we said that we want the probability distributions we care about to be very close to delta distributions. Therefore, we need a definition of distance between two probability distributions.

\begin{definition}
For $\phi, \psi \in \Sigma(\mathcal{A})$, the \emph{distance} between $\phi$ and $\psi$, denoted $d(\phi, \psi)$, is $1-\langle \phi, \psi \rangle$.
\end{definition}
It makes sense to call this notion a distance since it satisfies the triangle inequality.

\begin{theorem}\label{thm:step5}
Let $k,t > 0$ be integers. Let $\psi\colon G \times G \to \Sigma(\mathcal{A})$ be a $(1-\eta)$-bihomomorphism on a high-rank bilinear Bohr set $B$ defined by a bi-affine map $\beta$ with codimension $k$ and rank $t$, and write $B'' = \{ (w,h) : \beta(w,h) =0\}$. Then there exists $\widetilde{B} \subset B''$ of density $1- \rho$ and $\widetilde{\psi}\colon B''\to G$ such that $d \left(\psi(w,h), \delta_{\widetilde{\psi}(w,h)} \right) \leq 64 \eta  p^{-3k}$ for $(w,h)\in\widetilde B$ Here we take $\rho = 16 \eta^4p^{-16 k}$.

Given query access to the probability distribution $\psi$ and an explicit description for $\beta$, there exists an algorithm \texttt{query-tilde-psi} that makes $O(\poly(\eta^{-1}, p^{k})\log(\delta^{-1}))$ queries to $\psi$ and with probability at least $1- \delta$ outputs $\widetilde{\psi}(a,b)$ for $(a,b) \in \widetilde{B}$ and has no guarantees otherwise. 
\end{theorem}

Define $\psi' = \lozenge(\mathbf{1}_B\psi)$. The $\lozenge$ operator can be interpreted as forming a probability distribution by sampling a random vertical parallelogram. Putting this in another way, we can sample from $\psi'$. 

\vspace{\algtopskip}
\noindent \fbox{
\parbox{\textwidth}{
\texttt{psi-prime($\phi$,$B$,w,h)}:
\begin{itemize}
    \item Sample $x,y,y'\in G$ uniformly at random and repeat until $(x,y), (x,y+h), (x+w,y'),(x+w,y'+h) \in B$. 
    \item Return the product \texttt{psi($\phi$,x,y)}\texttt{psi($\phi$,x,y+h)}${}^\ast$\texttt{psi($\phi$,x+w,y')}${}^\ast$ \texttt{psi($\phi$,x+w,y'+h)}.
\end{itemize}
}
}\vspace{\algbotskip}

In order to identify $\widetilde{\psi}$, we will effectively be doing a majority vote. 

\vspace{\algtopskip}
\noindent \fbox{
\parbox{\textwidth}{
\texttt{query-tilde-psi($\phi$,$B''$,w,h)}:
\begin{itemize}
    \item If $(w,h) \not \in B''$ return $\perp$. Else, execute \texttt{psi-prime($\phi$,$B''$,$w$,$h$)} for $r$ times and return the most popular value among these $r$ values. 
\end{itemize}
}
}\vspace{\algbotskip}

\begin{proof}
Take $r = O(\log(\delta^{-1}))$ in \texttt{query-tilde-psi}. By assumption, if $(w,h) \in \widetilde{B}$, we have that $\Pb[\psi(w,h) = \widetilde{\psi}(w,h)] \geq 1 - 64 \eta p^{-3k}>3/4$. By Lemma~\ref{lem:CH}, if we let the number of samples for which $\psi(w,h) = \widetilde{\psi}(w,h)$ be $R$, then it follows that  $\Pb\left [R < \frac{r}{2} \right] \leq \exp(-O(r))\leq\delta$.
This implies that the majority vote output of \texttt{query-tilde-psi} is with probability at least $1-\delta$ the value $\widetilde{\psi}(w,h)$ for $(w,h) \in \widetilde{B}$ as desired.
\end{proof}

In the next step we use a bilinear Bohr set $B''=\{(x,y)\in G\times G: \beta''(x,y)=0\}$ where $\beta''\colon G\times G\to\F_p^k$ is a bi-linear map with rank at least $10k$ and $G=\F_p^n$. Note that each column $B''_{w\bullet}=\{(w,y)\in B''\}$ is of the form $\{w\}\times V$ where $V$ is a linear subspace of $G$ of dimension between $n-k$ and $n$. The same is true for the rows $B''_{\bullet h}$.

\begin{theorem}\label{thm:step6}
Let $\delta>0$, $\epsilon < p^{-3k}/1000$ and let $\beta''\colon G \times G\to\F_p^k$ be a bi-linear map with rank at least $10k$. Define $B''=\{(x,y)\in G \times G\colon \beta''(x,y)=0\}$. Given an explicit description of $\beta''$ and query access to $\widetilde{\psi}\colon  B'' \to G$ and the guarantee that $\widetilde{\psi}$ is additive in each variable on an $(1- 5\epsilon^2)$-fraction of elements of $B''$, the algorithm \texttt{bi-affine} makes $O(\poly(n,p^k, \log(\delta^{-1})))$ queries to \texttt{query-tilde-psi} and with probability at least $1- \delta$ outputs a bi-affine map $T \colon G^2 \to G$ that agrees with $\widetilde{\psi}$ on an $(1- 15\epsilon p^{2k})$-fraction of elements of $B''$.
\end{theorem}

Roughly, the idea is as follows: the first step is to obtain a 90\% subset $\widetilde{B} \subset B''$ such that $\widetilde{\psi}_{\bullet h}$ is additive on $\widetilde{B}_{\bullet h}$ if $\widetilde{B}_{\bullet h} \neq \emptyset$ and $\widetilde{\psi}_{w \bullet}$ is additive on $\widetilde{B}_{w \bullet}$ if $\widetilde{B}_{w \bullet} \neq \emptyset$. The bi-affine map $T'$ will agree with $\widetilde\psi$ on $\widetilde B$. Then applying \cite[Lemma 6.24]{GM17} we are able to extend $T'$ row-wise and then column-wise uniquely to a function $T'$ additive in each variable and defined on the entire of $B''$. 

The next step is to extend the domain from $B''$ to $G^2$. This will require making some choices. We will make a good choice of $w$, and then extend $T'$ to the column $G_{w \bullet}$ arbitrarily. The choice of $w$ is good in the sense that for any $(x,y)$ we are able to find $(x_1, y_1), (x_2, y_2), \ldots, (x_7, y_7) \in G_{w \bullet} \cup B''$ such that these 8 points form a 4-arrangement. Call such a 4-arrangement \textit{good}. We may now define the bi-affine map $T'$ in the unique way that respects these good 4-arrangements.

\begin{proof}
Call a column $B''_{w \bullet}$ \textit{good} if $\tilde\psi\bigr|_{B''_{w \bullet}}$ is additive on a subset of large relative density in $B''_{w \bullet}$. 

First, we will check if $B''_{w \bullet}$ is good. 

\vspace{\algtopskip}
\noindent \fbox{
\parbox{\textwidth}{
\texttt{many-additive-triples-col($B'', (w,h)$)}:
\begin{itemize}
    \item Use \texttt{sampler($B''_{w \bullet}, k, \F_p^n$)} to obtain $(w,x_1), \ldots, (w,x_{k})$. If at least $(1 - 3\epsilon)k$ of them satisfy $\widetilde{\psi}(w,h) = \widetilde{\psi}(w,h - x_i) + \widetilde{\psi}(w,x_i)$, return 1. Otherwise, return 0. 
\end{itemize}
}
}\vspace{\algbotskip}

\vspace{\algtopskip}
\noindent \fbox{
\parbox{\textwidth}{
\texttt{is-col-good($B'', w$)}:
\begin{itemize}
    \item Use \texttt{sampler($B''_{w \bullet}, r, \F_p^n$)} to obtain $(w,y_1), \ldots, (w,y_{r})$. For each $y_i$, run \texttt{many-additive-triples-col($B'', (w,y_i)$)}. If at least $(1-2\epsilon)r$ of them output 1, return 1. Otherwise, return 0. 
\end{itemize}
}
}\vspace{\algbotskip}

\begin{claim}\label{claim:is-col-good}
There exists a choice of $r, k = O(\poly(\log(\delta^{-1})))$ such that \texttt{is-col-good($B'', w$)} has the following guarantees. If there is a subset of $B''_{w \bullet}$ of density at least $1-\epsilon$ such that $\tilde\psi$ is additive on this subset then the algorithm outputs 1 with probability at least $1-\delta$ ($w$ is good) and if $\tilde\psi$ is not additive on any subset of $B''_{w \bullet}$ of density at least $1 - 5\epsilon$ in $B''_{w \bullet}$, then the algorithm outputs 0 with probability at least $1-\delta$ ($w$ is bad).
\end{claim}

\begin{proof}
If $w$ is good, let $M_w\subset B''_{w \bullet}$ be a set of density at least $1 - \epsilon$ such that $\tilde\psi$ is additive on $M_w$. For $(w,h) \in M_w$ we have that with probability $1 - \delta/(2r)$, \texttt{many-additive-triples-col($B'', (w,h)$)} outputs 1. This is because as long as $(w,x) \in M_w$ and $(w,h-x)\in M_w$ then $\widetilde{\psi}(w,h) = \widetilde{\psi}(w,h - x) + \widetilde{\psi}(w, x)$. Both of these events occurs with probability at least $1 - \epsilon$ for a random choice of $x\in B''_{w\bullet}$ so \texttt{many-additive-triples-col($B'', (w,h)$)} outputs 1 with the desired probability by a Chernoff bound.

Consequently, in \texttt{is-col-good}, since the density of $M_w$ in $B''_{w \bullet}$ is at least $1 - \epsilon$, with probability $ 1 - \delta/2$ we sample $(1 - 2\epsilon)r$ elements of $B''_{w \bullet} \cap B_1$ and with probability at least $1 - \delta/2$ \texttt{many-additive-triples-col} evaluates to 1 on each of them. 

Now we show the converse. Let $M_w \subset B''_{w \bullet}$ be such that for each $(w,h) \in M_w$, there exists a corresponding $ N_h \subset B''_{w \bullet}$ of density at least $1 - 4 \epsilon$ with $\widetilde{\psi}(w,h) = \widetilde{\psi}(w,h - x) + \widetilde{\psi}(w,x)$ for any $x \in N_h$. We will show that if the density of $M_w$ in $B''_{w \bullet}$ is less than $1 - 5\epsilon$, then \texttt{is-col-good} outputs 0 with probability at least $1-\delta$. First note that with probability $1 - \delta/2$ the fraction of the samples taken in \texttt{is-col-good} from $M_w$ is at most $1 - 4 \epsilon$. Furthermore, for $(w,h)\not\in M_w$, the probability that \texttt{many-additive-triples-col} outputs 1 on $(w,h)$ is at most $1-\delta/(2r)$ since the density of $N_h$ is less than $1 - 4 \epsilon$. Thus we have shown that if the density of $M_w$ in $B''_{w \bullet}$ is less than $1 - 5\epsilon$ then \texttt{is-col-good} outputs 0 with probability at least $1-\delta$.

Lastly, we show that $\widetilde{\psi}$ is additive on $M_w$ which will complete the proof.

To do so we need to check that if $(w,h_1), (w,h_2), (w,h_1 + h_2) \in M_w$ then $\widetilde{\psi}(w,h_1) + \widetilde{\psi}(w,h_2) = \widetilde{\psi}(w, h_1 + h_2)$. Consider $N_{h_1}, N_{h_2}$ and $N_{h_1 + h_2}$ as defined above, which each have density at least $1-4 \epsilon$ in $B''_{w \bullet}$. Thus, $S=N_{h_1} \cap (h_1-N_{h_1+h_2}) \cap (h_2+N_{h_2})$ has density at least $1-12 \epsilon$ in $B''_{w \bullet}$. For $x\in S$ we have $x\in N_{h_1}$ and $h_1-x\in N_{h_1+h_2}$ and $x+h_2\in N_{h_2}$.

Our next goal is to find $x,y$ such that $x\in S$ and $y\in M_w$ and $x+y\in M_w$ and $x\in N_{x+y}$ and $-x\in N_{y}$. The first equation fails to hold for $12 \epsilon$ fraction of pairs $(x,y)$, the second and third for $5\epsilon$ fraction each and the fourth and fifth for at most $9\epsilon$ fraction each. These sum to less than 1, so we can find a pair $x,y$ satisfying the above conditions. For this pair we have 
\[ \begin{cases} \widetilde{\psi}(w,h_1) = \widetilde{\psi}(w,h_1 - x) + \widetilde{\psi}(w,x), \\ \widetilde{\psi}(w,h_2) = \widetilde{\psi}(w,x+h_2) + \widetilde{\psi}(w,-x), \\ \widetilde{\psi}(w,h_1 + h_2) = \widetilde{\psi}(w,h_1 -x) + \widetilde{\psi}(w,h_2+x), \\
\widetilde{\psi} (w, x+y) = \widetilde{ \psi} (w,x) + \widetilde{\psi} (w,y), \\ \widetilde{\psi} (w,y) = \widetilde{\psi}(w,-x) + \widetilde{\psi}(w,x+y). \end{cases} \]
These imply that $\widetilde{\psi}(w,h_1) + \widetilde{\psi}(w,h_2) = \widetilde{\psi}(w,h_1 + h_2)$, as desired. 
\end{proof}

Similarly, we can define the concept of a good row and furnish a tester \texttt{is-row-good}.

Once we certified that a column and row are both good, we then need to check if $(w,h)$ ``is a good element''. Define $C_{w,h}=\{(w,x)\in B_{w\bullet}'':\tilde\psi(w,h)=\tilde\psi(w,h-x)+\tilde\psi(w,x)\}$. Define $C_w\subset B''_{w\bullet}$ to be the set of $(w,h)$ such that $C_{w,h}$ has relative density at least $1 - 2\epsilon$ in $B''_{w\bullet}$. Analogously, define $R_{w,h}$ and $R_h$ for the rows.

Say that $B''_{w\bullet}$ is a good column if $C_w$ has relative density at least $1 - \epsilon$ in $B''_{w\bullet}$. Analogously define the notion of a good row. Say that $(w,h)$ is a good cell if $B''_{w\bullet}$ is both a good column and a good row, and $(w,h)\in C_w \cap R_h$.

Conversely, $B''_{w \bullet}$ is a bad column if $C_w$ has relative density at most $1 - 5\epsilon$ in $B''_{w \bullet}$. Similarly define the notion of a bad row. Say that $(w,h)$ is a bad cell if either $B''_{w \bullet}$ is a bad column or $B''_{w \bullet}$ is a bad row or if $C_{w,h}$ has relative density less than $1 - 4 \epsilon$ in $B''_{w\bullet}$ or if $R_{w,h}$ has relative density less than $1 - 4 \epsilon$ in $B''_{\bullet h}$. Note that a cell that is not good is not necessarily bad; the bad cells form a subset of the cells that are not good. 

\vspace{\algtopskip}
\noindent \fbox{
\parbox{\textwidth}{
\texttt{is-cell-good($B'', (w,h)$)}:
\begin{itemize}
    \item If \texttt{is-col-good($B'', w$)} returns 0, return 0. Otherwise proceed.
    \item If \texttt{is-row-good($B'', w$)} returns 0, return 0. Otherwise proceed.
    \item If \texttt{many-additive-triples-col($B'', (w,h)$)} returns 1, return 1. Otherwise proceed. 
    \item If \texttt{many-additive-triples-row($B'', (w,h)$)} returns 1, return 1. Otherwise, return 0.
\end{itemize}
}
}\vspace{\algbotskip}

\begin{claim}\label{claim:is-cell-good}

\texttt{is-cell-good} has the following guarantees. If $(w,h)$ is a good cell, then with probability at least $1- \delta$, \texttt{is-cell-good} returns 1. If $(w,h)$ is a bad cell, then with probability at least $1-\delta$, \texttt{is-cell-good} returns 0. 
\end{claim}

\begin{proof}
If $(w,h)$ is a good cell, then with probability at least $1 - \delta/4$, \texttt{is-col-good($B'',w$)} and \texttt{is-row-good($B'',w$)} both return 1 by a Chernoff bound. Since $(w,h) \in C_w$, with probability $1 - 5\epsilon$ a randomly sampled $(w,x) \in B''_{w \bullet}$ satisfies $\widetilde{\psi}(w,h) = \widetilde{\psi}(w, h- x) + \widetilde{\psi}(w, x)$. Consequently, by a Chernoff bound, with probability $1 - \delta/4$ the output from the second step is 1 as well.

If $(w,h)$ is a bad cell and $B_{w \bullet}''$ is a bad column, then with probability at least $1 - \delta/4$, \texttt{is-col-good($B'',w$)}  and \texttt{is-row-good($B'',w$)} both return 0 by a Chernoff bound. Otherwise, if $(w,h)$ is a bad cell with $C_{w,h}$ being not dense enough, then with probability at least $1 - \delta/4$ in the third step of the algorithm it would return 0. If $(w,h)$ is a bad cell with $R_{w,h}$ being not dense enough, then with probability at least $1 - \delta/4$ in the fourth step of the algorithm it would return 0.
\end{proof}

We will prove that there is a unique map $T\colon B''\to G$ that is additive in each variable and agrees with $\tilde\psi$ on the good cells. First, we show that the restriction of $\widetilde{\psi}$ on the good cells is additive in both variables. The proof of Claim~\ref{claim:is-col-good} shows that for each good column $B''_{w\bullet}$ the function $\widetilde{\psi}|_{C_w}$ is additive in the first variable. Similarly, for each good row $B''_{\bullet h}$ the function $\widetilde{\psi}|_{R_h}$ is additive in the second variable. This shows that $\tilde\psi$ is additive in both variables when restricted to the good cells.

We claim that most columns in $B''$ are good. By assumption, $\tilde\psi$ is additive on a subset $S$ of $B''$ of density at least $1-5\epsilon^2$. We show that if a column $B''_{w \bullet}$ is not good, then $S \cap B''_{w \bullet}$ has relative density at most $1-\epsilon$ in $B''_{w \bullet}$. This is because $C_w\supseteq S\cap B''_{w\bullet}$. By Markov and the fact the all column of $B''$ have the same size up to a factor of $p^k$, it follows that the fraction of columns that are good is at least $1- 5 \epsilon p^{k}$. Similarly, we get the same bound for good rows.

Next, we claim that the fraction of cells in $B''$ which are not good is at most $4\epsilon p^k + 10\epsilon p^{2k}$. To show this, observe that a cell is not good if it satisfies one of the following four conditions:
\begin{itemize}
    \item It lies in a column that is not good. 
    \item It lies in a row that is not good. 
    \item It lies in a good column $B_{w \bullet}''$ but does not lie in $C_w$.
    \item It lies in a good row $B_{\bullet h}$ but does not lie in $R_h$.
\end{itemize}
The fraction of cells lying in a not good column or not good row is at most $10 \epsilon p^{2k}$ from our earlier calculations. Since the relative densities of $B_{w \bullet}'' \backslash C_w$ and $B_{\bullet h} \backslash R_h$ are each at most $2 \epsilon p^k$, it follows, by our restriction on $\epsilon$, that at most $4\epsilon p^k + 10\epsilon p^{2k} < p^{-k}/56$ fraction of cells in $B''$ are not good. 

We now show that a large fraction of columns have a large fraction of good cells. By a similar Markov argument as before, and once again using the fact that the column of $B''$ have the same size up to a factor of $p^k$, we have that at least a $ 1 - 8p^k(4\epsilon p^k + 10\epsilon p^{2k}) > 6/7 $ fraction of columns have relative density of the good cells at least $7/8$. 

We may first apply \cite[Lemma 6.23]{GM17} to extend $\widetilde{\psi}$ uniquely to be defined on all the cells in good columns. Our choice of parameters also allows us to apply \cite[Lemma 6.24]{GM17} to ensure that this extension preserves additivity in the first variable. Using the same argument again, we may apply \cite[Lemma 6.23]{GM17} again to extend $\widetilde{\psi}$ uniquely to obtain a map $T$ additive in each variable defined on the entire of $B''$. 

Before proceeding, we make one more observation. Let $T': B'' \to G$ be the unique map that is additive in each variable and agrees with $\widetilde{\psi}$ on the cells that are not bad. Since the set of cells that are not bad is a superset of the set of good cells, to show that $T'$ exists it suffices to prove that $\widetilde{\psi}$ is additive when restricted to cells that are not bad and an analogous argument as before would allow us to extend the map to the rest of $B''$. Let $M_w = \{ (w,x) \in B_{w \bullet}'': C_{(w,x)} \text{ has density at least }1 - 4 \epsilon\}$ and $N_h = \{ (x,h) \in B_{\bullet h}'': R_{(x,h)} \text{ has density at least }1 - 4 \epsilon\}$. Note that the set of not bad cells is a subset of $M_w \cap N_h$. The proof of Claim~\ref{claim:is-col-good} shows that $\widetilde{\psi}|_{M_w}$ is additive in the first variable and $\widetilde{\psi}|_{N_h}$ is additive in the second variable. It follows that $\widetilde{\psi}$ is additive in both variables on the not bad cells. 

However, by the uniqueness of $T$, it follows that $T' = T$. In particular, it follows that if $(x,y)$ is not a bad cell, then $T(x,y) = \widetilde{\psi}(x,y)$. 

Next, we will show that we may algorithmize this process to retrieve query access on $T$ with high probability. 

\begin{claim}
There exists a choice of  $r,s=O(\poly\log(\delta^{-1}))$ such that the algorithm \texttt{T} has the following guarantee. For each $(w,h)\in B''$, with probability at least $1-\delta$, \[\texttt{T}(B'',\tilde\psi,(w,h))=T(w,h).\]
\end{claim}

\vspace{\algtopskip}
\noindent \fbox{
\parbox{\textwidth}{
\texttt{T($B'',\tilde\psi, (x,y)$)}:
\begin{itemize}
    \item If \texttt{is-cell-good($B'', (x,y)$)} returns 1, then return $\tilde\psi(x,y)$. Otherwise proceed.
    \item If \texttt{is-col-good($B'', x$)} returns 1, then sample $y_1,\ldots,y_r\in B''_{x\bullet}$. For each $i$, if \texttt{is-cell-good($B'', (x,y_i)$)} and \texttt{is-cell-good($B'', (x,y-y_i)$)} both return 1, then return $\tilde\psi(x,y_i)+\tilde\psi(x,y-y_i)$. Otherwise, proceed.
    \item Finally, sample $x_1,\ldots,x_s\in B''_{\bullet y}$. For each $i$, if \texttt{is-col-good($B'', x_i$)} and \texttt{is-col-good($B'', x-x_i$)} both return 1, then return \texttt{T($B'',\tilde\psi, (x_i,y)$)}+\texttt{T($B'',\tilde\psi, (x-x_i,y)$)}.
\end{itemize}
}
}\vspace{\algbotskip}

\begin{proof}
For a given $(w,h) \in B''$, if it is a good cell, then by Claim~\ref{claim:is-cell-good} the first step returns 1 with probability at least $1 - \delta$ and $\texttt{T}(B'', \widetilde{\psi}, (w,h)) = T(w,h) = \widetilde{\psi}(w,h)$ when this happens. 

If $(w,h)$ is a bad cell in a good column, then we may guarantee with probability at least $1 - \delta/2$ that in step 1 \texttt{is-cell-good($B'',(w,h)$)} returns 0, while in step 2 \texttt{is-col-good($B'', (w,h) $)} returns 1. By choosing $r = O(\poly(\log(\delta^{-1})))$, since the density of good cells in a good column is at least $6/7$, we may guarantee that with probability at least $1 - \delta/2$ that there exists $(w,y_i)$ such that \texttt{is-cell-good($B'', (w,y_i)$)} and \texttt{is-cell-good($B'', (w,y-y_i)$)} both return 1. Consequently, with probability $1 - \delta$, we have $\texttt{T}(B'', \widetilde{\psi}, (w,h)) = T(w,h) = \widetilde{\psi}(w, y_i) +  \widetilde{\psi}(w, y - y_i)$.

If $(w,h)$ is neither a bad nor a good cell in a good column, we do not have any guarantees on the application of \texttt{is-cell-good} in the first step. By our arguments before the claim, if \texttt{is-cell-good} returns 1, then $\texttt{T}(B'', \widetilde{\psi}, (w,h)) = T(w,h) = \widetilde{\psi}(w,h)$. If \texttt{is-cell-good} returns 0, then we may use the same analysis for the second step as in the previous paragraph.

Lastly, if $(w,h)$ is a bad cell in a bad column, then we may guarantee with probability $ 1- \delta/2$ that in step 1 \texttt{is-cell-good($B'',(w,h)$)} returns 0 and in step 2 \texttt{is-col-good($B'', (w,h)$)} returns 0. By choosing $r = O(\poly(\log(\delta^{-1})))$, since the density of good columns is at least $6/7$, we may guarantee that with probability at least $1 - \delta/2$ that there exists $(w,y_i)$ such that \texttt{is-col-good($B'', (w,y_i)$)} and \texttt{is-col-good($B'', (x,y-y_i)$)} both return 1. By our earlier arguments, we may guarantee with probability at least $1 - \delta/2$ that $\texttt{T}(B'', \widetilde{\psi}, (w, y_i)) = T(w, y_i)$ and $\texttt{T}(B'', \widetilde{\psi}, (w, y - y_i)) = T(w, y - y_i)$. Taken together, we may ensure with probability at least $1 - \delta$ that then $\texttt{T}(B'', \widetilde{\psi}, (w,h)) = T(w,h) = \widetilde{\psi}(w,h)$.
\end{proof}

In the next stage, we extend the domain of $T$ from $B''$ to (almost) all of $G^2$. We start with query access to $T\colon B''\to G$ via $\texttt{T}(B'',\tilde\psi,(x,y))$. Then we extend the domain in stages.

First, we identify $x_1$ such that $\beta(x_1, \cdot)$ is full rank. Note that for a choice of $x$, we can certify if $\beta(x, \cdot)$ is full rank. In fact, \cite[Lemma 5.3]{GM17} ensures that a random $x \in G$ has this property with probability at least $1 - \poly(p^{-k})$. So by sampling $O(\log(\delta^{-1}))$ values from $G$ we may ensure that we have at least one candidate for $x_1$.

We now extend the domain of $T$ to $B''\cup \{x_1\} \times G$. We make many arbitrary choices in this step. First we choose a basis $h_1,\ldots,h_t$ for $B''_{x_1\bullet}$. We query $\texttt{T}(B'', \widetilde{\psi}, (x_1,0))$ as well as $\texttt{T}(B'', \widetilde{\psi}, (x_1,h_i))$ for each $i$, and we may ensure with probability $1 - \delta$ that each of these queries agrees with $T$. Then we extend to a basis $h_1,\ldots,h_n$ of $G=\F_p^n$ in an arbitrary way. Finally we define $T''(x_1,h_i)=0$ for $i>t$. This defines the affine function \[T''\left(x_1,\sum_{i=1}^n c_ih_i\right) =T(x_1,0)+\sum_{i=1}^n c_i (T(x_1,h_i)-T(x_1,0)).\]

Using Claim~\ref{claim:is-cell-good}, observe that with probability $1 - \delta$, we may ensure that $T$ agrees with $\widetilde{\psi}$ on the good cells. We have already shown that the density of good cells is at least $ 1 - 15 \epsilon p^{2k}$. Consequently, in order to complete the proof of the theorem it suffices to demonstrate how to extend the domain of $T$ to the entirety of $G^2$.

In \cite[Section 6]{GM17}, Gowers and Mili\'cevi\'c show that there is a unique extension of $T$ from $(\{x_1\}\times G)\cup B''$ to all of $G^2$. We give an algorithm that gives query access to this unique extension on a large fraction of $G^2$. We aim to obtain query access to $T$ on $L = \{ (x,y) \in G^2:  \beta'(x, \cdot) \text{ has full rank}\}$. \cite[Lemma 5.3]{GM17} ensures that $L$ is at least a $1- p^{-9k}$ fraction of $G^2$. For any $(x,y) \in L$, we find a 4-arrangement containing $(x,y)$ and only elements from $B'' \cup G_{x_1 \bullet}$. Then because $T$ has to respect this 4-arrangement and since we have already specified the values of $T$ on $B'' \cup G_{x_1 \bullet}$, we will be able to recover the value of $T(x,y)$ uniquely. 

Now we discuss how to find the desired 4-arrangement. By the full rank condition on $\beta'(x, \cdot)$, we can solve the linear equation specified by $\beta'(x, h) = \beta(x,y)$ for $h$. Let $L_1 = \{ w \in G: \beta(w + x_1, \cdot) \text{ has full rank}\}$ and $L_2 = \{ w \in G: \beta(w + x, \cdot) \text{ has full rank}\}$. By \cite[Lemma 5.3]{GM17}, a random $w \in G$ satisfies $w \in L_1 \cap L_2$ and $\beta'(x_1+w, h) = 0$ with probability at least $1 - 2p^{-8k} - p^{-k}$. Sample $O(\log(\delta^{-1}))$ values from $G$ to ensure that we have at least one candidate for $w$ and we may certify to identify this candidate precisely. For this choice of $w$ and $h$, note that we have $\beta'(x - w, h) = \beta'(x_1 + w,h) = 0$.

Next, by the full rank conditions on $\beta(w+x_1, \cdot)$ and $\beta(w+x, \cdot)$, we may solve the ensuing linear equations to find $y_2$ and $y_3$ such that $\beta(x_1 +w, y_2) = \beta(x-w, y_3) = 0$. Let $y_1$ be an arbitrary element of $G$. Our choice of parameters ensures that $\beta(x_1 +w, y_2 +h) = \beta(x_1 +w, y_2) = \beta(x-w, y_3) = \beta(x-w, y_3 +h) = \beta(x, y - h) = 0$. In particular, we have that $(x_1, y_1), (x_1, y_1 + h), (x_1 + w, y_2 +h), (x_1 + w, y_2), (x-w, y_3), (x-w, y_3 +h), (x, y -h) \in B'' \cup G_{x_1 \bullet}$. These points, together with $(x,y)$, give the desired 4-arrangement.

Lastly, we use this query access to give an explicit description of the bi-affine map $T$. We sample a basis for $G$ that the algorithm succeeds on and then query at each point in the basis. Sample $t = O(\poly(n, \log(\delta^{-1})))$ values from $G$. Prune these values by discarding those $x$ such that $\beta(x, \cdot)$ is not full rank. With probability $ 1- \delta$, we retain $\Omega(n^2 + \log(\delta^{-1}))$ points $w_1, \ldots, w_t$. The number of subspaces of $G$ is $O(\exp(n^2))$ and the probability that a random $w_i$ lies in a specific $(n-1)$-dimension subspace is $p^{-1}$. Consequently, $\Pb[\dim \Span( w_1, \ldots, w_t) \le n] = O(\exp(n^2 - t)) \leq \delta$ and we may extract a basis $b_1, \ldots, b_n$ for $G$ among $w_1, \ldots, w_t$. Using the values for $T$ at $\{ (b_i, b_j) \}$, we may output an explicit description for $T$.
\end{proof}

\begin{theorem}\label{thm:step7} 
Let $c,\delta>0$. Given a query access to a bounded $f\colon  G \to \C$ and an explicit description of a bi-affine map $T\colon  G \times G \to G
$ such that $\E_{a,b}\abs{\widehat{\partial_{a,b}f}(T(a,b))}^2 \geq c$, there exists an algorithm \texttt{find-cubic} that makes $O(\poly(n,1/c,\log(1/\delta)))$ queries to $f$ and with probability at least $1- \delta$ outputs a cubic $\kappa$ with the guarantee that $\abs{\E_x f(x) \omega^{\kappa(x)}} \geq \qpoly(c)$.
\end{theorem}
Let the bilinear part of $T$ be $T^L$ and and define $\tau(x,y,z) = T^L(x,y) \cdot z$. Let $\kappa(x) = \tau(x,x,x)$.

\vspace{\algtopskip}
\noindent \fbox{
    \parbox{\textwidth}{
\texttt{find-cubic(T,f)}:\\
\textbf{Input} explicit expression of a bi-affine map $T$, query access to $f \colon G \to \C$\\
\textbf{Output} a cubic polynomial $\kappa$
\begin{itemize}
    \item Using the formulas as described above, we can obtain $\kappa(x)$. In turn this provides us with query access to $g(x)$.
    \item Run \texttt{find-quadratic(g)} and let the output be $q(x)$.
    \item Return $\kappa+q$. 
\end{itemize}
}
}
\vspace{\algbotskip}

\begin{proof}
Given an explicit description of $T^L$, we can get an explicit representation of the trilinear form $\tau(a,b,c) = T^L(a,b) \cdot c$. In turn we are able to obtain an explicit description of $\lambda(x) = \tau(x,x,x)$. We can combine Lemma 11.1 with the remarks at the end of section 10 of \cite{GM17} to obtain that $g(x) = f(x) \omega^{-\kappa(x)}$ has large $U^3$ norm: $\norm{g}_{U^3} = \Omega(c)$. We finish by invoking Theorem~\ref{thm:u3inv}. The algorithm guarantees that with probability at least $1 - \delta$ it outputs a quadratic form $q$ with $\abs{\E_x f(x) \omega^{(\lambda +q)(x)}} = \abs{\E_x g(x) \omega^{q(x)}} \geq \qpoly(c)$. Since $\kappa = \lambda + q$ is a cubic, we indeed obtain the guarantees we claim. 
\end{proof}

\subsection{Putting everything together} \label{sec:put-it-together}
In this section, we will see how the theorems we have proven so far fit together. We will prove a version of Theorem~\ref{thm:alg-u4-inverse} with a bound of $\eta^{-1} = \exp \exp \qpoly (\ep^{-1})$ as given by \cite{GM17} instead of $\eta^{-1} = \exp \qpoly (\ep^{-1})$. In the next section, we will give quantitative improvements to \cite{GM17} by removing an $\exp$ in the bounds on $\eta^{-1}$, which then leads to Theorem~\ref{thm:alg-u4-inverse}. We (re)state the version of the algorithmic $U^4$ inverse theorem that we will prove in this section.

\begin{theorem}[algorithmic $U^4$ inverse theorem with weaker bounds on $\eta^{-1}$]
\label{thm:main-but-weaker}
Given a prime $p \geq 5$ and $\delta, \epsilon>0$, set $\eta^{-1}= \exp \exp\qpoly(\epsilon^{-1})$. There is an algorithm, which, given a bounded function $f\colon\mathbb F_p^n\to\mathbb C$ that satisfies $\|f\|_{U^4}\geq\epsilon$, makes $O(\poly(n, \eta^{-1}, \log(\delta^{-1})))$ queries to $f$ and, with probability at least $1-\delta$, outputs a cubic polynomial $P\colon\mathbb F_p^n\to\mathbb F_p$ such that 
\begin{equation*}
    |\E_x f(x)\omega^{-P(x)}|>\eta.
\end{equation*}
\end{theorem}

\begin{proof}
Beginning with query access to $f$ with $\norm{f}_{U^4} \geq \ep$, apply Theorem~\ref{thm:step1} and run the corresponding algorithm with parameters $(a,b) \in A_1$ if $\norm{\widehat{\partial_{a,b}f}}_{\infty} \geq \ep^{16}$ to get \texttt{member-A} which with probability at least $1-\delta/8$ outputs 1 if $(a,b) \in A_1$ and 0 if $(a,b) \not \in A_2$. Since $\norm{f}_{U^4}^{16} = \E_{a,b}\norm{\partial_{a,b}f}_{U^2}^4$, it follows by averaging that there exists a set $A \subset G^2$ of density $\Omega(\ep^{16})$ such that $ \norm{\widehat{\partial_{a,b}f}}_{4}^4 = \norm{\partial_{a,b}f}_{U^2}^4 = \Omega(\ep^{16})$. Since $\norm{\widehat{\partial_{a,b}f}}_{U^2}^4 \leq \norm{\widehat{\partial_{a,b}f}}_{2}^2 \norm{\widehat{\partial_{a,b}f}}_{\infty}^2 \leq \norm{\widehat{\partial_{a,b}f}}_{\infty}^2$, it follows that for each $(a,b) \in A$ we have that $\norm{\widehat{\partial_{a,b}}f}_{\infty} = \Omega(\epsilon^8)$. In particular, this argument shows that the density of $A_1$ is $\Omega(\poly(\epsilon))$.

We also have \texttt{query-phi} which in $O(\poly(n, \ep^{-1}, \log(\delta^{-1})))$ queries to $f$ outputs $\phi(a,b)$ with the desired properties with probability at least $1- \delta/8$. 

Now, we want to pass from the implicit 1\% structure on $\phi$ to 99\% structure for $\phi$. Using \texttt{member-A} as well as \texttt{query-phi} as primitives, apply Theorem~\ref{thm:step2.2} to get a membership tester \texttt{member-A-tilde} for a subset $\tilde{A} \subset A$ for which $\phi\bigr|_{\tilde{A} \cap (G \times \{b \})}$ is a Freiman homomorphism. Since the density of $\tilde{A}$ is $\Omega(\poly(\ep))$, it follows by \cite[Lemma 3.7]{GM17} that $\phi$ respects a $\Omega(\poly(\ep))$ fraction of 4-arrangements in $A_1$, which in turn implies by \cite[Corollary 3.9]{GM17} that $\phi$ respects a $\Omega(\poly(\ep))$ fraction of second-order 4-arrangements in $A_1$. Now, using \texttt{member-A-tilde} as a primitive for \texttt{approx-f} in Theorem~\ref{thm:step2.1}, we get a membership tester \texttt{member-A-prime} for a subset $A' \subset \tilde{A}$ such that $A'$ contains $\poly(\eta, \epsilon)\abs{G}^{32}$ second-order 4-arrangements and $\phi$ respects a $1- \eta$ fraction of these. We will set $\eta = 10^{-10}$. A back-of-the-envelope calculation will show that this choice of $\eta$ is sufficiently small for future use. 

It is more convenient to now work with $\psi = \lozenge \phi$. By \cite[Lemma 4.1]{GM17}, we have that $\psi$ is a $(1-\eta)$-bihomomorphism with respect to $\mu = \lozenge \mathbf{1}_{A'}$. Before proceeding further, we describe how to:
\begin{itemize}
    \item Obtain query access to $\psi$ given query access to $\phi$.
    \item Estimate $\mu$ given \texttt{member-A-prime}.
\end{itemize}
Note that $\psi$ can be interpreted as a probability distribution given by $\phi(P)$ for a uniformly random vertical parallelogram $P$ with width $w$ and height $h$. Here, if $P = ((x,y), (x,y+h), (x+w,y'), (x+w,y'+h))$, then $\phi(P) = \phi(x,y)\phi(x,y+h)^\ast\phi(x+w,y')^\ast\phi(x+w,y'+h)$. This allows us to gain query access to $\psi$.

\vspace{\algtopskip}
\noindent \fbox{
\parbox{\textwidth}{
\texttt{psi($\phi$,w,h)}:\\
\textbf{Input} query access to $\phi \colon G \times G \to \mathcal{A}$, $w, h \in G$\\
\textbf{Output} estimate of $\psi(w,h)$
\begin{itemize}
    \item Sample $3r$ values $\{x_i\}_{i=1}^r$, $\{y_i\}_{i=1}^r$, $\{y_i'\}_{i=1}^r$ from $G$ such that $(x_i, y_i), (x_i, y_i+h), (x_i+w,y_i'), (x_i+w,y_i'+h) \in A'$. Then return
    \[
    \dfrac{1}{r} \sum_{i=1}^r \phi(x,y)\phi(x,y+h)^\ast \phi(x+w,y')^\ast\phi(x+w,y'+h).
    \]
\end{itemize}
}
}\vspace{\algbotskip}

The membership tester \texttt{member-A-prime} implies we have query access to $\mathbf{1}_{A'}$. 

Next, we recall a sub-routine that we first introduced in the previous section. 

\vspace{\algtopskip}
\noindent \fbox{
\parbox{\textwidth}{
\texttt{box(f,w,h)}:
\begin{itemize}
    \item Sample $3r$ values $\{x_i \}_{i=1}^{r}, \{y_i\}_{i=1}^{r}, \{y_i' \}_{i=1}^{r}$ and output\[ \frac{1}{r} \sum_{i=1}^{r} f(x_i,y_i) \overline{f(x_i, y_i+h)f(x_i+w, y_i')} f(x_i+w, y_i'+h).\] 
\end{itemize}
}
}\vspace{\algbotskip}

Assume that the output of \texttt{member-A-prime} satisfies the guarantees. Take $r = O(\poly(\tau^{-1}, \log(\xi^{-1})))$, then it follows from Lemma~\ref{lem:CH} that \texttt{box($\mathbf{1}_{A'}$,w,h)} gives query access to $\mu'(\nu, \omega)$ such that $\norm{\mu-\mu'(\nu, \omega)}_{\infty} \leq \omega$  with probability at least $1 - \nu$.

Next, we will obtain some structure on the underlying set first by applying bilinear Bogolyubov and then passing to a suitable high-rank bilinear Bohr set. Using \texttt{box($\mathbf{1}_{A'}$,w,h)} as a primitive in Theorem~\ref{thm:step3.1}, with $O(\exp(\qpoly(\xi^{-1})) \cdot \poly(n, \log(\delta^{-1})))$ queries to \texttt{box($\mathbf{1}_{A'}$,w,h)}, with probability $ 1 - \delta$ we can retrieve explicit descriptions for $T_1, \ldots, T_m$ as in Theorem~\ref{thm:step3.1}. With explicit descriptions $T_1, \ldots, T_m$ and \texttt{box($\mathbf{1}_{A'}$,w,h)}, we have by Theorem~\ref{thm:step3.2} with probability at least $1- \tilde{\delta}$ we can obtain an output of a bi-affine map $\beta\colon G \times G \to \F_p^k$ with $\norm{F - \proj_{\beta}F}_2 \leq \xi$. In particular, note that Theorem~\ref{thm:step3.2} gives an explicit description of $\beta$. We now pass to a high rank bilinear Bohr set. Let $t = \lceil 60k + 9 \log (\eta^{-1}) \rceil$. Apply Theorem~\ref{thm:step4.1} using the explicit description of $\beta$ to get a basis for $X_0, Y_0$ with $\dim X_0, \dim Y_0 \leq tk$ such that the corresponding Bohr decomposition has rank at least $t$. To pass down to one specific bilinear Bohr set, we will apply Theorem~\ref{thm:step4.2}. In the specific context of our application we set:
\begin{itemize}
    \item $\phi := \psi/\norm{\psi}_1 $ where as we recall that $\psi = \lozenge \phi$.
    \item $ \mu := \lozenge \mathbf{1}_{A'}$.
    \item Write $B_{v,w,z}^{(x,y)}$ for the Bohr set that $(x,y)$ lies in; then $\xi(x,y) := \E_{(\widetilde{x},\widetilde{y}) \in B^{(x,y)}_{v,w,z}}[\lozenge \mathbf{1}_{A'}(\widetilde{x},\widetilde{y}) - \proj_{\beta} \lozenge \mathbf{1}_{A'} (\widetilde{x},\widetilde{y})]^2$.
\end{itemize}
We need to check that we have the primitives that Theorem~\ref{thm:step4.2} requires. First, we have sample access to the probability distribution via \texttt{psi}. By using \texttt{box}, it follows via Lemma~\ref{lem:CH} we can approximate $\lozenge \mathbf{1}_{A'}$ to arbitrary additive precision with arbitrarily high probability. Lastly, since $\xi$ is an expected value, we can approximate it as follows.

\vspace{\algtopskip}
\noindent \fbox{
\parbox{\textwidth}{
\texttt{xi(x,y)}:
\begin{itemize}
    \item We can iterate through all $O(\eta^{-1})$ choices for the triple $(v,w,z)$ to identify $B := B_{v,w,z}^{(x,y)}$. 
    \item Execute \texttt{sampler}($B, t, G^2$) and let its output be $(a_1, b_1), \ldots, (a_t, b_t)$.
    \item Using \texttt{box} to approximate $\lozenge \mathbf{1}_{A'}$, return an estimate of the value $t^{-1} \sum_{i=1}^{t} \left( \lozenge \mathbf{1}_{A'}(a_i, b_i) - \proj_{\beta} \lozenge \mathbf{1}_{A'} (a_i,b_i) \right)^2$.
\end{itemize}
}
}\vspace{\algbotskip}

If we want to approximate $\xi(x,y)$ within an additive error of $\omega$ with confidence at least $ 1- \nu$, we can pick $t = O(\poly(\omega, \log(\nu^{-1})))$ and estimate $\lozenge \mathbf{1}_{A'}$ within an additive error of $\omega/2$ with confidence at least $1 - \delta/(2t)$. In particular, within the bounds on our runtime, we can ensure that the outcomes $\mathcal{B} := B_{v,w,z}$ of \texttt{high-rk-bohr-set} in Theorem~\ref{thm:step4.2} has the guarantees that $\mu(x) \geq [\mu]^8/4$ and $\xi(x) \leq 3/2 \gamma^{-1} \zeta$ for $x \in \mathcal{B}$.  

The next step is to run a majority vote style argument to pass from 99\% structure to 100\% structure. Apply Theorem~\ref{thm:step5}, using \texttt{psi} to get query access to $\psi$ and also the explicit description for the Bohr set $\mathcal{B}$ from Theorem~\ref{thm:step4.2}. The output of Theorem~\ref{thm:step5} gives us query access \texttt{query-tilde-psi} to $\widetilde{\psi}$ with the property that $d(\psi(w,h), \widetilde{\psi}(w,h)) \leq 64 \eta p^{-3k}$. The existence of such a $\widetilde{\psi}$ follows from \cite[Lemma 6.21]{GM17}. We will first show that such a $\tilde{\psi}$ is unique. We will prove that there cannot be $\tilde{\psi_1} \neq \tilde{\psi_2}$ such that for both $i$ we have that 
\[ d \left(\psi(w,h), \delta_{\tilde{\psi}_i(w,h)} \right) \leq 64 \eta p^{-3k}.\]

Indeed, note that by applying the triangle inequality for $d$ on $d \left(\psi(w,h), \delta_{\tilde{\psi}_1(w,h)} \right)$ we get that 
\[ d \left(\psi(w,h), \psi(w,h) \right) \leq 2 \ep. \]
Here, $\ep = 64 \eta p^{-3k}$. This in turn implies by H\"older's inequality that 
\[ 1- 2 \ep \leq \norm{\psi(w,h)}_2^2 \leq \norm{\psi(w,h)}_1 \norm{\psi(w,h)}_{\infty} = \norm{\psi(w,h)}_{\infty}. \]
So it follows that for $\tilde{\psi}_2$ we must have that $\left \langle \delta_{\tilde{\psi}_2(w,h)}, \psi(w,h) \right \rangle \leq 2 \ep$ and so $d \left(\psi(w,h), \delta_{\tilde{\psi}_2(w,h)} \right) \geq 1 - 2\ep > \ep$ since $\ep < 1/3$.

In \cite[Lemma 6.21]{GM17}, $\psi$ is approximated by a function $\widetilde{\psi}'$ that is almost additive in each variable, satisfying the following properties simultaneously:
\begin{enumerate}
    \item [(1)] $\widetilde{\psi}'(w_1,h) + \widetilde{\psi}'(w_2, h) = \widetilde{\psi}'(w_1 + w_2, h)$ for all triples $(w_1, w_2, h)$ outside a set of density at most $16 \eta^4 p^{-16k}$. 
    \item [(2)] $\widetilde{\psi}'(w,h_1) + \widetilde{\psi}'(w,h_2) = \widetilde{\psi}'(w, h_1 + h_2)$ for all triples $(w_1, w_2, h)$ outside a set of density at most $16 \eta^4 p^{-16k}$. 
    \item [(3)] $d \left(\psi(w,h), \delta_{\widetilde{\psi}'(w,h)} \right) \leq 64 \eta p^{-3k}$ for all $(w,h)$ outside a set of density at most $16 \eta^4 p^{-16k}$. 
\end{enumerate}
As we have shown that such a $\widetilde{\psi}'$ satisfying (iii), if it exists, is unique, and so it follows that the $\widetilde{\psi}$ given by \texttt{query-tilde-psi} also satisfies properties (1) and (2). In other words, $\widetilde{\psi}$ is additive on an at least $1 - 48 \eta^4 p^{-16k}$ fraction of $B''$. This allows us to apply Theorem~\ref{thm:step6} to recover such a bi-affine map $T\colon G^2 \to G$, which extends the domain of $\widetilde{\psi}$ to $G^2$ and agrees with $\widetilde{\psi}$ on an at least a $1 - O(p^{-4k})$ fraction of $B''$.

To finish up, we need to ``symmetrize'' and ``anti-differentiate'' to recover the correlating cubic. To that end, using the description of $T$ provided by Theorem~\ref{thm:step6} as input in Theorem~\ref{thm:step7}, we can find a cubic $\kappa$ such that 
\[ \abs{\E_x f(x)\omega^{-\kappa(x)}}>\eta\]
where $\eta^{-1}=\exp\exp\qpoly(\epsilon^{-1})$. Using a union bound, assuming that none of the substeps failed, $\kappa$ has a correlation of $\eta^{-1} = \exp \exp \qpoly (\ep^{-1})$ with $f$ following the arguments in \cite[Section 8, Section 11]{GM17}. For the runtime guarantees, it suffices to observe that each individual step of the algorithm stays within $O(\poly(n, \eta, \log(\delta^{-1})))$ queries to $f$.
\end{proof}

\section{Improving the quantitative bounds}

In this section we explain how to improve the shape of the quantitative bounds from $\eta^{-1} = \exp \exp \qpoly (\ep^{-1})$ to $\eta^{-1} = \exp \qpoly (\ep^{-1})$. We do this first in the non-algorithmic setting, proving Theorem~\ref{thm:exp-drop}, and then we explain how to implement this argument algorithmically, improving the bounds in Theorem~\ref{thm:main-but-weaker}. We recently learned that a similar argument was independently discovered by Shachar Lovett.

One of the exponentials arises due to a technical Fourier analytic lemma in \cite{GM17}, which we restate here.

\begin{theorem}[{\cite[Theorem 4.15]{GM17}}]\label{thm:GM4.15}
For every $\xi > 0$ there exists a positive integer $k$ with the following property. Let $f \colon G^2 \to \C$ be any bounded function. Then there is a bi-affine map $\beta \colon G^2 \to \F_p^k$ such that for $F = \lozenge f$ we have the approximation $\norm{F - \proj_{\beta}F}_2 \leq \xi$. Moreover, $k$ can be taken to be $4m^34^m/\xi^2$ where $m = \exp(2^{69}(\log(\xi^{-1})+\log p)^6)$.
\end{theorem}

We briefly describe where one of the exponentials arise in this step. In order to approximate $F$, it is reasonable to just approximate it by its large Fourier spectrum; precisely, we may consider $F'(x,y) = \sum_{r: r \in \Spec_{\gamma}(F_{\bullet y})}\widehat{F_{\bullet y}}(r) \omega^{x.r}$. In order to get the affine structure $\beta$, it follows that we would want to cover the large Fourier spectrum $\bigcup_y \Spec_{\gamma}(F_{\bullet y})$ by affine maps. Gowers and Mili\'cevi\'c do this in \cite[Lemma 4.10]{GM17}. The issue with approximating $F(x,y)$ by $F'(x,y) = \sum_{i=1}^{m} \widehat{F_{\bullet y}}(T_i y)\omega^{\langle x, T_i y\rangle}$ is that it is possible for $T_i y = T_j y$. As discussed before, for each $i$, Gowers and Mili\'cevi\'c consider $u_i$ where
\[ u_i(y) = \begin{cases} 0 &T_jy = T_i y \text{ for some } j<i, \\ \widehat{F_{\bullet y}}(T_i y) & \text{otherwise}. \end{cases}\]
We may then try to approximate $F$ using $G(x,y) = \sum_{i=1}^{m} u_i(y) \omega^{\langle x, T_i y \rangle}$. In order to show that $G(x,y)$ is a good approximation, it turns out that we need to get a handle on $\norm{\widehat{u_i}}_1$, which is where an exponential arises.

Gowers and Mili\'cevi\'c write $u_i(y) = v_i(y) \mathbf{1}(y: i = \min \{ j:T_j y = T_i y \})$ where $v_i(y) = \E_xF(x,y) \omega^{x.T_i y}$. Since $\norm{v_i(y)}_1 \leq 1$ by \cite[Lemma 4.11]{GM17}, it suffices to estimate $\norm{\mathbf{1}(y: i = \min \{ j:T_j y = T_i y \})}_1$. To that end, they write $\mathbf{1}(y: i = \min \{ j:T_j y = T_i y \})$ as a $\pm 1$ combination of indicators of subspaces, and the $L^1$ norm of each of these subspaces would be bounded by 1 so that the desired $L^1$ norm would be bounded by the number of such subspaces we introduce. However, they do this partitioning directly using the principle of inclusion and exclusion, which produces $2^m$ subspaces and causes an exponential to arise; precisely, they write 
\[ \mathbf{1}(y: i = \min \{ j:T_j y = T_i y \}) = \prod_{k < i} (1 - \mathbf{1}(y: k \in \{ j:T_j y = T_i y \})).\]

We aim to do a more careful analysis of $\mathbf{1}(y: i = \min \{ j:T_j y = T_i y \})$. We can think of our end goal as effectively finding some affine maps such that for all $i$ we have that $\mathbf{1}(y: i = \min \{ j:T_j y = T_i y \})$ is constant on the level sets of these affine maps. The intuition for what we want to do is that we can ``pre-partition'' our ambient space in such a way as to make many $T_i - T_j$ have ``high rank''. This reduces the amount of ``overlapping space'' that we would need to do the PIE argument on, and will give better bounds.

As an illustration of this idea, suppose that $\ker(T_1 - T_2) = \{ x_1 = 0, x_2 = 0 \}$. Our goal is to introduce some additional affine forms such that the indicator of this subspace is constant on the corresponding level sets. Using Gowers and Mili\'cevi\'c's argument, we would write $\mathbf{1}(x_1 = 0, x_2 = 0) = p^{-2} \sum_{s,t \in \F_p} \omega^{tx_1 + sx_2}$. We partition the space by adding in the $p^2$ forms $\{ tx_1 + sx_2 \}_{s,t \in \F_p}$. However, note that $ tx_1 + sx_2 \in \spn \{ x_1, x_2 \}$ so if we have ``pre-partitioned'' our space by introducing the forms $\{x_1, x_2 \}$ then $\ker(T_1 - T_2)$ will be measurable and we will have saved on using many redundant forms.

\begin{theorem}\label{thm:4.15imp}
For every $\xi>0$ there exists a positive integer $k$ with the following property. Let $G=\F_p^n$ and let $f\colon G\times G\to\C$ be any bounded function. Then there is a bi-affine map $\beta\colon G^2\to\F_p^k$ such that, writing $F$ for the mixed convolution $\talloblong f$ and $\proj_\beta$ for the averaging projection on to the level sets of $\beta$, we have the approximation $\|F-\proj_\beta F\|_2\leq\xi$. Moreover, $k$ can be taken to be $O(m^3/\xi^2)$ where $m=\exp(2^{69}(\log(\xi^{-1})+\log p)^6)$.
\end{theorem}

\begin{proof}
Let $\epsilon,\gamma, r>0$ be constants to be chosen later. Define $g=f \multimapdotbothvert f$ (so that $F=g \multimapdotboth g$). By \cite[Lemma 4.10]{GM17}, there exist affine maps $T_1,\ldots,T_m\colon G \to G$ such that for all but at most $\epsilon|G|$ values of $h$, the $\gamma$-large spectrum of $g_{\bullet h}$ (meaning here the set of $u$ such that $\abs{\widehat{g_{\bullet h}}(u)}^2\geq\gamma$) is contained in the set $\{T_1h,\ldots,T_mh\}$ where $m = \qpoly(\gamma^{-1}, \ep^{-1})$.

Let $\mathfrak B=(L_1,\ldots,L_t)$ be a list of linear forms such that if $\E[\mathbf1(y: T_iy=T_jy)]\geq p^{-r}$, then $\mathbf1(y:T_iy=T_jy)$ is $\mathfrak B$-measurable. Note that one can choose $t\leq m^2r$. This is because for fixed $i,j$, if we write $U_{ij} = \ker(T_i - T_j)$ then $\dim U_{ij}^{\perp} \leq r$ and we can add to $\mathfrak B$ a basis for $U_{ij}^{\perp}$. Doing this for all possible pairs of $i,j$, we see that $t \leq \binom{m}{2}r \leq m^2r$. 

For $1\leq i\leq m$, define $u'_i(y)=\mathbf1(y:T_iy\not\in\{T_1y,\ldots,T_{i-1}y\})$.

\begin{claim}\label{claim:alm-meas}
$\|u'_i-\E(u'_i|\mathfrak B)\|_2^2\leq mp^{-r}$.
\end{claim}
\begin{proof}
We write $\mathcal{I} = \{ (i,j) : \Pb_y(T_iy = T_jy) \geq p^{-r} \}$. Consider $v_i'(y) = \mathbf{1}(y:T_i y \neq T_j y, \forall j < i \text{ s.t. } (j,i) \in \mathcal{I})$ corresponding to the ``small'' overlaps. Since $\{ y: v_i(y) = 1 \} = \{ \bigcup_{i \in \mathcal{I}} \{ y: T_i y = T_jy\} \}^{c}$ is $\mathfrak{B}$-measurable, and the closest approximation to $u_i'$ in $L^2$ by a function that is $\mathfrak{B}$-measurable is $\E(u_i' \mid \mathfrak B)$, it follows that 
\[\norm{u_i' - \E(u_i' \mid \mathfrak B)}_2^2 \leq \norm{u_i' - v_i'}_2^2 = \Pb_y(T_iy = T_jy \text{ for some } j < i \text{ s.t. } (j,i) \not \in \mathcal{I}) \leq  m  p^{-r}.\]
\end{proof}

Define $H(x,y) = \sum_{i=1}^{n} \widehat{F_{\bullet y}}(T_i y) u'_i(y) \omega^{x.T_iy}$, where for each $y$ we have that $\widehat{H_{\bullet y}}$ is the restriction of $\widehat{F_{\bullet y}}$ to the set $\{T_1 y, \ldots, T_my\}$. Because $F$ is defined by a convolution, we have that $\widehat{F_{\bullet y}} = \abs{\widehat{g_{\bullet y}}}^2$. Consequently, if $y$ is such that $\Spec_{\gamma}(g_{\bullet y}) \subset \{T_1y, \ldots, T_my \}$ then we have the inequality \[\norm{\widehat{H_{\bullet y}} - \widehat{F_{\bullet y}}}_2^2 \leq \norm{\widehat{H_{\bullet y}} - \widehat{F_{\bullet y}}}_1 \norm{\widehat{H_{\bullet y}} - \widehat{F_{\bullet y}}}_{\infty} \leq \gamma.\] For $y$ such that $\Spec_{\gamma}(g_{\bullet y}) \not \subset \{T_1y, \ldots, T_my \}$, we will use the na\"ive bound that $\norm{\widehat{H_{\bullet y}} - \widehat{F_{\bullet y}}}_2^2 \leq 1$. Taken together, since the density of the latter is at most $\ep$, we have that $\norm{H-F}_2^2 \leq \ep + \gamma$. 

For the next part of the argument, write 
\[ H(x,y) = \sum_{i=1}^{m} \underbrace{\widehat{F_{\bullet y}}(T_i y)}_{(I)} \underbrace{u'_i(y)}_{(II)} \underbrace{\omega^{x.T_iy}}_{(III)}. \]
We have that (III) is $(T_1,\ldots,T_m)$-measurable. Claim~\ref{claim:alm-meas} states that (II) is ``almost $\mathfrak B$-measurable''.
 
For $(I)$, we do a further approximation by truncating its Fourier transform.

For notational simplicity, write $w_i(y) = \widehat{F_{\bullet y}}(T_i y)$. Let $\lambda >0$ be a constant to be picked later. By \cite[Lemma 4.11]{GM17}, we have that $\norm{\widehat{w_i}(y)}_1 \leq 1$. This means by Parseval's that we can find a set $S_i$ of size at most $m^2/\lambda^2$ such that $\abs{\widehat{w_i}(y)} < \lambda^2/m^2$ for every $y \not \in S_i$. Consider $w_i'(r) = \sum_{y \in S_i} \widehat{w_i}(y) \omega^{y.r}$, so that by H\"older's inequality we have that $\norm{w_i - w_i'}_2^2 = \norm{\widehat{w_i} - \widehat{w_i'}}_2^2 \leq \norm{\widehat{w_i} - \widehat{w_i'}}_1 \norm{\widehat{w_i} - \widehat{w_i'}}_{\infty} \leq \lambda^2/m^2$.

Let $\mathfrak B'$ be the common refinement of $\fB$, the bi-affine forms $(T_1, \ldots, T_m)$, and the linear forms $y\mapsto v\cdot y$ for all $y\in S_i$ and all $1\leq i\leq m$. Note that the number of linear forms in $\mathfrak B$ is defined by at most $m^2r+m^3/\lambda^2$ linear forms and $m$ bilinear forms. Now (I) is approximated by $w_i$ which is $\fB'$-measurable, (II) is almost $\fB'$-measurable, and (III) is $\fB'$ measurable.

Now we put everything together. For simplicity, write $u_i'':=\E(u_i'|\fB)$. Now write \[H'(x,y)=\sum_{i=1}^m w'_i(y)u''_i(y)\omega^{x\cdot T_iy}.\] By definition, $H'$ is $\fB'$-measurable. We have the tools to show that $H'$ is a good approximation for $H$. First since $H'$ is $\fB'$-measurable,
\[\|H-\E(H|\fB')\|_2\leq \|H-H'\|_2\leq \sum_{i=1}^m\|w_iu'_i\omega^{ T_i}-w'_iu''_i\omega^{ T_i}\|_2.\]
Then for each $i$,
\begin{align*}
    \|w_iu'_i\omega^{ T_i}-w'_iu''_i\omega^{ T_i}\|_2 &= \|w_iu'_i-w'_iu''_i\|_2 \\
    &\leq\|w_i(u'_i-u''_i)\|_2+\|(w_i-w'_i)u''_i\|_2\\
    &\leq \|u_i-u_i''\|_2\|w_i\|_\infty+\|w_i-w'_i\|_2\|u''_i\|_{\infty}\\
    &\leq m^{1/2}p^{-r/2}+\lambda/m.
\end{align*}
Therefore we have shown that $\|H-\E(H|\fB')\|_2\leq m^{3/2}p^{-r/2}+\lambda$.

Now define $\beta \colon G^2 \to \F_p^k$ to be the natural bi-affine map corresponding to $\mathfrak B'$ with one coordinate  for each map defining $\fB'$. Note that here we can take $k \leq m(mr + m^2/\lambda^2) + m$. Combining the estimates above via the triangle inequality and taking $\ep = \gamma = \xi^2/8$ and $\lambda = \xi/4$ as well as $r = 3\log m + 2\log (1/\xi) - 4$, it follows that 
\[ \norm{F - \proj_{\beta} F}_2 \leq \norm{F - \E(H \mid \mathfrak{B}')}_2 \leq \norm{F - H}_2 + \norm{H - \E(H \mid \mathfrak{B}')}_2 \leq \xi/2 + \xi/2 = \xi\]
since the closest approximation to $F$ in $L^2$ by a function that is $\mathfrak B'$-measurable is $\proj_{\beta}F$. Furthermore, by our choice of parameters, we have that for the corresponding $\beta \colon G \to \F_p^k$ we can take $k = O(m^3/\xi^2)$, as desired. 
\end{proof}

In comparison with Theorem~\ref{thm:GM4.15} \cite[Theorem 4.15]{GM17}, note the codimension of the bi-affine map $\beta$ that we obtain is indeed smaller by one exponent. Now Theorem~\ref{thm:exp-drop} is proved by following the usual proof \cite{GM17} (and \cite{T21} for $p=2,3$) but replacing the usage of Theorem~\ref{thm:GM4.15} \cite[Theorem 4.15]{GM17} with Theorem~\ref{thm:4.15imp}.

Next, we discuss the implications of such a quantitative improvement by giving the proof of Theorem~\ref{thm:main} for the case of $p \geq 5$. We will handle the cases of $p = 2,3$ in the following section. We replace Theorem~\ref{thm:step3.2} with the following instead, and keeping the rest of the algorithm the same. In turn, this will imply a quantitative improvement in the bound of $\eta^{-1}$ by one less exponential, hereby proving Theorem~\ref{thm:main}.

\begin{theorem}\label{thm:mod-step3.2}
Let $f\colon  G \times G \to \C$ be a bounded function. Let \texttt{approx-f}($\ep, \delta, x$) be an oracle such that for every $x \in G^2$ we have with probability at least $1 - \delta$ that $\abs{f(x) - \texttt{approx-f}(\ep, \delta, x)} \leq \ep$. The affine maps $T_1, \ldots, T_m$ for some $m = \qpoly(\xi^{-1})$ are such that for all but at most $\xi \abs{G}$ points $(h,u) \in \left \{ (h,u): \abs{\widehat{f_{\bullet h}}(u)}^2 \geq \gamma \right \}$ we have $T_ih = u$. Write $F = \lozenge f$.

Given query access to \texttt{approx-f} and also an explicit description of the maps $T_1, \cdots, T_m$, there exists an algorithm \texttt{bohr-aff-map} that makes $O(\qpoly(\xi^{-1}) \cdot \poly(n, \log(\delta^{-1})) )$ queries to \texttt{approx-f} and with probability at least $1 - \delta$ returns a bi-affine map $\beta\colon  G \times G \to \F_p^k$ such that $\norm{F - \proj_{\beta} F}_2 \leq \xi$ where $k = O(\qpoly(\xi^{-1}))$.
\end{theorem}

The algorithm and proof are largely the same as that of Theorem~\ref{thm:step3.2}, with the additional linear forms that we need to add specified as in the proof of Theorem~\ref{thm:4.15imp}. 

\vspace{\algtopskip}
\noindent \fbox{
\parbox{\textwidth}{
    \texttt{mod-bogo-u($\mathcal{L}$,f,i,y,$\mathcal{L}_{\text{bad}}$)}:
    \begin{itemize}
        \item Using \texttt{box(f,w,h)} to get query access to $\lozenge f$, sample $r$ values $\{x_i\}_{i=1}^r$ from $G$ and return
        \[
        \dfrac{1}{r} \sum_{i=1}^r \lozenge f(x_i,y) \omega^{-x_i\cdot T_iy}.
        \]
\end{itemize}
}
}\vspace{\algbotskip}

\vspace{\algtopskip}
\noindent \fbox{
\parbox{\textwidth}{
    \texttt{mod-bohr-aff-map(f)}: 
    \begin{itemize}
        \item Run \texttt{bogo-aff-map(f)} and let its output be $\mathcal{L} = \{T_1, \ldots, T_m \}$.
        \item Let $\mathcal{J} = \emptyset$ and $\mathcal{J}_{\text{bad}} = \emptyset$. For each $1 \leq i,j \leq m$, sample $t$ elements $g_1, \ldots, g_t$ of $G$; if for some $g_r$ we have that $T_i g_r = T_j g_r$, add $(i,j)$ to $\mathcal{J}$. Otherwise, add $(i,j)$ to $\mathcal{J}_{\text{bad}}$. 
        \item For each pair $(i,j)$ in $\mathcal{J}$, use Gaussian elimination to return a basis $b^{(i,j)}_1, \ldots, b^{(i,j)}_{u}$ for $\ker(T_i - T_j)^{\perp}$. 
        \item Using \texttt{mod-bogo-u} to get approximate query access to $w_i(y) = \widehat{F}_{\bullet y}(T_i y)$, run \texttt{noisy-GL($w_i$, $\xi$)} and let the output be $S_i = \{v_{i1}, \ldots, v_{it_i} \}$.
        \item Return
        \[ \Scale[0.8]
        {\left(  T_1^{*}x.y, \ldots, T_m^{*}x.y, v_{11}.y, \ldots, v_{1t_1}.y, \ldots, v_{m1}.y, \ldots, v_{mt_m}.y, \left(b_1^{(i_1,j_1)}\right)^T y, \ldots, \left(b_{u_1}^{(i_1,j_1)}\right)^Ty, \ldots, \left(b_1^{(i_a,j_a)}\right)^T y, \ldots, \left(b_{u_a}^{(i_a,j_a)}\right)^Ty   \right)}
        \]
        where the elements of $\mathcal{J}$ are $(i_1, j_1), \ldots, (i_a, j_a)$. 
\end{itemize}
}
}\vspace{\algbotskip}

\begin{proof}
We analyze the query complexity and probabilistic guarantees of \texttt{mod-bohr-aff-map}. Let $r = 3 \log m + 2 \log(\xi^{-1}) - 4$. By the standard Chernoff bound (Lemma~\ref{lem:CH}), setting $t = O(\qpoly(\xi^{-1}) \cdot \poly( \log(\delta^{-1})))$ in \texttt{mod-bohr-aff-map} we have with probability at least $ 1- \delta/2$, if $\E[\mathbf{1}(y: T_iy = T_jy)] \geq 2p^{-r}$ then $(i,j)$ has been identified and will be added to $\mathcal{J}$. For the next step in \texttt{mod-bohr-aff-map}, note that via \texttt{mod-bogo-u} we are able to obtain query access to some $w_i'(y)$ such that with probability at least $ 1 - \delta\qpoly(\xi)/4$, we have that $\norm{w_i - w_i'}_{\infty} \leq \xi$. By the Chernoff bound, such an approximation $w_i'$ can be obtained  using $r = O(\poly(\xi^{-1}) \cdot log(\delta^{-1}))$ queries to \texttt{approx-f} in  \texttt{mod-bogo-u}. Applying \texttt{noisy-GL} to the approximation $w_i'$, for each $i$ we can ensure with probability at least $ 1 - \delta$ that the output of \texttt{noisy-GL}($w_i', \xi)$ is a list $L_i$ with the property that $\Spec_{\xi}(w_i) \subset L_i \subset \Spec_{\xi/2}(w_i)$ using $O(\qpoly(\xi^{-1}) \cdot \poly(n, \log(\delta^{-1})))$ queries to \texttt{approx-f}. Putting the above together with the proof of Theorem~\ref{thm:4.15imp}, it follows that with probability $1 - \delta$ we identify $\beta$ with the desired property, using a total of $O(\qpoly(\xi^{-1}) \cdot \poly(n, \log(\delta^{-1})))$ queries to \texttt{approx-f}.
\end{proof}

\section{Low characteristic}

We briefly summarize the difficulties that we face in low characteristics, namely when $p = 2,3$. For a more detailed discussion, we refer the reader to \cite[Section 1.3]{T21}. In our setting, the main issue in low characteristics is the symmetrization step intrinsic in Theorem~\ref{thm:step7}. The input to Theorem~\ref{thm:step7} is a bi-affine map $T$. Let $\tau(x,y,z) = T^{L}(x,y) \cdot z$ where $T^L$ is the trilinear part of $T$. This trilinear form $\tau(a,b,c)$ we identify has the following property that
\[ \E_x \E_{a,b,c} \partial_{a,b,c} f(x) \omega^{-\tau(a,b,c) - \rho(a).c - \kappa(b).c} \geq \exp \qpoly(\alpha), \]
where $\rho, \kappa\colon G \to G$ are some affine maps.

In Theorem~\ref{thm:step7}, we identified the cubic polynomial $\kappa(x) = \tau(x,x,x)$ as the cubic part of our output and recovered the lower degree part of our output using the $U^3$ inverse theorem. Paricularly, to execute such a proof strategy we need to identify a cubic polynomial such that 
\[ \norm{f \omega^{\kappa(x)}}_{U^3} = \left( \E_x \E_{a,b,c} \partial_{a,b,c} f(x) \omega^{D_{a,b,c}\kappa(x) }\right)^{1/8}\]
is large, and $D_{a,b,c} \kappa(x)$ is a symmetric trilinear form. Comparing this expression with the one from before, it follows that in order for such a strategy to work we ought to identify a symmetric trilinear form ``close to'' $\tau(a,b,c)$. In \cite{GM17}, the candidate for this symmetric trilinear form is given by 
\[ \sigma(a,b,c) = \frac{1}{6} \left( \tau(a,b,c) + \tau(a,c,b) + \tau(b,a,c) + \tau(b,c,a) + \tau(c,a,b) + \tau(c,b,a) \right). \]

Clearly, this type of symmetrization no longer works in characteristics of $p = 2,3$. Following \cite{T21}, we will perform a different form of symmetrization and integration in the cases of $p = 2$ and $p = 3$. Of note is that the situation is considerably more delicate when working in characteristic 2. The $U^4$ inverse theorem as stated with classical polynomials is false (with counterexamples independently discovered by \cite{LMS11} and \cite{GT09}), and we will need to introduce the concept of non-classical polynomials to recover such a $U^4$ inverse theorem. We give the necessary definitions below. For a more thorough introduction see \cite{TZ12}.

For $P\colon G\to G$, we write $D_hP(x):=P(x+h)-P(x)$ for the additive derivative. We also use $D_{h,h'}$ as shorthand for $D_hD_{h'}$.

\begin{definition}
A \emph{non-classical polynomial} of degree at most $k$ is a map $P \colon \F_p^n \to \R/\Z$ that satisfies 
\[ D_{h_1 \ldots h_{k+1}} P(x) = 0\]
for all $h_1, \ldots, h_{k+1}, x \in \F_p^n$. Note that a classical polynomial $Q \colon \F_p^n \to \F_p$ can be thought of as a classical polynomial by composing with the homomorphism $\F_p \hookrightarrow \{0, 1/p, \ldots, (p-1)/p \} \subset \R/\Z$. 

We will also write the \emph{total derivative} of a non-classical cubic polynomial $P$ as $d^3P \colon (\F_p^n)^3 \to \F_p$ given by 
\[ d^3P(h_1, h_2, h_3) = D_{h_1, h_2, h_3}P(0)\]
Note that we have $d^3P(h_1,h_2,h_3) = D_{h_1, h_2, h_3}P(x)$ for all $h_1,h_2,h_3,x \in \F_p^n$
\end{definition}

In this section, we provide algorithmic versions of the symmetrization and integration steps used to prove the quantitative $U^4$-inverse theorem for $p = 2,3$ in \cite{T21}. For concreteness, we state the statement of the $U^4$-inverse theorem in this setting which involves non-classical polynomials.

\begin{theorem}[{\cite[Theorem 1.3]{T21}}]
Fix a prime $p$. For every $\ep > 0$, there is a constant $\eta^{-1} = \exp \qpoly(\ep^{-1})$ such that for any bounded function $f \colon \F_p^n \to C$ which satisfies $\norm{f}_{U^4} \geq \ep$, there exists a non-classical cubic polynomial $P \colon \F_p^n \to \R/\Z$ such that 
\[ \left| \E_x f(x) \omega^{-P(x)}  \right| \geq \eta. \]
Furthermore, if $p \geq 3$, the polynomial can be taken to be classical.
\end{theorem}

The main idea here is, similar in spirit to \cite{Sam07}, that instead of constructing the symmetric linear forms explicitly, we may instead restrict to a subspace of small codimension on which the given linear form is symmetric and consider this restricted linear form in our arguments instead. We may then algorithmize this linear algebraic argument. 
Before proceeding to the subsections, we introduce several notions of rank that will help us measure how close a symmetric trilinear form is to our given trilinear form.

\begin{definition}
Let $V$ be a finite-dimensional $\F_p$-vector space $V$. For any $k$-linear form $T\colon V^k \to \F_p$, we define the following notions of rank. The \textit{analytic rank}, denoted $\arank$, is defined by $p^{-\arank(T)} = \E_{x_1,\ldots, x_k \in V}\omega^{T(x_1, \ldots, x_k)}$.
\end{definition}

\subsection{\texorpdfstring{$p=3$}{p=3} case}

Recall that up till this stage in the algorithm, we have produced a triaffine form $\tau'(a,b,c)$ with the property that 
\[\E_x \E_{a,b,c} \partial_{a,b,c} f(x) \omega^{-\tau'(a,b,c)} \geq \exp \qpoly(\ep^{-1}). \]

Let the trilinear part of $\tau'$ be $\tau$. The following lemmas guarantee that we will be able to find a symmetric trilinear form close to $\tau(a,b,c)$ with the necessary properties to integrate it into a (classical) cubic polynomial. We need to identify a special kind of symmetric trilinear forms known as a classical symmetric form (CSMs), which was first introduced in \cite{TZ10}. 

\begin{definition}
A \emph{classical symmetric trilinear form} is a map $T \colon (\F_p^n)^3 \to \F_p$ such that: 
\begin{itemize}
    \item For each $1 \leq i \leq 3$, fixing all the variables but $h_i$, the map $ h_i \mapsto T(h_1, h_2, h_3)$ is linear.
    \item $T(h_1, h_2, h_3)$ is symmetric; that is, it is invariant under permutations of $h_1, h_2, h_3$. 
    \item For all $h \in \F_p^n$, we have that $T(h,h,h) = 0$. 
\end{itemize}
\end{definition}

As in \cite[Section 4]{T21}, we are guaranteed the existence of a trilinear CSM $\sigma$ that is close in rank to $\tau$, with $\rank(\tau - \sigma) \leq \poly \log(\eta^{-1})$. This constitutes the symmetrization step. We now state the algorithmic version of this symmetrization. Throughout this section, the explicit description of trilinear forms we work with will be its representation as a 3-dimensional tensor. 

\begin{theorem}\label{thm:sym}
Given an explicit description of a trilinear form $\tau \colon (\F_3^n)^3 \to \F_3$ with the property that for any $\pi \in S_3$ we have that $\arank(\tau - \tau_{\pi}) \leq \delta$ where $\tau_{\pi}(x_1, x_2, x_3) = \tau(\pi(x_1), \pi(x_2), \pi(x_3))$, the algorithm \texttt{find-sym} in time $O(n^3)$ outputs an explicit description of a symmetric trilinear form $\widetilde{\sigma}\colon (\F_3^n)^3 \to \F_3$ satisfying $\rank(\tau - \widetilde{\sigma}) \leq 5 \delta$.
\end{theorem}

\begin{proof}
The goal is to find some subspace $V \leq \F_3^n$ such that $\tau \bigr|_{V^3}$ is a symmetric trilinear form with $\codim V = 5 \delta$. As observed in \cite{T21}, this is sufficient because writing $\F_3^n = V \oplus W$, we can define $\tilde{\sigma}(v_1 \oplus w_1, v_2 \oplus w_2, v_3 \oplus w_3) = \tau\bigr|_{V^3}(v_1, v_2, v_3)$. By \cite[Proposition 4.3]{T21}, we are guaranteed the existence of some subspace $U' \leq \F_3^n$ such that $\tau\bigr|_{U'}\colon (U')^3 \to \F_3$ is a symmetric trilinear form with $\codim U' \leq 5 \delta$. Let $\{ e_i \}$ be the standard basis vectors for $\F_3^n$. Let $T_i^{(1)}(u,v) = \tau(e_i, u,v)$, $T_i^{(2)}(u,v) = \tau(u,e_i,v)$ and $T_i^{(3)}(u,v) = \tau(u,v, e_i)$. Let $A_i^{(j)}$ be the matrix corresponding to $T_i^{(j)}(u,v)$ (with respect to any choice of basis). By Gaussian elimination in $O(n^3)$ time, we can find the subspace $U_i^{(j)} \leq \F_3^n$ of $v$ satisfying  $(A_i^{(j)} - (A_i^{(j)})^{T}) v = 0$. On this subspace $U_i^{(j)}$, there exists a symmetric matrix $B_i^{(j)}$ such that $B_i^{(j)} v = A_i^{(j)} v$ for $v \in U_i^{(j)}$. Let $U_i = \bigcap_j U_i^{(j)}$. In particular, this means for instance that for any $i$ and any $u,v \in U_1$, we have that $\tau(e_i, u , v) = T_i^{(1)}(u,v) = u^{T} B_i^{(j)} v = v^{T} B_i^{(j)} u = \tau(e_i, v, u)$. By linearity of $\tau$, we have that for any $w \in \F_3^n$ and $u,v \in U_1$ that $\tau(w,u,v) = \tau(w,v,u)$. Similarly for $U_2$ and $U_3$. In particular, this means that if we let $V = U_1 \cap U_2 \cap U_3$ then $\tau \bigr|_{V^3}$ is a symmetric linear form.

We claim that $\codim(V) = O(\exp \qpoly(\ep^{-1}))$. To that end, it suffices for us to note that $V^3 \supset (U')^3$. This is because, for example, any $v \in U'$ satisfies $v \in \ker(A_j^{(i)} - (A_i^{(j)})^{T}) = U_i^{(j)}$ by considering equalities of the form $\tau(e_i, e_j, v) = \tau(e_i, v, e_j)$.  Consequently, $U' \subset \bigcap U_i^{(j)} = V$. 
\end{proof}

\begin{theorem}\label{thm:CSM-3}
Given an explicit description of a symmetric trilinear form $\widetilde{\sigma} \colon (\F_3^n)^3 \to \F_3$  the algorithm \texttt{find-CSM} in time $O(n)$ outputs an explicit description of a trilinear CSM $\sigma$ such that $\rank(\widetilde{\sigma} - \sigma) = O(1)$.
\end{theorem}

\begin{proof}
The additional constraint that $\widetilde{\sigma}$ needs to satisfy in order to make it a CSM is that we need to restrict to a subspace $U \subset (\F_3^n)^3$ satisfying $\codim U = O(1)$ such that $\widetilde{\sigma}(x,x,x) = 0$ for all $x \in U$. As observed in Theorem~\ref{thm:sym}, this allows us to construct a $\sigma$ satisfying the properties of the theorem. 

By the proof of \cite[Proposition 4.4]{T21}, we have that $T\colon x \mapsto \widetilde{\sigma}(x,x,x)$ is a linear map. In particular, after extracting the linear map from querying $T(e_i) = \widetilde{\sigma}(e_i, e_i, e_i)$ for the standard basis $\{ e_i \}$ of $\F_p^n$ we can then output the codimension at most 1 subspace on which $\widetilde{\sigma}(x,x,x)$ vanishes. 
\end{proof}

Now that we are done with symmetrization, the final step is to integrate the obtained CSM $\sigma$ into a cubic polynomial. More generally, we have the following result by Tao and Ziegler.

\begin{theorem}[\cite{TZ10}]
Let $V$ be a finite dimensional vector space. Then we have the following short exact sequence 
\[ 0 \rightarrow \Poly_{\leq 2}(V \to \F) \rightarrow \Poly_{\leq 3}(V \to \F) \xrightarrow[]{d^3} \CSM^3(V) \rightarrow 0. \]
\end{theorem}

In \cite{TZ10}, Tao and Ziegler prove this theorem constructively indicating that it is feasible to algorithmically extract a cubic polynomial from a CSM. Here we demonstrate this explicitly for the cubic situation. 

\begin{theorem}\label{thm:integrate-3}
Given an explicit description of a CSM $\sigma \colon (\F_3^n)^3 \to \F_3$, the algorithm \texttt{integrate-cubic} in time $O(n^3)$ outputs a cubic polynomial $P$ such that $d^3P = \sigma$.
\end{theorem}

\begin{proof}
Expanding the CSM $\sigma(x,y,z)$ in terms of its monomials, note that the condition $\sigma(e_i,e_i,e_i) = 0$ forbids monomials of the form $x_iy_iz_i$. Note that given the tensor description of $\sigma(x,y,z)$, we can read off its coordinate-wise expansion from each entry of the tensor in $n^3$ time. In particular, we know that the $\sigma(x,y,z)$ has the form:
\[ \sigma(x,y,z) = \sum_{i,j,k} c_{ijk} \left(\sum_{\mathrm{sym}}x_iy_jz_k \right) + \sum_{i,j} d_{iij} \left( \sum_{\mathrm{sym}} x_iy_iz_j \right).\]

It suffices to note that $D_{h_1h_2h_3}x_ix_jx_k = \sum_{\mathrm{sym}}h_{1i}h_{2j}h_{3k}$ and $D_{h_1h_2h_3}x_i^2x_j/2 = \sum_{\mathrm{sym}}h_{1i}h_{2i}h_{3j}$. The former of the two is straightforward. The latter follows from a fairly tedious calculation. 
\begin{align*}
    D_{h_1h_2h_3}\frac{x_i^2x_j}{2} &= D_{h_1h_2} \left( h_{3i}x_ix_j + \frac{h_{3i}^2}{2} x_j + \frac{h_{3j}}{2} x_i^2 + h_{3i}h_{3j}x_i + \frac{h_{3i}^2 h_{3j}}{2} \right) \\
    &= D_{h_1}\left(  x_jh_{2i} h_{3i} + x_ih_{2j}h_{3i} + h_{2i}h_{3j} h_{3i} + h_{3i}^2h_{2j} + x_ih_{2i}h_{3j} + \frac{h_{2i}^2 h_{3j}}{2} + h_{2i}h_{3i}h_{3j}\right) \\
    &= h_{1j}h_{2i}h_{3i} + h_{1i} h_{2j}h_{3i} + h_{1i}h_{2i}h_{3j}. 
\end{align*}
These observations imply that we can take 
\[ P(x,y,z) = \sum_{i,j,k}c_{ijk}x_iy_jz_k + \sum_{i,j}\frac{d_{iij}}{2} x_i^2x_j. \]
\end{proof}

Putting the above parts together, we are able to prove Theorem~\ref{thm:step7} in characteristic $p=3$.

\begin{proof}[Proof of Theorem~\ref{thm:step7}]
Expanding the condition on $T$, and letting $\tau(x,y,z) = T^L(x,y,z) \cdot z$, 
\[ \E_x \E_{a,b,c} \partial_{a,b,c} f(x) \omega^{-\tau'(a,b,c)} \geq c \]
where $\tau'(a,b,c)$ is a tri-affine form with trilinear part $\tau$. By \cite[Lemma 10.3]{GM17} and \cite[Theorem 1.10]{J20}, $T$ satisfies the hypothesis of Theorem~\ref{thm:sym}. Successively using the output of Theorem~\ref{thm:sym} in Theorem~\ref{thm:CSM-3}, and then finally feeding this last output into Theorem~\ref{thm:integrate-3} gives the desired conclusion.
\end{proof}

\subsection{\texorpdfstring{$p=2$}{p=2} case}\label{sec:p=2}

Our goal in this subsection is to prove the following substitute for Theorem~\ref{thm:step7}.

\begin{theorem}\label{thm:step7-2}
Let $c,\delta>0$. Given query access to a bounded $f\colon  \F_2^n \to \C$ and an explicit description of a bi-affine map $T\colon \F_2^n \times \F_2^n \to \F_2$ such that $\E_{a,b}\abs{\widehat{\partial_{a,b}f}(T(a,b))}^2 \geq c$, there exists an algorithm \texttt{find-NCcubic} that makes $O(\poly(n,1/c,\log(1/\delta)))$ queries to $f$ and with probability at least $1- \delta$ outputs a non-classical cubic $\kappa$ with the guarantee that $\abs{\E_x f(x) (-1)^{\kappa(x)}} \geq \qpoly(c)$.
\end{theorem}

The strategy is similar to that in the case $p=3$. We first make $\tau(x,y,z) = T^L(x,y).z$ symmetric and then massage the resulting symmetric trilinear form into a form that facilitates anti-differentiating into a potentially non-classical cubic polynomial. Following \cite{T21}, it turns out that the relevant special trilinear form that we should consider are \emph{non-classical symmetric trilinear forms}.

\begin{definition}
A \emph{non-classical symmetric trilinear form} is a map $T\colon (\F_p^n)^3 \to \F_p$ such that:
\begin{itemize}
     \item For each $1 \leq i \leq 3$, fixing all the variables but $h_i$, the map $ h_i \mapsto T(h_1, h_2, h_3)$ is linear.
     \item $T(h_1,h_2, h_3)$ is symmetric; that is, it is invariant under permutations of $h_1, \ldots, h_k$. 
     \item $T(h_1, h_1, h_2) = T(h_1, h_2, h_2)$. 
\end{itemize}
\end{definition}

We will prove an algorithmic version of \cite[Proposition 4.6]{T21}.

\begin{theorem}\label{thm:nCSM-2}
Given an explicit description of a symmetric trilinear form $\widetilde{\sigma} \colon (\F_2^n)^3 \to \F_2$ such that there are 1-bounded functions $b_1, \ldots, b_7 \colon \F_2^n \to \C$ satisfying
\[ \left| \E_{x,y,z \in \F_2^n} b_1(x)b_2(y) b_3(z) b_4(x+y) b_5(x+z) b_6(y+z) b_7(x+y+z) (-1)^{\widetilde{\sigma}(x,y,z)} \right| \geq \delta, \]
the algorithm \texttt{find-nCSM} in time $O(n^3)$ outputs an explicit description of a non-classical trilinear form $\sigma$ such that $\rank(\widetilde{\sigma} - \sigma) = O(\log_2 (\delta^{-1}))$.

\end{theorem}

\begin{proof}[Proof of Theorem~\ref{thm:step7-2}]
As before, it suffices to identify a subspace $U \leq \F_2^n$ with small codimension such that $\widetilde{\sigma}(x,x,y) = \widetilde{\sigma}(x,y,y)$ for $x,y \in U$. In the proof of \cite[Proposition 4.6]{T21}, it was demonstrated that $B(x,y) = \widetilde{\sigma}(x,x,y) - \widetilde{\sigma}(x,y,y)$ is bilinear and $\arank B \leq 8 \log_2(\delta^{-1})$. By considering $B(e_i, e_j)$ where $\{e_i \}$ is the standard basis for $\F_2^n$, we are able to retrieve the matrix $A$ representing $B(x,y)$ with respect to $\{e_i \}$ in time $O(n^2)$. Note that finding the nullspace of $B(x,y)$ corresponds to finding the subspace $U \leq \F_2^n$ such that $Av = 0$ for $v \in U$. This can be computed in time $O(n^3)$ via Gaussian elimination. By our earlier observation, it follows that $\codim U = \rank B = \arank B \leq 8 \log_2(\delta^{-1})$. 
\end{proof}

Putting together Theorem~\ref{thm:sym} and Theorem~\ref{thm:nCSM-2}, we can find an approximating nCSM to the trilinear form $\tau$ that we started with. Now we need to integrate this nCSM to obtain a non-classical cubic polynomial. This is possible as a consequence of the following more general theorem proven in \cite{T21}.

\begin{theorem}[{\cite[Proposition 3.5]{T21}}]
For $k \geq 1$ and a nCSM $T\colon (\F_p^n)^k \to \F_p$, there exists a non-classical polynomial $P$ of degree at most $k$ such that $d^kP = T$.
\end{theorem}

We will give a more constructive way of integrating in the special case of $p = 2$ and $k = 3$. 

\begin{theorem}\label{thm:integrate-2}
Given an explicit description of a nCSM $\sigma \colon (\F_2^n)^3 \to \F_2$, the algorithm \texttt{integrate-nclassCubic} in time $O(n^3)$ outputs a non-classical cubic polynomial $P$ such that $d^3P = \sigma$.
\end{theorem}

Since this proof is fairly calculation intensive, we defer it to Appendix~\ref{sec:app}. Now we can finish up our proof for the characteristic $p=2$ case. 

\begin{proof}[Proof of Theorem~\ref{thm:step7-2}]
Expanding the condition on $T$, by averaging over $x \in \F_2^n$ and then expanding the derivatives, we find bounded functions $b_1, \ldots, b_7 \colon \F_2^n \to \C$ satisfying 
\[ \left| \E_{x,y,z \in \F_2^n} b_1(x) b_2(y) b_3(z) b_4(x+y) b_5(x+z) b_6(y+z)b_7(x+y+z) \omega^{\tau (x,y,z)} \right| \geq c\]
where $\tau(a,b,c) = T(a,b).c$. Applying \cite[Lemma 10.3]{GM17} and \cite[Theorem 1.10]{J20} we see that $T$ satisfies the hypothesis of Theorem~\ref{thm:sym}. We also satisfy the conditions of Theorem~\ref{thm:nCSM-2}. Successively passing to smaller subspaces as in the proof of the theorems allows us to find a subspace $U$ of codimension $\poly \log_2(c^{-1})$ such that $\tau\bigr|_U$ is a nCSM. As in the proof of Theorem~\ref{thm:sym}, this then allows us to construct a nCSM $\sigma$ with $\rank(\sigma - \tau) = O(\poly \log_2(c^{-1}))$. Apply Theorem~\ref{thm:integrate-2} to identify a non-classical polynomial $P$ such that $d^3P = \sigma$. By the proof of \cite[Theorem 1.3]{T21}, the closeness in rank of $\sigma$ to $\tau$ guarantees that $\norm{f\omega^{-P}}_{U^3} = \Omega(c)$ and we can finish off by invoking Theorem~\ref{thm:u3inv} to find the quadratic part of the polynomial.
\end{proof}

\section{Application to self-correcting Reed-Muller codes}

An application of our algorithmic $U^4$ inverse theorem, Theorem~\ref{thm:main}, is in coding theory and specifically in the setting of local decoding of Reed-Muller codes beyond the list decoding radius as given by Theorem~\ref{thm:decoding}.

We first prove the $p\geq 3$ case of Theorem~\ref{thm:decoding} by direct reduction to 
Theorem~\ref{thm:main}. 

\begin{proof}[Proof of Theorem~\ref{thm:decoding} for $p\geq 3$]
First, observe that it suffices to establish that $\Pb[f(x) = P(x)] \geq (1/p) + \ep$ implies that $\norm{\omega^{f}}_{U^4} \geq \ep$, at which point we can invoke the $U^4$ inverse theorem (Theorem~\ref{thm:main}) to get the existence of a cubic polynomial $Q$ satisfying the conditions as stated. 

Now, we note that 
\[ \left\lvert\sum_{t \neq 0} \omega^{t(f(x) - P(x))}\right\rvert = \begin{cases} p-1 &\text{if  } f(x) = P(x) \\ 1 &\text{otherwise}  \end{cases} \]
where in the case $f(x) \neq P(x)$ we recall that $\sum_i \omega^i = 0$. This implies that 
\begin{align*}
    \left\lvert\sum_{t \neq 0} \E_x \omega^{t(f(x) - P(x))}\right\rvert &\geq (p-1) \cdot \Pb[ f(x) = P(x)] - (1 -  \Pb[ f(x) = P(x)]) \\
    &> p \ep.
\end{align*}
By pigeonhole principle, there exists some $t \neq 0$ such that 
\[ \abs{\E_x \omega^{t(f(x) - P(x))}} \geq \frac{p}{p-1} \ep \geq \ep.\]
This in particular implies that 

\[
\ep \leq \abs{\E \omega^{tf - tP}} 
    \leq \norm{\omega^{tf - tP}}_{U^4} 
    = \norm{\omega^{tf}}_{U^4} 
    = \norm{\omega^{f}}_{U^4},
\]
where for the second inequality we used Gowers-Cauchy-Schwarz, the penultimate equality follows from the discrete derivative definition of $U^4$ norms which causes $P$ to vanish and the final equality comes from symmetry of the $p$th roots of unity. The resulting inequality puts us in a situation where we can apply Theorem~\ref{thm:main}.
\end{proof}

The case when the characteristic $p = 2$ is much more subtle, because a priori the above argument may produce a non-classical cubic polynomial. This discrepancy occurs because the above argument only uses the hypothesis that $\|\omega^f\|_{U^4}\geq\epsilon$ which, for $p=2$, is strictly weaker than the assumption that $f$ correlates with a classical cubic.

\begin{proof}[Proof of Theorem~\ref{thm:decoding} for $p=2$]
We are given $f\colon\F_2^n\to\F_2$ and the assumption that there exists a classical cubic polynomial $P\colon\F_2^n\to\F_2$ such that $\dist(f,P) \leq \frac{1}{2} - \ep$. We first use this hypothesis to obtain information regarding the proximity of $D_{ab}f$ and $D_{ab}P$. Indeed, begin by noting that since $\Pb_x[P(x) = f(x)] \geq \frac{1}{2} + \ep$, we have that 
\begin{align*}
    \E_{a} \Pb_x\left[ D_{a}f(x) = D_{a}P(x)\right] &= \E_{a}  \left[ \Pb_x[f(x+a)- f(x) = P(x+a) - P(x)]\right] \\
    &= \E_{a} \frac{1}{2^n} \left(\abs{H \cap (H +a)} + \abs{H^{c}\cap (H^{c} +a)} \right) \\
    & = \left(\frac{|H|}{2^n}\right)^2 + \left(\frac{|H^c|}{2^n}\right)^2 \\
    & \geq \left( \frac{1}{2} + \ep \right)^2 + \left( \frac{1}{2} - \ep \right)^2 = \frac{1}{2}+ 2\ep^2 
\end{align*}
where $H = \{x : f(x) = P(x) \}$ and the inequality is true by convexity. Now, let $H_a = \{ x: D_aP(x) = D_af(x)\}$ so that the above can be rewritten as 
\[\E_a \frac{\abs{H_a}}{2^n} \geq \frac{1}{2} + 2\ep^2. \]
Repeating this form of reasoning once more, we have by convexity that 
\begin{align*}
    \E_{a,b} \Pb_x \left[ D_{ab}f(x) = D_{ab} P(x) \right] &= \E_{a,b} \Pb(x+b, x \in H_{a} \text{ or } x+b, x \not \in H_{a}) \\
    &= \E_{a}\left(\frac{\abs{H_{a}}}{2^n}\right)^2 + \left(\frac{\abs{H_{a}^c}}{2^n}\right)^2  \\
     & \geq \left( \frac{1}{2} + 2 \ep^2 \right)^2 + \left( \frac{1}{2} - 2 \ep^2 \right)^2 = \frac{1}{2} + 8 \ep^4.
\end{align*}

By Markov's inequality there is a set $S \subset G^2$ of density at least $
\Omega(\ep^4)$ such that $\Pb_x[D_{ab}f(x)=D_{ab}P(x)]\geq \tfrac12+\epsilon^4$ for all $(a,b)\in S$.

Let $F = (-1)^f$. Note that for each $a,b$ the derivatives $D_{ab}P(x)$ is a linear function of $x$, say $D_{ab}P(x)=r_{a,b}.x+c_{a,b}$. Define $R_{a,b}:=\{\chi : |\widehat{\partial_{ab}F}(\chi)|\geq 2 \epsilon^4\}$. By Parseval's identity, we have that $\abs{R_{a,b}} \leq \ep^{-8}$. Furthermore, we can see that if $(a,b)\in S$ then $r_{a,b}\in R_{a,b}$ since \[\abs{\widehat{\partial_{a,b} F}(r_{a,b})}=\abs{\E_x\left[(-1)^{D_{ab}f(x)-r_{a,b}.x}\right]}=\abs{\E_x\left[(-1)^{D_{ab}f(x)-D_{ab}P(x)}\right]}=| 2\Pb_x(D_{ab}f(x)=D_{ab}P(x))-1|.\]For $(a,b)\in S$ the last quantity is at least $2\epsilon^4$.

An application of Theorem~\ref{thm:step1} implies that there exists an algorithm \texttt{member-A} which with probability $1 - \delta$ returns 1 for $(a,b)$ if $\norm{\widehat{\partial_{a,b}f}}_{\infty} \geq 2\ep^4$ and 0 if $\norm{\widehat{\partial_{a,b}f}}_{\infty} \leq \ep^4$. By our observation in the previous paragraph, \texttt{member-A} returns 1 on elements of $S$. 

Now, for \texttt{phi(F,a,b)}, instead of sampling uniformly at random from the list of large Fourier coefficients, we will use the process as described in \cite[Lemma 17]{TW11}: first run \texttt{linear-decomposition} (as described in \cite{TW11}) with parameter $\gamma = \delta = O(\epsilon^{28})$ on $\partial_{a,b} F$ and let \texttt{phi(F,a,b)} output $r$ with probability $\widehat{\partial_{a,b}F}(r)^2$ where we only consider sampling from the large Fourier coefficients. This choice of parameters for $\gamma$ and $\delta$ is because 
\begin{align*}
    &\Pb_{x,y \in G}[\phi(x+y, b) = \phi(x,b) + \phi(y,b) \land \widehat{\partial_{x,b}F}(\phi(x,b)) \geq \gamma \land \widehat{\partial_{y,b}F}(\phi(y,b)) \geq \gamma \land \widehat{\partial_{x+y,b}F}(\phi(x+y,b)) \geq \gamma ] \\ 
    &\geq \E_{x,y}\left[\sum_{\alpha, \beta} \widehat{\partial_{x,b}F}(\alpha)^2 \widehat{\partial_{y,b}F}(\beta)^2 \widehat{\partial_{x+y,b}F}(\alpha+\beta)^2 \right] - 3 \gamma - O(2^{-n}) \\
    &\geq \E_{y \in G}\left[\sum_{\alpha \in G} \widehat{\partial_{y,b}F}(\alpha)^6 \right] - 3 \gamma - O(2^{-n}) \\
    &\geq \epsilon^4 \E_{y: \{y,b \} \in S} \left[\sum_{\alpha \in G} \widehat{\partial_{y,b}F}(\alpha)^6 \right] - 3 \gamma - O(2^{-n}) \\
    &\geq \Omega(\epsilon^{28}) - 3 \gamma - O(2^{-n}).
\end{align*}
The calculations above resemble those of \cite[Lemma 4.10, Lemma 4.12]{HHL19}. One immediate consequence of such sampling is that since $\abs{ \widehat{\partial_{a,b}F} (r_{a,b})} = \Omega(\epsilon^4)$, for each $(a,b) \in S$, \texttt{phi(F,a,b)} outputs $r_{a,b}$ with probability at least $\Omega(\epsilon^8)$. It follows that with probability at least $ 1- \delta$ there exists a subset $\widetilde{A} \subset G^2$ with density at least $\Omega(\epsilon^{12})$ such that the output of \texttt{phi(F,a,b)} is $r_{a,b}$ for $(a,b) \in \widetilde{A}$.

\begin{claim}
There exists an algorithm \texttt{member-B} that makes $O(\poly(\ep^{-1}, \log(\delta^{-1})))$ queries to \texttt{member-A} and \texttt{phi} with the following properties: it outputs 1 if $a \in B$ and 0 otherwise, where $B$ has density at least $\Omega(\exp (\poly(\epsilon)))$ in $G^2$, $\abs{B \cap \widetilde{A}} \geq \exp(\poly(\epsilon)) \abs{\widetilde{A}}$ and $\phi|_{B \cap G \times \{ b\}}$ is a Freiman homomorphism for each $b$.
\end{claim}

The main idea is that while we do not have access to $\widetilde{A}$, by drawing sufficiently many samples (polynomial in $\epsilon$) we can ensure that the vertex $u$ we use to initialize the algorithm \texttt{BSG-test} of \cite[Section 4.2]{TW11} lies in $\widetilde{A}$ and then the output of the algorithm would allow us to capture a large subset of $\widetilde{A}$ with the desired properties.

\begin{proof}
Use \texttt{member-A} as a primitive and apply \texttt{BSG-test} from \cite{TW11} with the parameters of $\rho = \Omega(\epsilon^{28})$. For a set $X$, denote $X_{\phi} = \{ (x, \phi(x)): x \in X\}$. Note that with probability at least $\Omega(\epsilon^{12})$ over the choice of $u$ in \cite[Lemma 4.10]{TW11}, we have that $u \in \widetilde{A}_{\phi}$. Since \texttt{member-A} return 1 on $\widetilde{A}$, following through the proof of \cite[Lemma 4.10]{TW11}, for any choice of $u, \gamma_1, \gamma_2, \gamma_3$, there exists $A_{\phi}^{(1)} \subset A_{\phi}^{(2)}$ such that the output of $\texttt{BSG-test}$ is a membership tester that returns 1 on input $x$ with probability $1 - \delta$ when $x \in A_{\phi}^{(2)} \cap \widetilde{A}_{\phi}$ and returns 0 when $x \not \in A_{\phi}^{(1)} \cap \widetilde{A}_{\phi}$ with probability $1 - \delta$. Furthermore, with probability $\Omega(\poly(\epsilon))$ over the choice of $u, \gamma_1, \gamma_2, \gamma_3$, $\abs{A_{\phi}^{(1)} \cap \widetilde{A}_{\phi}} \geq \Omega(\poly(\ep)) \abs{G}$ and $\abs{A_{\phi}^{(2)} + A_{\phi}^{(2)}} \leq \poly(\ep^{-1})\abs{G}$.

Next, we follow \cite[Section 4.3]{TW11} in order to obtain an affine function $T$ such that with probability $1 - \delta$, $T$ agrees with $\phi$ on a $\exp(\poly(\epsilon))$ fraction of $\widetilde{A}$. By sampling $ \poly(\ep^{-1}, n^2, \log(\delta^{-1}))$ elements $((x,b), \phi(x,b))$ and running \texttt{BSG-test(u,$\cdot$)} on them, we may assume that on each element \texttt{BSG-test} satisfies the guarantees. We only retain the sampled points on which \texttt{BSG-test} returns 1. By \cite[4.13]{TW11}, with probability at least $ 1- \delta$ the retained points contain at least $t = \poly(\ep^{-1}, n^2, \log(\delta^{-1}))$ samples from $A_{\phi}^{(1)} \cap \widetilde{A}_{\phi}$. Call these samples $z_1, \ldots, z_t$. By \cite[Claim 4.14]{TW11}, with probability at least $1 - \delta$, $\abs{\langle z_1, \ldots, z_t \rangle} \geq (1/2) \abs{A_{\phi}^{(1)} \cap \widetilde{A}_{\phi}}$. Let $\langle z_1, \ldots, z_t \rangle \cap (A_{\phi}^{(1)} \cap \widetilde{A}_{\phi}) = Q$. The earlier bound ensures that when we continue with the rest of the arguments in \cite[Section 4.3]{TW11}, there exists a subset $Q' \subset Q$ of density at least $\exp(\poly(\ep))$ in $G$ such that for $(x, \phi(x)) \in B$ we have that $\phi(x) = Tx + c$. We will let $T_b = Tx + c$ in \texttt{member-A-tilde}. In particular, \texttt{member-A-tilde} returns 1 on $Q' \subset B$ which has density at least $\exp(\poly(\ep)))$ in $\abs{G}$. The guarantee of the Freiman homomorphism conditions follows from Theorem~\ref{thm:step2.2}.
\end{proof}

Notice that unlike Theorem~\ref{thm:step2.2}, where we obtain query access to a set with density at least $\Omega(\qpoly(\epsilon))$, here we have a worse bound of $\exp(\poly(\epsilon))$. This will cause the polynomial we recover at the end to have a worse correlation with $F$ than in the case of higher characteristics.

\begin{claim}
There exists an algorithm \texttt{member-B-prime} that makes $O(\poly(\ep^{-1}, \eta^{-1}, \log(\delta^{-1})))$ queries to \texttt{member-A} and $\phi$ with the following properties: there exists $B' \subset A_2$ such that \texttt{member-B-prime} returns $1$ if $(a,b) \in B' \cap A_1$ with probability at least $1- \delta$ and $0 $ if $(a,b) \not \in B'$ with probability at least $1 - \delta$, where $B'$ contains at least $\poly(\epsilon, \eta)\abs{G}^{32}$ second-order 4-arrangements and the proportion of its second-order 4-arrangements that are respected by $r_{a,b}\mathbf{1}_{\widetilde{A}}$ is at least $1 - \eta$.
\end{claim}

\begin{proof}
We use the same random algorithm as elaborated in Theorem~\ref{thm:step2.1}. 

\cite[Proposition 6.1]{G01} implies $\widetilde{A}$ contains at least $\ep^{29}\abs{G}^3$ additive quadruples. \cite[Corollary 3.9]{GM17} implies that $\phi|_{\widetilde{A}}$ respects at least $\poly(\epsilon)$ second-order 4-arrangements in $\widetilde{A}$.

A property of that random algorithm, which follows from the proof of \cite[Lemma 3.11]{GM17}, is that given a second-order 4-arrangement on whose points \texttt{member-A} returns 1, the random algorithm chooses all its points with probability $2^{\poly(\ep, \eta)}$ if it is not respected by $\phi$, and probability $2^{\poly(\ep, \eta)}(1 + 2^{-31})^{\poly(\ep, \eta)}$ if it is. Let $X$ be the number of second-order 4-arrangements in $\widetilde{A} \cap B'$ respected by $\phi$, and let $Y$ be the number of second-order 4-arrangements in $B'$ not respected by $\phi$. In particular, on expectation, $\E[x - \eta^{-1}Y] \geq \poly(\eta, \epsilon) \abs{G}^{32}$. This means that with probability at least $\poly(\epsilon, \eta)$, we have that $\Pb[X - \eta^{-1}Y] \geq 0$. 

Note that $\phi|_{\widetilde{A}} = r_{a,b}$. Consequently, by sampling $s_1, \ldots, s_k$ and $M_1, \ldots, M_k$ $\poly(\ep^{-1}, \eta^{-1}, \log(\delta^{-1}))$ times in the random algorithm, we obtain with probability $1 - \delta$ a set $B'$ such that the proportion of its second-order 4-arrangements that are respected by $r_{a,b}$ is at least $1 - \eta$. This implies that, by repeating the procedure for the random algorithm Theorem~\ref{thm:step2.1} $\poly(\ep^{-1}, \eta^{-1}, \log(\delta^{-1}))$ times, we obtain \texttt{member-B-prime} with the properties as desired.
\end{proof}

Next, run the arguments in Section~\ref{sec:put-it-together} up until Theorem~\ref{thm:step4.1}. Up till this point we have showed $r_{a,b}$ respects a large fraction of 4-arrangements on $\lozenge \mathbf{1}_{\widetilde{A}}$ and query access to a set which contains many 4-arrangements respected by $r_{a,b}$. Theorem~\ref{thm:step4.1} allows us to find a bilinear Bohr decomposition into high-rank bilinear Bohr sets $B_{v,w,z}$. An application of \cite[Theorem 5.8]{GM17} shows that $r_{a,b}$ respects a large fraction of 4-arrangements on at least one of these bilinear Bohr sets $B_{v,w,z}$. Call these bilinear Bohr sets \emph{wonderful}.

For simplicity of notation, write $\upsilon: G \times G \to \Sigma(\mathcal{A})$ for $\upsilon = \mathbf{1}_{\widetilde{A}} \delta_{r_{a,b}}$. Note that if $\upsilon$ is a $(1 - 4\eta)$-bihomomorphism with respect to $b_{v,w,z}$, then $\phi$ is a $(1 - 4\eta)$-bihomomorphism with respect to $b_{v,w,z}$ as well. Recall that we do not have query access to $r_{a,b}\mathbf{1}_{\widetilde{A}}$ but we only query access to $\phi$. In particular, if the algorithm \texttt{high-rk-bohr-set} from Theorem~\ref{thm:step4.2} is given as input a wonderful $(v,w,z)$, it returns 1 with probability $1 - \delta$. However, some $(v,w,z)$ on which \texttt{high-rk-bohr-set} returns 1 may not be wonderful. Consequently, we will run the subsequent steps of the algorithm on every single bilinear Bohr set on which \texttt{high-rk-bohr-set} returns 1. This does not affect our bounds on the run-time since there are at most $\poly(p^k)$ such bilinear Bohr set.

We proceed with the arguments in Section~\ref{sec:put-it-together} up till Theorem~\ref{thm:step5} in which we obtain a $\widetilde{\psi}$ that is additive in each variable on a large fraction of $B''$. By \cite[Corollary 6.9]{GM17} and the fact that $\upsilon$ is a $(1 - 4\eta)$-bihomomorphism with respect to $b_{v,w,z}$, there exists a subset $\widetilde{B} \subset B''$ with density at least $1 - \poly(\eta)$ such that $\widetilde{\psi}(a,b) = r_{a,b}$ for all $(a,b) \in \widetilde{B}$.

\begin{claim}
Assuming that we are working with a wonderful bilinear Bohr set, in time $\poly(n, \log(\delta^{-1}))$, we can identify a suitable extension of $\widetilde{\psi}$ to a bi-affine map $T$ such that $\sigma(a,b,c) = T^L(a,b).c$ is a CSM, where $T^L$ is the bilinear part of $T$.
\end{claim}

\begin{proof}
\texttt{is-cell-good} identifies a subset of $\widetilde{B}$, which would allow us to extend the domain of $r_{a,b}$ to the entirety of $B''$. The property that $T^L(a,b).c$ is a CSM translates to $T^L(x,x).x = 0$. We can then translate this into linear constraints on $T^L(x_1, \cdot)$. We introduce these constraints when we extend the domain of $T$ to $B'' \cup \{ \{x_1 \} \times G \}$ in Theorem~\ref{thm:step6}, where we arbitrarily made choices for some values of $T(x_1, \cdot)$. We are guaranteed that the resulting linear system has at least one solution because the map $T(a,b) = r_{a,b}$ satisfies the properties in the claim.
\end{proof}

Lastly, perform symmetrization and integration via the arguments laid out in detail in Section~\ref{sec:p=2} on this CSM. At this stage we may certify if the output is a classical polynomial $P$. We are guaranteed that with probability $1 - \delta$, we would obtain a classical cubic polynomial, and for such a polynomial we may proceed with Theorem~\ref{thm:step7} to obtain the desired polynomial $Q$.
\end{proof}

\appendix
\section{Deferred technical proofs} \label{sec:app}

\subsection{Proof of Theorem~\ref{thm:noisy-GL}}
\begin{proof}

The proof is standard. We will implement the divide-and-conquer strategy as in the classic Goldreich-Levin algorithm, with the key observation being that $\F_p^n$ has many subspaces. 

The additional observation here is that even though we are working with a noisy query $f'$, where $f$ has a heavy Fourier coefficient $f'$ does as well. More precisely, write $f' = f + e$. Then we know that $\norm{e}_{\infty} \leq 1$ and also that for at least $1 - \eta$ fraction of $r \in \F_p^n$ we have that $\abs{e(r)} \leq \omega$. Let the set of $r$ satisfying the latter condition be $S$. Let us expand $\widehat{f'} = \widehat{f} + \widehat{e}$ as follows 
\begin{align*}
    \abs{\widehat{f'}(x)} &\geq \abs{\widehat{f}(x)} - \abs{\E_r e(x) \omega^{r \cdot x}} \\
    &\geq \abs{\widehat{f}(x)} - \E_{r \in S} \abs{e(r)} - \E_{r \not \in S} \abs{e(r)} \\
    &\geq \abs{\widehat{f}(x)} - \eta - (1-\eta) \cdot \omega.
\end{align*}
In particular, if $\abs{\widehat{f}(x)} \geq \tau$ then $\abs{\widehat{f'}(x)} \geq \tau - \eta - (1-\eta) \cdot \omega$. Similarly, if $\abs{f(x)} < \tau$ then $\abs{f'(x)} < \tau + \eta + (1- \eta) \cdot \omega$.

The high-level picture is that we iteratively split the coefficients $A$ we are working with at the current stage of the algorithm into $p$ buckets $A = B_1 \sqcup B_2 \sqcup \cdots \sqcup B_p$. Define $f_i^{(A)}(x) = \sum_{r \in B_i}\widehat{f'}(r) \omega^{r \cdot x}$. Recall by Parseval's that we have the following expression 
\[ \norm{f_i^{(A)}}_2^2 = \sum_{r \in B_i} \abs{\widehat{f'}(r)}^2. \]
Because of our eventual choices of $B_i$ as subspaces as well as the fact that the Fourier transform is defined as an expected value, we are able to sample and approximate $\norm{f_i^{(A)}}_2^2$. Combining with our earlier observation, if our approximation of the norm is smaller than $\frac{3}{4} \left(\tau - \eta - (1-\eta) \cdot \omega \right)^2$ then we will be able to conclude with high probability that there does not exist $r \in B_i$ such that $\widehat{f}(r) \geq \tau$. Discard $B_i$ which have small corresponding norms. This process allows us to refine our search and home in on the large Fourier coefficients. Now iterate the algorithm by partitioning up the remaining ``alive'' (i.e. buckets which have not been discarded) $B_i$ until we get down to singleton sets. 

In more detail, choose the buckets as follows. The buckets are indexed by two values $a \in [n]$ and $b \in \F_p^{a}$. For $\vec{b} \in \F_p^{a}$ write $\vec{b}_{-1}$ for the vector obtained by truncating the first element of $b$ and let the $j$th element of $\vec{b}$ be $\vec{b}_j$. Let the standard basis vectors for $\F_p^n$ be $e_1, \cdots, e_n$. The bucket $B_{a,b}$ is recursively defined as 
\[ B_{a,\vec{b}} = \{ (\langle e_{a} \rangle^{\perp} \cap B_{(a-1),\vec{b}_{-1}}) + \vec{b}_1e_{a} \} \]
where if $a = 0$ the bucket corresponds to $\F_p^n$. The initial buckets are $B_{1,0}, \cdots, B_{1,p-1}$. The algorithm always splits a bucket $B_{i,\vec{b}}$ into $B_{i+1,(0,\vec{b})}, \ldots, B_{i+1,(p-1, \vec{b})}$. 

Define $f_{a, \vec{b}}(x) = \sum_{r \in B_{a, \vec{b}}} \widehat{f'}(r) \omega^{r \cdot x}$. Assume for the moment that we are able to estimate $\norm{f_{a, \vec{b}}}_2^2$ to within an additive error of $\pm \frac{1}{4}\left(\tau - \eta - (1-\eta) \cdot \omega \right)^2$. For all the ``alive'' buckets, approximate $\norm{f_{a, \vec{b}}}_2^2$ and discard them if the value is smaller than $\frac{3}{4} \left(\tau - \eta - (1-\eta) \cdot \omega \right)^2$. By our earlier computations:
\begin{itemize}
    \item If there exists $r \in B_{a, \vec{b}}$ such that $\widehat{f}(r) \geq \tau$ then it follows that $\widehat{f'}(r) \geq \tau - \eta - (1-\eta) \cdot \omega$ and so $\norm{f_{a, \vec{b}}}_2^2 \geq (\tau - \eta - (1-\eta) \cdot \omega)^2$. Given our precision, it follows we would not throw away any bucket that contains this $r$. 
    \item If there is $r$ such that $\widehat{f}(r) < \frac{\tau}{2} - \frac{3}{2} \cdot \left( \eta + (1-\eta) \cdot \omega \right)$ then it follows that $\widehat{f'}(r)<  \frac{1}{2} \left( \tau - \eta - (1-\eta) \cdot \omega \right)$. The singleton bucket $B_{n, \vec{b}}$ that contains $r$ has a corresponding $\norm{f_{B_{n, \vec{b}}}}_2^2 < \frac{1}{4} (\tau - \eta - (1-\eta) \cdot \omega)^2$. Given our precision this is strictly less than $\frac{1}{2} (\tau - \eta - (1-\eta) \cdot \omega)^2< \frac{3}{4} (\tau - \eta - (1-\eta) \cdot \omega)^2$ and if we had not already discarded the bucket corresponding to $r$ we would have discarded it when we reduced down to singleton buckets. 
\end{itemize}

Now, we estimate the 2-norm of $f_{a, \vec{b}}(x) = \sum_{r \in B_{a, \vec{b}}} \widehat{f'}(r) \omega^{r \cdot x} = \sum_{r \in \F_p^n}\mathbf{1}_{B_{a, \vec{b}}}(x)  \widehat{f'}(r) \omega^{r \cdot x}$. The first step is to write $f_{a, \vec{b}}$ as a convolution. Consider $u_{a, \vec{b}}(x) = \sum_{r \in B_{a, \vec{b}}}\omega^{r \cdot x}$. We will show that $f_{a, \vec{b}} = f*u_{a, \vec{b}}$. Start by observing that 

\[
    \widehat{u_{a, \vec{b}}}(y) = \E_x \sum_{r \in B_{a, \vec{b}}} \omega^{r \cdot x} \omega^{-y \cdot x}
    = \sum_{r \in B_{a, \vec{b}}} \E_x \omega^{(r-y) \cdot x}
    = \mathbf{1}_{B_{a, \vec{b}}}(y).
\]
Consequently, we have that

\[
    \widehat{g*u_{a, \vec{b}}}(r) = g^{\land}(r) u_{a, \vec{b}}^{\land}(r) 
    = \widehat{g}(r) \mathbf{1}_{B_{a, \vec{b}}}(y).
\]
This immediately implies that $f_{a, \vec{b}} = f'*u_{a, \vec{b}}$.

Because convolution is defined as an expected value, we are in a slightly better shape to estimate $f_{a, \vec{b}}$. To that end, we next describe how to calculate $u_{a, \vec{b}}$. First, make the observation that $B_{a, \vec{b}} = \langle e_1, \ldots, e_a \rangle^{\perp} + v_{\vec{b}}$ where $v_{\vec{b}} = \sum_{i=1}^{a} \vec{b}_i e_i$. For simplicity write $U_a = \langle e_1, \ldots, e_a \rangle$. Observe that $\sum_{r \in U_a} \omega^{x \cdot r} = \mathbf{1}_{U_a^{\perp}}(x) \abs{U_a}$. This in turn implies that 
\[
    u_{a, \vec{b}}(x) = \sum_{r \in U_a} \omega^{x \cdot (r + v_{\vec{b}})} 
    = \omega^{x \cdot v_{\vec{b}}} \sum_{r \in U_a} \omega^{x \cdot r} 
    = \mathbf{1}_{U_a^{\perp}}(x) \abs{U_a} \omega^{x \cdot v_{\vec{b}}}.
\]
Combining all the pieces that we have so far, and recalling that $\abs{U_a} \abs{U_a^{\perp}} = p^n$ we can write 
\begin{align*}
    \norm{f_{a, \vec{b}}}_2^2 &= \E_{x \in \F_p^n} \abs{f'*u_{a, \vec{b}}(x)}^2 \\
    &= \E_{x \in \F_p^n} \abs{\E_{y \in \F_p^n}[f'(x-y) u_{a, \vec{b}}(y)]}^2 \\
    &= \E_{x \in \F_p^n} \abs{p^{-n} \sum_{y \in \F_p^n} f'(x-y) \abs{U_a} \mathbf{1}_{U_a^{\perp}}(y) \omega^{y \cdot v_{\vec{b}}}}^2 \\
    &= \E_{x \in \F_p^n} \abs{\abs{U_a^{\perp}} \sum_{y \in U_a^{\perp}} f'(x-y) \omega^{y \cdot v_{\vec{b}}}}^2 \\
    &= \E_{x \in \F_p^n} \abs{\E_{y \in U_a^{\perp}}[f'(x-y) \omega^{y \cdot v_{\vec{b}}}]}^2.
\end{align*}

In this form, it becomes clear that we are able to sample to approximate $\norm{f_{a, \vec{b}}}_2^2$. Since $\abs{f'(x-y)\omega^{v \cdot v_{\vec{b}}}} \leq 1$, by Lemma~\ref{lem:CH} for fixed $x$ we can estimate $\E_{y \in U_a^{\perp}}[f'(x-y) \omega^{y \cdot v_{\vec{b}}}]$ to within an additive error of $\frac{1}{\sqrt{8}} \cdot (\tau - \eta - (1- \eta) \cdot \omega)$ with confidence $1 - \delta$ via at most $O(\poly(1/\tau, 1/\eta, 1/\omega) \cdot \log(1/\delta))$ samples. Once more this time unfixing $x$ by Lemma~\ref{lem:CH} it follows that we can estimate $\norm{f_{a, \vec{b}}}_2^2$ to within an additive error of $\frac{1}{4} \left(\tau - \eta - (1-\eta) \cdot \omega \right)^2$ with confidence $1- \delta$ using $O(\poly(1/\tau, 1/\eta, 1/\omega) \cdot \log(1/\delta))$ samples. Any ``alive'' bucket has 2-norm at least $\left(\tau - \eta - (1-\eta) \cdot \omega \right)^2$ so by Parseval's theorem there can be at most $\left(\tau - \eta - (1-\eta) \cdot \omega \right)^{-2}$ ``alive'' buckets. Each bucket will be split at most $n$ times, and finding the corresponding 2-norm for each bucket takes at most time $O(\poly(1/\tau, 1/\eta, 1/\omega) \cdot \log(1/\delta))$ as we have already discussed. Combining all the estimates, and taking $\delta = \delta/n$, it follows that the overall running time is $O(n \log n\poly(1/\tau, 1/\eta, 1/\omega,\log(1/\delta)))$ as claimed.
\end{proof}

\subsection{Proof of Theorem~\ref{thm:BSG-test}}

We begin by describing and motivating \texttt{BSG-test}, and defer the technical proof that such an algorithm satisfies the guarantees of Theorem~\ref{thm:BSG-test} to the end of this subsection. We build a (random) bipartite graph $G$ with vertices $A \cup A$ (call one copy $A^{(1)}$ and another $A^{(2)}$) and edge set $E_{\gamma}$ for $\gamma > 0$ defined as
\[ E_{\gamma} := \{ (a_1, a_2): \abs{\{ (a,b) \in A \times A: a + b = a_1 + a_2\}} \geq ((\rho/2) + \gamma) \cdot \abs{A} \}. \]

We begin by showing that $G$ has many edges, where precisely we will show that the edge density of $G$ is at least $\rho/2 - \gamma$.

\begin{claim}
If $E(A,A) \ge \rho \abs{A}^3$ where $E(A,A)$ is the additive energy of $A$, then the density of $(a_1,a_2) \in A \times A$ such that the number of $\{ (a,b) \in A \times A : a + b = a_1 + a_2 \}$ is at least $((\rho/2)+\gamma)\abs{A}$ is at least $\rho/2 - \gamma$.
\end{claim}
\begin{proof}
For each $x \in A+A$, let $r_x$ be the number of $(a,b) \in A \times A$ such that $a+b=x$. Define a set $S$ be the elements of $A+A$ such that $r_x \ge ((\rho/2)+\gamma)\abs{A}$. Then we have
\[
\rho\abs{A}^3 \le E(A,A) = \sum_{x \in S} r_x^2 + \sum_{x \notin S} r_x^2.
\]
Since $r_x \le ((\rho/2)+\gamma)\abs{A}$, we have that $\sum_{x \notin S} r_x^2 \le ((\rho/2)+\gamma)\abs{A} \sum_{x \notin S} r_x \le ((\rho/2)+\gamma)\abs{A}^3$. Therefore,
\[
\sum_{x \in S} r_x^2 \ge \dfrac{\rho-\gamma}{2}\abs{A}^3.
\]
For each $x \in A+A$, $r_x \le \abs{A}$, so
\[
\sum_{x \in S} r_x \ge \dfrac{1}{\abs{A}} \sum_{x \in S} r_x^2 \ge \dfrac{\rho-\gamma}{2}\abs{A}^2.
\]
Hence the density of $(a_1,a_2) \in A \times A$ such that the number of $\{ (a,b) \in A \times A : a + b = a_1 + a_2 \}$ is at least $((\rho/2)+\gamma)\abs{A}$ is at least $\rho/2-\gamma$.
\end{proof}

We first establish a test for if an edge is present in $G$. 

\vspace{\algtopskip}
\noindent \fbox{
    \parbox{\textwidth}{
    \texttt{Edge-test(a,b)}:
    \begin{itemize}
        \item Sample $t$ elements of $A$ say $a_1, \ldots, a_t$. 
        \item Answer 1 if for at least $(\rho/2)t$ indices we have that $a + b - a_i \in A$ and 0 otherwise.
    \end{itemize}
}
}\vspace{\algbotskip}

As a direct consequence of Lemma~\ref{lem:CH} we have the following guarantee for \texttt{Edge-test}. 
\begin{claim}
Given $\delta, \gamma >0$, the output of \texttt{Edge-test(a,b)} with $t = O(\gamma^{-2}\rho^{-2} \cdot \log(\delta^{-1}))$ queries satisfies the following guarantee with probability at least $1 - \delta$: 
\begin{itemize}
    \item If \texttt{Edge-test(a,b)} outputs 1 then $(a,b) \in E_{-\gamma}$.
    \item If \texttt{Edge-test(a,b)} outputs 0 then $(a,b) \not \in E_{\gamma}$
\end{itemize}
\end{claim}

Let $\eta = \rho/2$. For a random element $a \in A^{(2)}$ define the following sets:
\begin{itemize}
    \item $N_{\gamma}(a) := \{ b: (a,b) \in E_{\gamma} \}$, here implicitly $N_{\gamma}(a) \subset A^{(1)}$.
    \item $N_{\gamma}(b) := \{c : (b,c) \in E_\gamma\}$ for each $b \in A^{(1)}$, here $N_\gamma(b) \subset A^{(2)}$.
    \item $M_{\gamma, \eta}(a) := \{ b \in N_{\gamma}(a): \Pb_{c \in A^{(2)}}[c \in N_{\gamma}(b)] \geq \eta \}$.  
    \item $G_{\gamma_1, \gamma_2, \gamma_3, \eta_1, \eta_2, \eta_3, \eta_4}(a) := \{ b \in M_{\gamma_1, \eta_1}(a): \Pb_{c \in M_{\gamma_2, \eta_2}(a)}[\Pb_{d \in A^{(2)}}[d \in N_{\gamma_3}(b) \cap N_{\gamma_3}(c)] \leq \eta_3] \leq \eta_4 \}$. 
\end{itemize}

Tracing through the proof of Balog-Szemer\'edi-Gowers, we have the following. 
\begin{lemma}\label{lem:act-BSG}
Let the graph with edge set $E_{\gamma}$ have density at least $\rho_{\gamma}$ and consider $A' = G_{\gamma, \gamma, \gamma, \rho_{\gamma}/2, \rho_{\gamma}^3/20, \rho_{\gamma}/5}(u)$ for a uniformly random vertex $u \in A^{(2)}$. Then with probability at least $3\rho_{\gamma}/4$ the set $A'$ satisfies both: 
\begin{itemize}
    \item $\abs{A'} \geq \rho_{\gamma}^2\abs{A}/16$, and 
    \item $\abs{A' + A'} \leq (1/\rho_{\gamma})^{O(1)} \abs{A}$. 
\end{itemize}
\end{lemma}
What this suggests is that with positive probability, by passing to $G_{\gamma, \gamma, \gamma, \rho, \rho, \rho, \rho}(u)$ for a randomly selected $u \in A$ we would obtain a set with desired small doubling. This motivates us to give an approximate test to determine if $b \in G_{\gamma, \gamma, \gamma, \rho, \rho, \rho, \rho}(u)$, which by Lemma~\ref{lem:act-BSG} we would expect to have small doubling.

\vspace{\algtopskip}
\noindent \fbox{
    \parbox{\textwidth}{
    \texttt{BSG-test}($a,b,\gamma_1, \gamma_2, \gamma_3, \gamma_4, \gamma_5, \eta_1, \eta_2, \eta_3, \eta_4$):
    \begin{itemize}
        \item If \texttt{Edge-test}($a,b$) = 0, return 0.
        \item Sample $e_1, \cdots, e_m$ from $A$. Compute $T = m^{-1}\sum_{i=1}^{m} \texttt{Edge-test}(b, e_i, \gamma_2)$ and if $T \leq \eta_1$ return 0. 
        \item Sample $a_1, \ldots, a_r$ from $A$. For each $i \in [r]$, only retain those $i$ for which \texttt{Edge-test}($a, a_i, \gamma_3$) returns 1.
        \item Of the remaining samples $a_1', \ldots, a_{s}'$, for each $i \in [s]$ further sample $b_1^{(i)}, \ldots, b_t^{(i)}$ as well as $c_1^{(i)}, \ldots, c_u^{(i)}$ from $A$. 
        \item Compute:
        \begin{itemize}
            \item $X_{ij} = $ \texttt{Edge-test}($a_i, b_j^{(i)}, \gamma_4$)
            \item $Y_{ij} = $ \texttt{Edge-test}($a_i, c_j^{(i)}, \gamma_5$)
            \item $Z_{ij} =$
            \texttt{Edge-test}($b, c_j^{(i)}, \gamma_5$)
        \end{itemize}
        \item Let $B_i = 1$ if $t^{-1}\sum_{j=1}^{t} X_{ij} \geq \eta_2/2 $, and 0 otherwise.
        \item Let $C_i = 1$ if $u^{-1} \sum_{j=1}^{u} Y_{ij}Z_{ij} \leq \eta_3^3/20$ and 0 otherwise.
        \item Answer 1 if $s^{-1}\sum_{i=1}^{s}B_iC_i \leq \eta_4/5$ and 0 otherwise.
    \end{itemize}
}
}\vspace{\algbotskip}

\begin{proof}
First, by recalling $E_{\gamma_1} \subset E_{\gamma_2}$ for $\gamma_2 \leq \gamma_1$ we can check that 
\[A^{(1)}(u) := G_{\gamma_1 + \gamma_1', \gamma_2 - \gamma_2',  \gamma_3 + \gamma_3', \eta_1 + \eta_1', \eta_2 - \eta_2', \eta_3 + \eta_3', \eta_4 - \eta_4'}(a) \subset G_{\gamma_1, \gamma_2, \gamma_3, \eta_1, \eta_2, \eta_3, \eta_4}(u) =: A^{(2)}(u).\]

Choose the parameters $r,t,u$ such that the additive error in all the estimates in \texttt{BSG-test} is at most $\rho^3/3200$ with probability at least $1 - \delta$. Specifically, by Lemma~\ref{lem:CH} we take $m,t,u$ to be $\poly(\rho^{-1}, \log(\delta^{-1}))$. We can also take $r = \poly(\rho^{-1}, \log(\delta^{-1}))$ with suitable hidden constants such that at least $s = \poly(\rho^{-1}, \log(\delta^{-1}))$ samples are retained. This means we can take $\eta_1' = \eta_2' = \eta_3' = \eta_4' = \rho^3/1600$. To choose $\gamma_i$, divide the interval $[-\rho/2, \rho/2]$ into $1600/\rho^2$ intervals of length $\rho^3/1600$ each. Choose a random interval uniformly at random and let $\gamma$ and $\gamma'$ be such that the selected interval is given by $[\gamma - \gamma', \gamma+ \gamma']$. Then set $\gamma = \gamma_i$ and $\gamma' = \gamma_i'$ for $i \in [3]$.

We also take $\eta_1 = \eta_2 = \rho/4$, $\eta_3 = \rho^3/160$ and $\eta_4 = \rho/10$. 

Tracing the proof of Balog-Szemer\'edi-Gowers, we see that if $\abs{A^{(2)}(u)} \geq \rho^{O(1)} n$ then $\abs{A^{(2)}(u)+A^{(2)}(u)} \leq \rho^{-O(1)}(u)$. To that end it suffices to show that with high probability that $\abs{A^{(1)}(u)}$ is sufficiently large. Because $A^{(1)}(u) \subset A^{(2)}(u)$ by our observation at the beginning of the proof, this would imply the desired small doubling of $A^{(2)}(u)$. 

To make the parallel with \cite{TW11} a little more explicit, write 
\[ B_{\gamma_1, \gamma_2, \gamma_3, \eta_1, \eta_2, \eta_3, \eta_4}(a) := \{ b \in M_{\gamma_1, \eta_1}(a): \Pb_{c \in M_{\gamma_2, \eta_2}(a)}[\Pb_{d \in A^{(2)}}[d \in N_{\gamma_3}(b) \cap N_{\gamma_3}(c)] \leq \eta_3] \geq \eta_4 \}, \]
where $B_{\gamma_1, \gamma_2, \gamma_3, \eta_1, \eta_2, \eta_3, \eta_4}(u) = M_{\gamma_1, \eta_1}(u) \backslash G_{\gamma_1, \gamma_2, \gamma_3, \eta_1, \eta_2, \eta_3, \eta_4}(u)$. 

For simplicity of notation, write $M(u) = M_{\gamma_1 + \gamma_1', \eta_1 - \eta_1'}(u)$ and $B(u) = M(u) \backslash A^{(1)}(u)$.

From the proof of Balog-Szemer\'edi-Gowers, we have that $\E_u[M(u)] \geq \rho^2\abs{A}/4$. To that end, once we obtain an upper bound on $\E_u[B(u)]$ we will be able to obtain the desired lower bound on $\abs{A^{(1)}(u)}$. 

To end off, we apply a modified form of \cite[Claim 4.11]{TW11}. A pair $(v,v_1)$ is called \textit{bad} if $\abs{N_{\gamma_1'}(v) \cap N_{\gamma_1'}(v)} \leq 9 \rho^3/1600$. 
\begin{claim}
There exists a choice for the sub-interval $[\gamma_3-\gamma_3', \gamma_1 + \gamma_1']$ of length $\rho^3/1600$ in $[-\rho/2,\rho/2]$ such that 
\[ \E [\#\{\text{bad pairs } (v,v_1): v \in M_{\gamma_1 + \gamma_1', \eta_1 + \eta_1'}(u) \land v_1 \in M_{\gamma_3 - \gamma_3', \eta_1 - \eta_1'}(u)\}] \leq \rho^3/800 \cdot \abs{A}^2 \]
\end{claim}
Before we begin the proof for this claim, using the observation that $E_{\gamma} \subset E_{\gamma'}$ for $\gamma' \leq \gamma$ we can conclude that $M_{\gamma_1, \eta_1} \subset M_{\gamma_2,  \eta_1} \subset M_{\gamma_2,  \eta_2}$ for $\gamma_2 \leq \gamma_1$ and $\eta_2 \leq \eta_1$. 
\begin{proof}
First, decompose
\begin{align*}
    &\E_u[\#\{\text{bad pairs } (v,v_1): v \in M_{\gamma_1 + \gamma_1', \eta_1 + \eta_1'}(u), v_1 \in M_{\gamma_3 - \gamma_3', \eta_1 - \eta_1'}] \\
    &= \E_u[\#\{ \text{bad pairs }(v,v_1): v, v_1 \in M_{\gamma_1 + \gamma_1', \eta_1 + \eta_1'}(u)\}]\\
    &\qquad +\E_u[\#\{ \text{bad pairs }(v,v_1): v \in M_{\gamma_1 + \gamma_1', \eta_1 + \eta_1'}(u) \land v_1 \in M_{\gamma_3 - \gamma_3', \eta_1 - \eta_1'}(u) \backslash M_{\gamma_1 + \gamma_1', \eta_1 + \eta_1'}(u)\} ].
\end{align*}
Of the $\binom{\abs{A}}{2}$ choices for $(v, v_1)$, if we get a bad pair then each $u$ is in $ N_{\gamma_1'}(v) \cap N_{\gamma_1'}(v)$ with probability at most $\rho^3/1600$. Since $M_{\gamma_1', \eta_1 + \eta_1'}(a) \subset N_{\gamma_1'}(a)$, it follows that the first summand is at most $\rho^3/1600 \cdot \abs{A}^2$. 

We bound the second summand by 
\[\abs{A} \cdot \E_u[\abs{M_{\gamma_3 - \gamma_3', \eta_1 - \eta_1'}(u) \backslash M_{\gamma_1 + \gamma_1', \eta_1 + \eta_1'}(u)}] \leq \abs{A} \cdot \left(\E_u[\abs{N_{-\gamma_3'}(u)}] - \frac{\rho}{5}\E_u[\abs{N_{\gamma_1'}(u)}]\right). \]
Let $f(\gamma_3, \gamma_1) = \E_u[\abs{N_{-\gamma_3'}}] - \frac{\rho}{5}\E_u[\abs{N_{\gamma_1'}}]$. When $\gamma_3' = \gamma_1' = \rho/2$ we have that $f(\gamma_3, \gamma_1) \leq \abs{A}$ and when $-\gamma_3' = \gamma_1' = \rho/2$  we have that $f(\gamma_3, \gamma_1) = 0$. Since $f(\gamma_3, \gamma_1)$ is monotonically increasing in the first variable and monotonically decreasing in the second variable, it follows that there must be an interval of length $\rho^3/1600$ say $[-\gamma_3', \gamma_1']$ such that $f(\gamma_3, \gamma_1) \leq \rho^3/1600 \cdot \abs{A}$. It suffices to take $[\gamma_3-\gamma_3', \gamma_1 + \gamma_1']$ to be the endpoints of this interval. 
\end{proof}

Lastly, note that 
\[ \#\{\text{bad pairs } (v,v_1): v \in M_{\gamma_1 + \gamma_1', \eta_1 - \eta_1'}(u), v_1 \in M_{ \gamma_3- \gamma_3', \eta_2 - \eta_2'}(u) \} \geq \abs{B(u)} \cdot \rho\abs{A}/20.\]
so that the above claim rewrites as $\E_u[\abs{B(u)}] \leq (\rho^2/40) \abs{A}$. Since there are $1600/\rho^2$ choices for the interval that we picked, this happens with probability $\rho^2/1600 = \poly(\rho)$. For this choice of parameters, we have that $\E_u[\abs{A^{(1)}(u)}] \geq (\rho^2/4 - \rho^2/40) \cdot \abs{A} = \poly(\rho) \cdot \abs{A}$, as desired.
\end{proof}

\subsection{Proof of Theorem~\ref{thm:integrate-2}}

\begin{proof}
As in the proof of Theorem~\ref{thm:integrate-3}, we will work with nCSM $\sigma$ in terms of its monomials. By considering $\sigma(e_i, e_i, e_j) = \sigma(e_i, e_j, e_j)$ it follows that if the monomial $cx_iy_iz_j$ is present then so must $cx_iy_jz_j$. This in particular means that we can decompose $\sigma$ into the following form: 

\[ \sigma(x,y,z) = \sum_{i,j,k} c_{ijk} \sum_{\mathrm{sym}} x_iy_kz_k + \sum_{i,j} d_{ij} \left(  \sum_{\mathrm{sym}} x_iy_jz_j + \sum_{\mathrm{sym}} x_iy_iz_j \right) + \sum_i f_i x_iy_iz_i.\]

It suffices to note the following:
\begin{enumerate}
    \item [(i)] $D_{h_1h_2h_3} \frac{\abs{x_i}}{8} = \frac{h_{1i}h_{2i}h_{3i}}{2}$. 
    \item [(ii)] $D_{h_1h_2h_3} \frac{\abs{x_i} \abs{x_j}}{4} =  \frac{1}{2} \left( h_{3i}h_{2j}h_{1i} + h_{3i}h_{2i}h_{1j} + h_{3i}h_{2j}h_{1j} + h_{3j}h_{2j}h_{1i} + h_{3j}h_{2i}h_{1j} + h_{3j}h_{2i}h_{1i} \right)$.
    \item [(iii)] $D_{h_1h_2h_3}\frac{\abs{x_i}\abs{x_j}\abs{x_k}}{2} = \frac{1}{2} \sum_{\mathrm{sym}}h_{1i}h_{2j}h_{3k}$. 
\end{enumerate}

Here we recall that we consider $\F_2$ in $\mathbb{T}$ as $\frac{1}{2}\Z/\Z$. Assuming the validity of these observations for the moment, it follows that we can take the non-classical polynomial to be 
\[ P(x,y,z) = \sum_{i,j,k} c_{ijk} \frac{\abs{x_i}\abs{x_j}\abs{x_k}}{2} + \sum_{ij} d_{ij} \frac{\abs{x_i} \abs{x_j}}{4} + \sum_i f_i \frac{\abs{x_i}}{8}. \]

Now, let us establish the validity of these observations. We start with (i). 
\begin{align*}
    D_{h_1h_2h_3} \frac{\abs{x_i}}{8} &=  D_{h_1h_2} \left( \frac{h_{3i}}{8} + \frac{3h_{3i}\abs{x_i}}{4} \right) \\ 
    &= D_{h_1} 3h_{3i} \left( \frac{h_{2i}}{4} + \frac{h_{2i}\abs{x_i}}{2} \right) \\
    &= D_{h_1} \frac{h_{3i}h_{2i} \abs{x_i}}{2} \\
    &= \frac{h_{3i}h_{2i}h_{1i}}{2}. 
\end{align*}

Next, we calculate (ii). Start with
\begin{align*}
    D_{h_3} \frac{\abs{x_i} \abs{y_j}}{4} &= \frac{\abs{x_i + h_{3i}} \abs{x_j + h_{3j}}}{4} - \frac{\abs{x_i} \abs{x_j}}{4} \\
    &= \frac{\abs{x_i + h_{3i}}\abs{x_j + h_{3j}}}{4} - \frac{\abs{x_i}\abs{x_j +h_{3j}}}{4} + \frac{\abs{x_i}\abs{x_j + h_{3j}}}{4} - \frac{\abs{x_i}\abs{x_j}}{4} \\
    &= h_{3i} \left( \underbrace{\frac{\abs{x_i + h_{3j}}}{2}}_{I} + \underbrace{\frac{\abs{x_j + h_{3j}}}{4}}_{II} \right) + h_{3j} \left( \underbrace{\frac{\abs{x_i}\abs{x_j}}{2}}_{III} + \underbrace{\frac{\abs{x_i}}{4}}_{IV} \right).
\end{align*}

That is, we have $D_{h_2h_3} \frac{\abs{x_i} \abs{x_j}}{4} = h_{3i}(D_{h_2}(I) + D_{h_2}(II)) + h_{3j}(D_{h_2}(III) + D_{h_3}(IV))$. We calculate each of these terms in turn.

\begin{align*}
    D_{h_2}(I) &= \frac{\abs{x_j + h_{3j} + h_{2j}} \abs{x_i + h_{2i}}}{2} - \frac{\abs{x_j + h_{3j}} \abs{x_i}}{2} \\ &= \frac{\abs{x_j + h_{3j} + h_{2j}} \abs{x_i + h_{2i}}}{2} - \frac{\abs{x_j + h_{3j}} \abs{x_i + h_{2i}}}{2} + \frac{\abs{x_j + h_{3j}} \abs{x_i + h_{2i}}}{2} - \frac{\abs{x_j + h_{3j}} \abs{x_i}}{2} \\ 
    &= \frac{h_{2j}\abs{x_i + h_{2i}}}{2} + \frac{h_{2i}\abs{x_j + h_{3j}}}{2} 
\end{align*}

We also have the following.

\[ D_{h_2}(II) = \frac{\abs{x_h + h_{3j} + h_{2j}}}{4} - \frac{\abs{x_j + h_{3j}}}{4} = \frac{h_{2j}}{4} + \frac{\abs{x_j} h_{2j}}{2}. \]

Putting together the analogous (III) and (IV), we get that 
\[ D_{h_2h_3} \frac{\abs{x} \abs{y}}{4}  = h_{3i} \left( \frac{h_{2j}\abs{x_i + h_{2i}}}{2} + \frac{h_{2i}\abs{x_j + h_{3j}}}{2} + \frac{h_{2j}}{4} + \frac{\abs{x_j} h_{2j}}{2} \right) + h_{3j} \left( \frac{h_{2j} \abs{x_i + h_{2i}}}{2} + \frac{h_{2i} \abs{x_j}}{2} + \frac{h_{2i}}{4} + \frac{\abs{x_i} \abs{h_{2i}}}{2} \right).\]

Differentiating once more, we see that 
\[D_{h_1h_2h_3} \frac{\abs{x} \abs{y}}{4} = \frac{1}{2} \left( h_{3i}h_{2j}h_{1i} + h_{3i}h_{2i}h_{1j} + h_{3i}h_{2j}h_{1j} + h_{3j}h_{2j}h_{1i} + h_{3j}h_{2i}h_{1j} + h_{3j}h_{2i}h_{1i} \right). \]

Lastly, for (iii), note that we have

\begin{align*}
    D_{h_3} \frac{\abs{x_i} \abs{x_j} \abs{x_k}}{2} &=  \frac{\abs{x_i + h_{1i}} \abs{x_j + h_{1j}} \abs{x_k + h_{1k}}}{2} - \frac{\abs{x_i}\abs{x_j + h_{1j}}\abs{x_k + h_{1k}}}{2}\\
    &\qquad+ \frac{\abs{x_i}\abs{x_j + h_{1j}}\abs{x_k + h_{1k}}}{2} - \frac{\abs{x_i} \abs{x_j} \abs{x_k + h_{1k}}}{2} + \frac{\abs{x_i} \abs{x_j} \abs{x_k + h_{1k}}}{2} - \frac{\abs{x_i}\abs{x_j} \abs{x_k}}{2} \\
    &= \frac{h_{1i} \abs{x_j + h_{1j}} \abs{x_k + h_{1k}}}{2} + \frac{\abs{x_i} h_{1j} \abs{x_k + h_{1k}}}{2} + \frac{h_{1k}\abs{x_i}\abs{x_j}}{2}.
\end{align*}
Repeating the same type of calculation twice more gives $D_{h_1h_2h_3} \frac{\abs{x_i}\abs{x_j}\abs{x_k}}{2} = \frac{1}{2} \sum_{\mathrm{sym}} h_{1i}h_{2j}h_{3k}$.
\end{proof}

\bibliographystyle{alpha}
\bibliography{main}

\end{document}